\documentclass[mnsc,sglanonrev]{informs4-plain2}

\RequirePackage{tgtermes}
\RequirePackage{newtxtext}
\RequirePackage{newtxmath}
\RequirePackage{bm}
\RequirePackage{endnotes}
\usepackage[shortlabels]{enumitem}

\usepackage[colorlinks=true,citecolor=blue]{hyperref}

\usepackage[toc,page]{appendix}

\OneAndAHalfSpacedXII 

\usepackage{algorithm}
\usepackage{algpseudocode}
\usepackage{tikz}

\usepackage{dsfont}
\def\id{\mathds{1}}
\def\laweq{\buildrel \mathrm{d} \over =}

\def\N{\mathbb{N}}
\def\P{\mathbb{P}}
\def\p{\mathbb{P}}
\def\E{\mathbb{E}}
\def\G{\mathcal{G}}
\def\R{\mathbb{R}}

\def\M{\mathcal{M}}

\def\X{\mathcal{X}}

\def\X{\mathcal{X}}

\def\d{\mathrm{d}}
\def\Q{\mathcal{Q}}

\DeclareMathOperator*{\esssup}{ess\text{-}sup}

\newcommand{\ES}{\mathrm{ES}}

\DeclareMathOperator*{\conv}{conv}
\DeclareMathOperator*{\dom}{dom}

\usepackage{natbib}
 \bibpunct[, ]{(}{)}{,}{a}{}{,}%
 \def\bibfont{\small}%
 \def\newblock{\ }%
\EquationsNumberedThrough    

\TheoremsNumberedThrough     
\ECRepeatTheorems  %

\MANUSCRIPTNO{MNSC-0001-2024.00}

\begin{document}


 \RUNAUTHOR{Shen, Van Oosten, Wang}

\RUNTITLE{Partial Law Invariance and  Risk Measures}

\TITLE{Partial Law Invariance and Risk Measures}

\ARTICLEAUTHORS{%
\AUTHOR{Yi Shen}
\AFF{Department of Statistics and Actuarial Science, University of Waterloo, \EMAIL{yi.shen@uwaterloo.ca}}

\AUTHOR{Zachary Van Oosten}
\AFF{Department of Statistics and Actuarial Science, University of Waterloo, \EMAIL{zjvanoos@uwaterloo.ca}}

\AUTHOR{Ruodu Wang}
\AFF{Department of Statistics and Actuarial Science, University of Waterloo, \EMAIL{wang@uwaterloo.ca}}
} 

\ABSTRACT{%
We introduce partial law invariance,  generalizing law invariance and probabilistic sophistication widely used in decision theory, as well as statistical and financial applications. 
Partial law invariance may be interpreted as law invariance restricted to events for which there is no model uncertainty, reflecting practical situations in decision theory and  financial risk management. 
We fully characterize partially law-invariant coherent risk measures via a novel representation formula.
Strong partial law invariance is defined to bridge the gap between the above characterization and the classic representation formula of Kusuoka.  
We propose a few classes of new risk measures, including partially law-invariant versions of the Expected Shortfall and the entropic risk measures, and illustrate their applications in risk assessment under different types of uncertainty. 
We provide a tractable optimization formula for computing a class of partially law-invariant coherent risk measures and give a numerical example. 
}%




\KEYWORDS{Law invariance, probabilistic sophistication, ambiguity, coherent risk measures, source theory, Expected Shortfall} 

\maketitle


\section{Introduction}\label{sec:intro}
 
Risk measures are popular tools in financial regulation and decision making. 
In finance, the usual interpretation of a risk measure is that it assigns a monetary amount to each financial position, representing the capital that should be held alongside the position to make it acceptable. In decision-theoretic terms, a risk measure
is often interpreted as the unique certainty equivalent representing a preference relation over random variables, typically with some financially relevant properties. Risk measures such as the Value-at-Risk (VaR) and the Expected Shortfall (ES, also known as CVaR) are implemented in the financial regulatory framework of \cite{BASEL19}, and they are also useful tools in optimization; see \cite{FS16} and \cite{MFE15} for general treatments of risk measures in risk management and finance, and \cite{SDR21} in the context of optimization.

Law invariance is the fundamental property that makes risk measures, as well as many decision models, applicable in practice. Once a reference probability model is specified, either through statistical estimation or expert opinion, law-invariant risk measures evaluate financial positions that share the same distribution under this probability model equally. This connects the theory of risk measures to statistics, thus allowing for a description of the probabilistic structure of financial positions, arguably a basic step in any quantitative financial analysis.
In decision theory, law invariance is often known as probabilistic sophistication (\cite{MS92}).

Although law invariance is a fundamental concept, it may not be universally desirable in all contexts. For instance, model misspecification, or model uncertainty, occurs when the assumed probabilistic model for various financial positions does not accurately represent the true, unknown, distribution of these positions. This can be problematic as the value of a law-invariant risk measure may differ significantly between the assumed model and the objective truth. For instance, during the 2008 financial crisis, the probability model used for default rates of mortgages was questionable, and the loss distribution for mortgage-backed securities was underestimated, leading to insufficient capital requirements for many companies to remain solvent (e.g., \cite{MFE15}). Although we take a risk measure perspective in this paper, the problem of model uncertainty extends beyond this and can affect any (usually statistical) process that relies on estimating a probabilistic model, followed by computing a quantity based on that model. We will generally refer to uncertainty in the probabilistic model as ``model uncertainty", which is also known as Knightian uncertainty or ambiguity in the literature. There are many avenues of research that account for model uncertainty in risk assessment in different ways.
One area is distributionally robust optimization, in which optimization with risk measures is performed against a worst-case scenario; see e.g., \cite{GOO03}, \cite{NPS08}, and \cite{ZF09}. Another area is to study risk measures formulated from multiple scenarios; see e.g., \cite{DL18}, \cite{DKW19} and \cite{WZ21}.
Model uncertainty is deeply rooted in decision theory, often called ambiguity in that literature. The goal there is to model preferences of decision makers who do not know the precise likelihood of events that govern their outcomes. Ambiguity is the key ingredient in the influential work of Choquet expected utility and maxmin expected utility pioneered by \cite{S89} and \cite{GS89}, followed by, e.g., the smooth ambiguity model of \cite{KMM05} and the variational preferences of \cite{MMR06}.

 Most of the literature, including the works cited above, addresses model uncertainty through three interrelated approaches. The first approach models financial positions as real-valued random variables (measurable functions), starting with a fixed reference probability model and a law-invariant risk measure, such as VaR, ES, or a function that represents expected utility preferences, that captures the risk preferences under full confidence in the reference model (that is, pure risk). Model uncertainty is then incorporated by aggregating the risk measure across a family of alternative probability models, often close to the reference probability in some way. This aggregation is typically the worst-case aggregation, as is common in distributionally robust optimization. The second approach, as in \cite{ADEH99} and \cite{FS16}, also models financial positions as measurable functions, but it defines a risk measure directly on the space of these functions without fixing a reference probability model. This method handles model uncertainty by never supposing a reference probability model from the start, and it does not need a concept of pure risk (or preferences over pure risks). The third approach, common in decision theory, adopts the two-stage framework of \cite{AA63}, by building preferences over state-dependent monetary lotteries. It begins with a law-invariant risk measure, typically one representing some expected utility preferences, but it is applied to the monetary lotteries, instead of the mappings from states to lotteries. 
Values of the risk measure over different states are aggregated in various ways such as in
 \cite{S89}, \cite{GS89}, and \cite{MMR06}.
 In this setting, pure risk refers to constant mappings that map each state to the same lottery. 

 In this paper, we take a novel approach. We model financial positions as measurable functions and begin with a reference probability model but not a law-invariant risk measure. To account for model uncertainty, risk measures are assumed to satisfy law invariance only partially, that is, only with respect to random outcomes depending on events with no model uncertainty.  
 This approach to modeling risk measures can address situations where some events are subject to more uncertainty than others. Traditional approaches, such as distributionally robust optimization, cannot handle this setting, as these models of uncertainty cannot directly incorporate sub-$\sigma$-algebras.   

The following example can motivate this idea.
After estimating a probability model for various loss factors, it is natural that different factors will have different levels of model uncertainty. For example, this difference might result from a statistical inference procedure, where the confidence regions for parameters tied to certain loss factors could be larger or smaller than those for other loss factors, for instance, due to data scarcity. From this point, we can typically divide the loss factors into two categories: the loss factors for which we have no, or negligible, model uncertainty and those that have non-negligible model uncertainty. We will call all the events related to the loss factors in the first category the ``model-certain source." A financial position is said to depend on this source if its outcome is determined by events within this source. The next step is to select a risk measure for decision making. Requiring law invariance would not be able to distinguish factors in or outside the model-certain source, which is not desirable. However, for financial positions depending on the model-certain source, it is reasonable to measure their risk level solely based on distributions given by the probability model, as these positions involve no model uncertainty; thus, a partial form of law invariance is desirable. With this guiding principle in mind, we introduce partial law invariance, based on a suitably chosen source, which is mathematically modeled by a sub-$\sigma$-algebra $\mathcal{G}$ of the whole model space. A risk measure is partially law invariant  (with respect to $\mathcal G$) if it equates the risks of financial positions that (a) both belong to this source (i.e., $\mathcal{G}$-measurable) and (b) share identical distributions. This property reflects the discussion above with $\mathcal G$ being the model-certain source, but it is applicable beyond this interpretation. 

 The property of partial law invariance was inspired by, and closely related to, the recent paper of \cite{BBLW23}, where the authors formalize the idea of sources. In that paper, sources are defined as algebras of events, and the main idea is that uncertainty attitudes may differ between sources. We take a fixed source, corresponding to the source of risk in \cite{BBLW23}, allowing us to obtain mathematically tractable representation results. Our setting is later generalized to incorporate multiple sources and 
 Section \ref{sec:mults} discusses the connection between their setting and ours.
 A similar concept in decision theory involves having a $\sigma$-algebra with no ambiguity regarding the probability, as seen in the expected uncertainty utility model of \cite{GP14}.

We summarize the main contributions of the paper below. As the first contribution, we introduce the concept of partial law invariance as a new and relevant property for risk measures and preference relations in Section \ref{sec:2-def}, with two substantive motivating examples in Sections \ref{sec:21} and  \ref{sec:22}, and a connection to stochastic dominance in Section \ref{sec:31}. As the second contribution, a main theoretical result is presented, Theorem \ref{th:friday} in Section \ref{sec:32}, which is a representation for coherent risk measures (\cite{ADEH99}) satisfying partial law invariance.  Third, generalizing the idea of a single source in Sections \ref{sec:def}--\ref{sec:3}, the case of multiple sources is studied with a representation result Theorem \ref{th:multi} in Section \ref{sec:mults}, where we also explain the differences between the source theory of \cite{BBLW23} and our approach. Fourth, we offer several extensions of the above results. A technical tool, called coherent adjustments, in Section \ref{sec:coh-cha} leads to an alternative representation of coherent risk measures in Theorem \ref{theorem:1}.
This enables us to derive Theorem \ref{theo:gen1}, which extends the well-known Kusuoka's representation (\cite{K01} and \cite{FR05}) under a stronger notion of partial law invariance in Section \ref{sec:strong}.
With no surprise to specialists in risk measures, most results can be smoothly generalized to the case of convex risk measures, and we present some of them in Section \ref{sec:gen}.  {Finally, we discuss some explicit classes of partially law-invariant   risk measures. Formulas to compute the partially law-invariant ES are studied in Section \ref{sec:PartES}. To provide some concrete example, two constructions of partially law-invariant coherent risk measures are discussed in Section \ref{sec:twocon}. An application  with numerical examples is presented in Section \ref{sec:computation}.} Section \ref{sec:con} concludes the paper. We will focus on an atomless probability space in the main part of the paper. The analysis in the case of a finite probability space is presented in the appendices,
along with all proofs of technical results.

\section{Partial law invariance: Definition and motivating examples}\label{sec:def}

\subsection{Risk measures, preference relations, and partial law invariance}

\label{sec:2-def}

We, throughout, fix a probability space $(\Omega,\mathcal{F},\mathbb{P})$, where $\mathbb{P}$ is our reference probability model. We denote by $L^0$ the set of real-valued measurable functions, where functions that agree $\mathbb{P}$-almost surely (a.s.)~are identified, and by $L^\infty$ and $L^1$ the subsets of random variables that are essentially bounded and integrable, respectively.
A generic set $\mathcal{X}\subseteq L^0$ is interpreted as the collection of financial positions of interest, and $X\in \mathcal{X}$ represents the loss of the financial position for various states of the world $\omega\in\Omega$.\footnote{A positive value of $X$ represents a loss, and negative value of $X$ represents a profit; this is the standard loss/profit convention in the risk measure literature.} It will always be assumed that $\mathcal{X}$ contains the constant random variables, identified with elements of $\R$. A risk measure is a mapping $\rho:\X\to \R$, with the main interpretation of the amount of capital needed to make a financial loss $X$ acceptable.

For a decision-theoretic formulation, let $\succsim$ be a total preorder over the set $\mathcal{X}$, representing the preferences of a decision maker over the financial positions. We will also make two regularity assumptions on $\succsim$.
\begin{itemize}\setlength{\itemsep}{1pt}
    \item[ (C1)] For all $X\in\mathcal{X}$, $\{c\in \mathbb{R}:X\succsim c\}$ and $\{c\in \mathbb{R}:c\succsim X\}$ are closed.
    \item[(C2)]  If $a,b\in \mathbb{R}$ and $a>b$, then $b\succ a.$ (Recall that elements of $\X$ represent losses.)
\end{itemize}
These assumptions allow us to guarantee the existence of a unique certainty equivalent, as described in the following standard result.
\begin{proposition}
        \label{prop:cerEq}
        Under $\mathrm{(C1)}$ and $\mathrm{(C2)}$, there exists a unique certainty equivalent for every $X\in \mathcal{X}$, that is, for every $X\in \mathcal{X}$, there exists a unique $c_X\in\mathbb{R}$ such that $X\sim c_X$. 
    \end{proposition}
Proposition \ref{prop:cerEq} implies that we can find a risk measure $\rho:\mathcal{X}\to\mathbb{R}$ to represent $\succsim$ via $X \succsim Y \iff \rho(X) \leq \rho(Y)$ (e.g., the certainty equivalent, but not necessarily). The reason for the reversal is to stay consistent with our convention that risk measures represent riskiness, thus the smaller the better. In what follows, we formulate properties for $\rho:\X \to \R$, unifying risk measures and preference relations.

Now, we define our main concept of partial law invariance. Given a sub-$\sigma$-algebra $\mathcal{G}\subseteq \mathcal{F}$, we define $\mathcal{X}(\mathcal{G})=\{X\in \mathcal{X}:X\text{ is }\mathcal{G}\text{-measurable}\}$. To establish technical results, we will assume that $(\Omega,\mathcal G,\p)$ is atomless, that is, $\{\p(A):A\subseteq B;~A\in \mathcal G\}=[0,\p(B)]$ for all $B\in \mathcal G$, although this assumption is not needed for defining our main concept of partial law invariance. For $X,Y\in\mathcal{X}$, we write $X \overset{\mathrm{d}}{=}_\mathbb{P} Y$ if $X$ and $Y$ have the same distribution under $\mathbb{P}$, that is $\mathbb{P}(X\leq x)=\mathbb{P}(Y\leq x)$ for all $x\in\mathbb{R}$. For $\rho: \X \to \R$ (resp.~for $\succsim$ on $\X$), we consider the following property.
\begin{enumerate}
\item[] $\mathcal G$-law invariance:  $\rho(X)=\rho(Y)$ (resp.~$X\sim Y$) for all $X,Y\in \X(\mathcal{G})$ that satisfy $X\laweq_\p Y$.
\end{enumerate}

When $\mathcal G$ is the largest $\sigma$-algebra $\mathcal F$, the obtained $\mathcal F$-law invariance is the widely studied \emph{law invariance} of  the risk measure literature  (or, \emph{probabilistic sophistication} in decision theory). As far as we know, $\mathcal G$-law invariance with $\mathcal{G}\neq\mathcal{F}$ has not been studied in the literature. It is easy to see that law invariance is stronger than $\mathcal G$-law invariance. When making a contrast, we will refer to $\mathcal G$-law invariance as \emph{partial law invariance}
and to $\mathcal F$-law invariance as \emph{full law invariance}, where the $\sigma$-algebras should be clear from the context. We also say that a set $\mathcal{X}\subseteq L^0$ is $\mathcal{G}$-law invariant if  $X\in \mathcal{X}\cap L^0(\mathcal{G})$ and $X\overset{\d}{=}_\mathbb{P}Y\in L^0(\mathcal{G})$ imply $Y\in \mathcal{X}$.  

To get a better intuition for partial law invariance, let us consider monotonicity of 
preferences and risk measures. For $X,Y\in \mathcal{X}$, we write $X\geq_{1} Y$ if $\mathbb{P}(X\leq x)\leq \mathbb{P}(Y\leq x)$ for all $x\in\mathbb{R}$, which is the usual first-order stochastic dominance (FSD). We can formulate two notions of monotonicity.
\begin{enumerate}
\item[] Monotonicity:  $X\succsim Y$ for all $X,Y\in \mathcal{X}$ with $Y\geq X$ (in the sense of $\p$-a.s.~throughout).
\item[] $\mathcal{G}$-FSD monotonicity: $X\succsim Y$ for all $X,Y\in \mathcal{X}(\mathcal{G})$ with $Y\geq_1 X$.
\end{enumerate}
The following simple result clarifies the connection between $\mathcal G$-FSD monotonicity and $\mathcal G$-law invariance.
\begin{proposition}\label{FirstMot}
    Suppose that the preference relation $\succsim$ is monotonic and $\mathcal{X}$ is $\mathcal{G}$-law invariant. Then $\succsim$ is $\mathcal{G}$-FSD monotonic if and only if $\succsim$ is $\mathcal{G}$-law invariant.
\end{proposition}
Therefore, partial law invariance can be equivalently interpreted as consistency with respect to FSD over $\mathcal G$-measurable random variables. 

The main motivations for partial law invariance, in the absence of full law invariance, are model uncertainty in risk evaluation, illustrated in Section \ref{sec:21}, as well as ambiguity in decision making, illustrated in Section \ref{sec:22}. In the following examples, and for the rest of the paper, we write $\sigma(\cdot)$ for the $\sigma$-algebra generated by either a random variable, a collection of random variables, or a collection of sets in $\mathcal F$.

\subsection{Partial law invariance and model uncertainty}
\label{sec:21}

In this section we elaborate on the role of partial law invariance in the context of model uncertainty. 
 
Suppose a firm's portfolio manager is tasked with making investment decisions involving the financial positions in $\mathcal X$. Among the random variables in $\mathcal X$,    the manager knows that the distribution of a financial position $Z$ is accurate under the probability model $\p$; that is, there is no (or negligible) model uncertainty. For instance, $Z$ may be the losses from a regular property insurance portfolio, for which there are enough data to confidently estimate its distribution. The distribution of the other random variables $X$ under the probability model may be subject to a considerable amount of model uncertainty, such as the losses of a stock without a sufficiently long history of data. The manager needs to decide on a risk measure. In this context, law invariance naturally applies when dealing with positions $g(Z)$ that are functions of $Z$ because the manager is confident about their distributions. However, for random variables $X$ that are not $\sigma(Z)$-measurable, 
if $X\laweq_{\mathbb{P}} g(Z)$, it is unreasonable to insist that  $X$ and $g(Z)$ are equally risky due to the model uncertainty associated with $X$; the same can be said for two general random variables $X$ and $Y$.  In other words, it is natural to require $\sigma(Z)$-law invariance but not (full) law invariance. 

Denote by $\mathcal{M}_1$ the set of probability measures on $(\Omega,\mathcal{F})$ that are absolutely continuous with respect to $\mathbb{P}.$ 
The manager may choose a collection of probability measures $\mathcal Q\subseteq \mathcal{M}_1$
to reflect model uncertainty.
Because $Z$ is the loss whose distribution the manager is certain about, the set $\mathcal Q$ should satisfy that for all $\mu \in \mathcal{Q}$, $F_{\mu ,Z}=F_{\p, Z}$, where $F_{\mu,Z}$ is the cumulative distribution function (cdf) of $Z$ under $\mu$. For example, the manager could take 
\begin{equation}
       \label{eq:ES-sup}
        \mathcal{Q}=\left\{ \mu  \in \mathcal{M}_1\mid F_{\mu , Z}=F_{ \p,Z}\text{ and }d(\mu ,\mathbb{P})<\epsilon\right\},
\end{equation} 
where $\epsilon>0$ and $d:\mathcal{M}_1\times\mathcal{M}_1\to[0,\infty]$ is some measure of divergence, such as a Wasserstein distance or the Kullback–Leibler divergence. For $X\in \X$, the manager needs to evaluate the risk of $X$ under each $\mu\in \mathcal{Q}$, and then aggregate these evaluations in some way. For instance, if the manager uses the mean to quantify risk and the worst-case approach of \cite{W45,W49} (axiomatized by \cite{GS89}) to aggregate the risk evaluations, then the resulting risk measure is $\rho:X\mapsto \sup_{\mu \in \mathcal Q}\E^\mu[X]$. We can see that $\rho$ is  $\sigma(Z)$-law invariant  because $\rho(X)=\E^\p[X]$ for $X$ that is $\sigma(Z)$-measurable.

Using the mean as a risk measure is overly simplistic. A risk measure common in financial practice is the Expected Shortfall (ES also known as CVaR). ES serves as the standard risk measure for market risk in the banking regulatory framework of \cite{BASEL19}; see \cite{RU02} for its optimization and  \cite{WZ21b} for an axiomatization. 
For a given probability measure $\mu\in \mathcal M_1$, the ES at level $\alpha\in [0,1)$ is given by
$$
\ES^\mu_\alpha (X)=\min_{x\in \R} \left(x+ \frac{1}{1-\alpha}\E^\mu[ (X-x)^+]\right),~~~~X\in L^\infty.
$$ 
The function $\ES^\mu_\alpha (X)$ can be defined as the average of the left-quantile function of $X$ under $\mu$ from $\alpha$ to $1$, and also on larger spaces in $L^0$; see e.g., \cite{RU02} and \citet[Chapter 8]{MFE15}. For simplicity, we write $\ES_\alpha=\ES_\alpha^\p$. Following the guidelines of \cite{BASEL19}, for each $\mu\in \mathcal{Q}$, suppose that the decision maker uses $\mathrm{ES}_\alpha^\mu$ at level $\alpha=0.975$ to evaluate the risk of $X\in L^\infty$ under $\mu$. The resulting risk measure is the worst-case ES (see \cite{ZF09}), given by
\begin{equation}\label{eq:worst-ES}
    \overline{\ES}^\mathcal Q_\alpha(X)=\sup_{\mu\in \mathcal{Q}}\mathrm{ES}_\alpha^\mu(X),~~~~X\in L^\infty.
\end{equation}
The risk measure $ \overline{\ES}^\mathcal Q_\alpha$ with $\mathcal Q$ in \eqref{eq:ES-sup} is $\sigma(Z)$-law invariant. Therefore, considering robust risk optimization will naturally lead to partial law invariance.

Another common way to assess model uncertainty is through the use of multiple prior probability measures. Law-invariant risk measures under multiple probability measures are studied by \cite{WZ21} under the name of scenario-based risk measures,
which are formulated with the following property, for a set $\mathcal{Q}\subseteq \mathcal{M}_1$.
\begin{enumerate}
\item[] $\Q$-basedness: $\rho(X)=\rho(Y)$    for all $X,Y\in  L^{\infty}$ satisfying $F_{\mu , X}=F_{ \mu,Y}$ for all $\mu\in \Q$. 
\end{enumerate}
This property means that if two risks are identically distributed under all probability measures of interest described by $\Q$, then they are assessed as equally risky.

It is evident that law invariance is a particular case of both  $\mathcal G$-law invariance, by choosing $\mathcal G=\mathcal F$, and $\Q$-basedness, by choosing $\mathcal Q=\{\p\}$. Moreover, law invariance is stronger in general than  both  $\mathcal G$-law invariance 
and $\Q$-basedness (if $\p\in \Q$). 
What is not so obvious is how $\mathcal{G}$-law invariance and $\Q$-basedness are connected, which we address next. 
For $\mu\in \mathcal{M}_1$, let $D_\mu=\mathrm{d}\mu/\mathrm{d}\mathbb{P}\in L^1$, that is, its Radon–Nikodym derivative with respect to $\p$.

  \begin{proposition}
  \label{prop:connection}
  If
  $\rho:L^\infty\to \R$ is $\Q$-based, then it is $\mathcal G$-law invariant for any $\mathcal G$ independent of $\sigma\left(D_\mu:\mu\in \Q\right)$.
  \end{proposition}

This connection between $\mathcal{Q}$-basedness and $\mathcal{G}$-law invariance reinforces the idea that partial law invariance of a risk measure can serve as an indicator of sources without model uncertainty. In this context, observe that $\mathcal{Q}$ can be viewed as a set of competing probability models. The independence between $D_\mu$ and $\mathcal G$ for all $\mu \in \mathcal{Q}$ implies   $\mu|_\mathcal{G} = \mathbb{P}|_\mathcal{G}$ for each $\mu \in \mathcal{Q}$. This means that all the competing probabilities in $\mathcal{Q}$ agree on $\mathcal{G}$, which can be directly interpreted as an absence of model uncertainty on $\mathcal{G}$, as in our motivating example.

\subsection{Partial law invariance in decision models}
\label{sec:22}

We will explore a connection between matching probabilities (defined below) and partial law invariance.  The idea of matching probability is used to analyze decisions under ambiguity; see e.g.,  \cite{DKW16} and \cite{BHSW18}. 
Suppose an agent has a preference relation $\succsim$ over a set of random variables $\mathcal{X}$ on $(\Omega,\mathcal F,\p)$ that includes the indicator functions. The agent does not have access to $\p$, so events in $\mathcal F$ can be associated with ambiguity to this agent. Some events, such as a coin toss or the simulation of a random number, are not ambiguous to this agent, and we denote by $\mathcal G$ the $\sigma$-algebra generated by these unambiguous events. For a decision scientist to understand the ambiguity attitude of the agent, the common approach is to associate each event $A\in \mathcal F $ with a number $p_A\in [0,1]$ that corresponds to the specified probability of an equally preferred event from $\mathcal G$, that is, $p_A=\mathbb{P}(B)$ where $B\in \mathcal{G}$ and $\id_A\sim \id_B$.
A questionnaire for the agent would look like ``would you prefer getting $\$ 1$ if event $A$ happens, or getting $\$1$ with probability $p$". Assuming some regularity of the preference, 
a unique $p_A\in [0,1]$ exists. This $p_A$ is called the matching probability of $A$. 
For this procedure to be sensible, an implicit assumption is that all events $A,B\in \mathcal{G}$ with $\mathbb{P}(A)=\mathbb{P}(B)$ satisfy $\id_A\sim \id_B$. Note that $p_A$ is a numerical representation of the preference of the agent on $\{\id_A:A\in \mathcal F\}$,
and this numerical representation satisfies 
$\mathcal G$-law invariance, but it should not, in general, satisfy full law invariance, as discussed below.

We will now turn to some classical models in decision theory under ambiguity, such as those of \cite{GS89}, \cite{KMM05} and \cite{MMR06}. These models also naturally give rise to partially law-invariant functionals and preferences. We first describe a general framework that allows these models to be translated into our setting. Let $(\Omega_0,\mathcal{F}_0,\mathbb{P}_0)$ denote the probability space representing the potential future states of the world. Suppose a decision-maker does not possess knowledge of $\mathbb{P}_0$ but has access to a lottery-generating device, modeled by $((0,1],\mathcal{B}((0,1]),\lambda)$, where $\lambda$ denotes the Lebesgue measure. Putting these two sources of uncertainty together, one arrives at the probability space $$(\Omega,\mathcal{F},\mathbb{P})=(\Omega_0 \times (0,1],\mathcal{F}_0 \otimes \mathcal{B}((0,1]),\mathbb{P}_0\times \lambda),$$
where $\mathcal{F}_0  \otimes \mathcal{B}((0,1])$ is the product $\sigma$-algebra. Elements of $L^\infty$ are called the Savage acts, and the decision maker has a preference relation $\succsim$ represented by $\rho: L^\infty\to \mathbb{R}$. Let $\mathcal{G} = \sigma(\pi)$, where $\pi:\Omega\to (0,1]$ is given by the projection onto the second coordinate. We refer to the elements of $L^\infty(\mathcal{G})$ as lotteries. This way of modeling lotteries with the product space structure is used by \cite{SW92,SW97} and \cite{KMM05}. Denote by $$\Delta=\{\mu\in \mathcal{M}_1:\mu(A\times B)=\mu(A\times (0,1])\, \lambda(B) \text{ for all }A\in \mathcal{F}_0\text{ and }B\in \mathcal{B}((0,1])\},$$ which represent the potential probability measures the decision maker needs to consider; note that the second marginal is always $\lambda$ since there is no ambiguity about the lotteries. 
Many examples in decision theory can be formulated via the above setting, and we list three below. In each model, $\ell:\R\to \R$ is a loss function, with the corresponding function $u(x)=-\ell(-x)$ being a von Neumann-Morgenstern utility function.
(a)
  The maxmin expected utility  of \cite{GS89}, given by $$\rho(X)=\max_{\mu\in \mathcal{S}}\int_{\Omega}\ell(X)~\d\mu=-\min_{\mu\in \mathcal{S}}\int_{\Omega}u(-X)~\d\mu,$$ where $\mathcal{S}\subseteq \Delta$;
  (b) The smooth ambiguity model of  \cite{KMM05}, given by $$\rho(X)=\int_{\Delta}\phi\left(\int_{\Omega}\ell(X)~\d\mu\right)Q(\d \mu)=-\int_{\Delta}\widetilde \phi\left(\int_{\Omega}u(-X)~\d\mu\right)Q(\d \mu),$$ where $\phi$ is an increasing transformation, $\widetilde \phi(x)=-\phi(-x)$, and $Q$ is a   probability measure over $\Delta$;
  (c) The variational preferences model of \cite{MMR06}, given by $$\rho(X)=\max_{\mu\in \Delta}\left(\int_{\Omega}\ell(X)~\d\mu-c(\mu)\right)=-\min_{\mu\in \Delta}\left(\int_{\Omega}u(-X)~\d\mu+c(\mu)\right),$$
    where $c:\Delta\to [0,\infty]$ satisfies $c(\mu)=0$ for some  $\mu\in \Delta$.  
Because measures in $\Delta$ have the same second marginal $\lambda$,  all of the above models satisfy $\rho(Z)=   \E^\p [\ell(Z)]$ for each lottery $Z\in L^\infty(\mathcal G)$, 
and hence they are $\mathcal{G}$-law invariant, but clearly not law invariant in general.  

To connect back to matching probabilities, assume that the agent's preference relation $\succsim$ is given by the maxmin expected utility model and $(\Omega_0,\mathcal{F}_0,\mathbb{P}_0)$ is atomless. Standard arguments can justify that
$\succsim$ is law invariant on $\{\id_A:A\in \mathcal F\}$ (i.e., the matching probability $A\mapsto P_A$ is law invariant) if and only if $\mathcal{S}=\{\mathbb{P}_0\times \lambda\}$, i.e., the agent has no ambiguity.

\section{Coherent risk measures and  associated preferences}

\label{sec:3}

We will study partial law invariance in the framework of coherent risk measures and their relationship to preferences. We focus on coherent risk measures for three reasons: First, they are prominent in financial regulation 
and portfolio optimization  (e.g., \cite{MFE15} and \cite{FS16}). Second, they give rise to tractable representation results that provide insights into the structure and implications of partial law invariance. Third, they can be translated into the max-min preferences of \cite{GS89} with suitable interpretation. For the rest of the paper, we will assume that $\mathcal{X}=L^\infty$, as is common in the risk measure literature.

\subsection{Preferences and risk measures}

\label{sec:31}
 A \emph{coherent risk measure} (\cite{ADEH99}) is a mapping $\rho:L^\infty\to\mathbb{R}$ satisfying the following properties.
\begin{itemize}\setlength{\itemsep}{1pt}
    \item[] Cash invariance: $\rho(X+c)=\rho(X)+c$ for all $X\in L^\infty$ and $c\in\mathbb{R}$;
    \item[] Monotonicity: $\rho(X)\leq \rho(Y)$ for all $X,Y\in L^\infty$ with $X\leq Y$;
    \item[] Convexity: $\rho(\lambda X+(1-\lambda) Y)\leq \lambda \rho(X)+(1-\lambda)  \rho(Y)$ for all $X,Y\in L^\infty$ and $\lambda \in [0,1]$;
    \item[] Positive homogeneity: $\rho(\lambda X)=\lambda \rho(X)$ for all $X\in L^\infty$ and $\lambda\geq 0$.
\end{itemize}
Moreover, a \emph{convex risk measure} is a mapping that satisfies the first three properties in the above list, and a \emph{monetary risk measure} is a mapping that satisfies the first two properties in the above list; see \cite{FR02} and \cite{FS02a}. For background on risk measures, see Appendix \ref{app:reviewRisk}.

A common continuity assumption in the risk measure literature is Fatou continuity. For $\rho:L^\infty\to\mathbb{R}$, we say that $\rho$ is \emph{Fatou (continuous)} if for all uniformly bounded sequences $(X_n)_{n\in \N}$ in $L^\infty$ that converge $\p$-a.s.~to some $X\in L^\infty$, it holds that $\rho(X)\leq \liminf_{n\to\infty}\rho(X_n).$
Fatou continuity is essential because a risk measure $\rho:L^\infty\to\mathbb{R}$ is Fatou and coherent if and only if there exists $\mathcal{S}\subseteq \mathcal{M}_1$ such that
\begin{equation}
    \label{eq:intCoh}
    \rho(X)=\sup_{\mu\in \mathcal{S}}\mathbb{E}^\mu[X] ,~~~~X\in L^\infty;
\end{equation} see Proposition \ref{prop:coherent} in Appendix \ref{app:reviewRisk} for a rigorous statement of this result. This representation is significant for two reasons. First, it links coherent risk measures to the maxmin expected utility model of \cite{GS89}. With a given loss function $\ell$, 
the random variables $\ell(X)$ can be interpreted as a subjective loss; symmetrically, for a utility function $u$ and random wealth $X$, $u(X)$ is known as the util of $X$.  Applying coherent risk measures to these subjective losses (or ``dis-utils'') agrees with the maxmin expected utility model of \cite{GS89}. Second, the worst-case expectation is evaluated with respect to a set of probability measures instead of finitely additive probabilities (which is the case without Fatou continuity). This distinction highlights the practical relevance of Fatou coherent risk measures, as statistical procedures yield probability measures rather than additive probabilities.

Preference relations that are represented by coherent risk measures will be called \emph{CRM preference relations}, and they are described by a few properties; see \cite{DK13} and Proposition \ref{proposition:CRM} in Appendix \ref{app:2}. For CRM preference relations, $\mathcal G$-law invariance is connected to second-order stochastic dominance (SSD), which we show below. For $X,Y\in L^\infty$, we write $X\geq_{2} Y$ if $\mathbb{E}[\ell(X)]\geq \mathbb{E}[\ell(Y)]$ for all increasing convex functions $\ell:\mathbb{R}\to\mathbb{R}$. In the literature, the common definition is that for two random variables $Z$ and $W$ representing gains, $Z$ dominates $W$ in SSD if $\E[u(Z)]\ge \E[u(W)]$ for all increasing concave functions $u$. This is equivalent to $-W\ge_{2} -Z$, via the connection $\ell(x)=-u(-x)$, where $-W$ and $-Z$ represent losses. Next, we define SSD consistency. 
\begin{enumerate}
\item[] $\mathcal{G}$-SSD monotonicity: $X\succsim Y$ for all $X,Y\in L^\infty(\mathcal{G})$ with $Y\geq_2 X$.
\end{enumerate}
In general, $\mathcal{G}$-SSD monotonicity implies $\mathcal{G}$-FSD monotonicity because $\ge_{1}$ is stronger than $\ge_2$. These two notions of monotonicity are equivalent for CRM preferences,  proved via standard arguments in the literature.
\begin{proposition}
    \label{proposition:CRM2}
    For a CRM preference relation $\succsim$, the following are equivalent: (i) $\mathcal{G}$-law invariance; (ii) $\mathcal{G}$-FSD monotonicity; (iii) $\mathcal{G}$-SSD monotonicity. 
\end{proposition}

In the rest of the paper, we will focus on risk measures instead of preference relations. Given a monetary risk measure $\rho:L^\infty\to\mathbb{R}$, we call $\mathcal{A}_\rho=\{X\in L^\infty\mid \rho(X)\leq 0\}$ the \emph{acceptance set} of $\rho$. Acceptance sets are a standard way to describe monetary risk measures; for some classical results regarding acceptance sets, see Appendix \ref{app:reviewRisk}. The following proposition relates partial law invariance and the acceptance set in the case of a monetary risk measure.

\begin{proposition}\label{proposition:acceptance}
    A monetary risk measure $\rho$ is $\mathcal G$-law invariant if and only if $\mathcal{A}_\rho$ is $\mathcal G$-law invariant.
\end{proposition}

\subsection{Representation of partially law-invariant coherent risk measures}
\label{sec:32}

Our next task is to understand the structure of all $\mathcal G$-law invariant coherent risk measures. First, we need some notation. We denote by $\mathcal{M}(\mathcal{G})$ (resp.\ $\mathcal{M}_1(\mathcal{G})$) the set of finite signed measures (resp.~probability measures) on $(\Omega,\mathcal{G})$ that are absolutely continuous with respect to $\mathbb{P}|_\mathcal{G}$.  We write  $\mathcal{M}=\mathcal{M}(\mathcal{F})$, and $\mathcal{M}_1=\mathcal{M}_1(\mathcal{F})$ for simplicity and to stay consistent with the notation from earlier. If $\mu\in\mathcal{M}_1(\mathcal{G})$, we will continue to use the notation $D_\mu$ to denote the Radon–Nikodym derivative with respect to $\p|_\mathcal{G}$, as $\mathrm{d}\mu/\mathrm{d}\mathbb{P}|_{\mathcal{G}}\in L^1(\mathcal{G})$. For $\mathcal S\subseteq \mathcal M(\mathcal{G})$ or $\mathcal S\subseteq\mathcal{M}$, we write $\widehat {\mathcal S}=\{D_\mu:\mu\in \mathcal S\}$, 
which is contained in $L^1(\mathcal{G})$ or $L^1$. Furthermore, in both cases, $\mathcal{S}$ is said to be  $\mathcal{G}$-law invariant if the set $\widehat{\mathcal{S}}$ is $\mathcal{G}$-law invariant (defined in Section \ref{sec:2-def}). For notational simplicity, when $\mu\in \mathcal{M}_1(\mathcal{G})$ we define $\E^\mu $ on $L^\infty$ by $\E^\mu[X]=\E[D_\mu X]$.
For $\mathcal S\subseteq \mathcal M$, we write $\mathcal{S}^{\mathcal{G}}=\{\mu|_\mathcal{G}:\mu\in \mathcal{S}\}\subseteq \mathcal{M}(\mathcal{G})$. Similarly, for $\mathcal{B}\subseteq L^1$, we write $\mathcal{B}^\mathcal{G}=\{\mathbb{E}[X|\mathcal{G}] : X\in\mathcal{B}\}\subseteq L^1(\mathcal{G})$. There is no ambiguity when we use the notation $\widehat{\mathcal{R}}^\mathcal{G}$ because
$(\widehat{\mathcal{R}})^\mathcal{G}=\widehat{(\mathcal{R}^\mathcal{G})}$. Given $\mu\in \mathcal{M}_1(\mathcal{G})$, let $\mathcal{E}(\mu)=\{\nu\in \mathcal{M}_1:\nu|_\mathcal{G}=\mu\}$.
 Clearly, $\mathcal E(\mu) $ depends on $\mathcal G$, which will be implicitly fixed in all places. 
 If $\mu\in \mathcal{M}_1$, we will denote $\mathcal{E}(\mu|_\mathcal{G})$ by $\mathcal{E}(\mu)$ for simplicity.

Let $\mathfrak{X}$ and $\mathfrak{Y}$ be two vector spaces. 
Given a dual pairing $\left(\mathfrak{X},\mathfrak{Y},\langle\; ,\;\rangle\right)$
let $w(\mathfrak{X},\mathfrak{Y})$ denote the weak topology on $\mathfrak{X}$. This makes $\left(\mathfrak{X},w(\mathfrak{X},\mathfrak{Y})\right)$   a locally convex (Hausdorff) topological vector space. Given some $\mathcal{R}\subseteq \mathfrak{X}$, let $\conv(\mathcal{R})$ denote the convex hull of $\mathcal{R}$. It is well known that $\left(L^\infty,\mathcal{M},\langle\;,\;\rangle\right)$ is a dual pairing under $\langle X,\mu\rangle =\mathbb{E}^\mu[X]=\mathbb{E}\left[D_\mu X\right].$ It is also well known that $\left(L^\infty ,L^1 ,\langle\;,\;\rangle\right)$ is a dual pairing under $\langle X,Y\rangle =\mathbb{E}[XY].$ Indeed, $w\left(L^\infty,\mathcal{M}\right)$ and $w\left(L^\infty,L^1\right)$ coincide since $\left(\mathcal{M},w\left(\mathcal{M},L^\infty\right)\right)\cong\left(L^1,w\left(L^1,L^\infty\right)\right)$ through the mapping $\mu\mapsto D_\mu$.

For a Fatou coherent risk measure $\rho:L^\infty\to\mathbb{R}$, its \emph{supporting set} is defined as:
$$\mathcal{S}_\rho=\left\{\mu\in \mathcal{M}_1:\mathbb{E}^{\mu}[X]\leq \rho(X)~\text{for all}~X\in L^\infty\right\}.$$
Using the subjective loss interpretation for the random variables, we see that the supporting set is the set of priors in the the maxmin expected utility model of \cite{GS89}.

We are now ready to present a general representation for Fatou $\mathcal G$-law-invariant coherent risk measures reminiscent of representation \eqref{eq:intCoh}.
\begin{theorem}
\label{th:friday}
For a mapping $\rho:L^\infty\to \R$, the following are equivalent.
\begin{enumerate}[(i) ]
    \item The risk measure $\rho$ is Fatou continuous, $\mathcal G$-law invariant, and coherent;
    \item there exists a convex $\mathcal S\subseteq \mathcal M_1 $ such that the $w(\mathcal{M}(\mathcal{G}),L^\infty(\mathcal{G}))$-closure of $\mathcal{S}^\mathcal{G}$ is $\mathcal G$-law invariant and
    \begin{equation}\label{eq:friday}
\rho(X) =\sup_{\mu\in \mathcal S} \E^\mu[X] ,~~~~X\in L^\infty;
\end{equation}
    \item the $w(\mathcal{M}(\mathcal{G}),L^\infty(\mathcal{G}))$-closure  of 
 $(\mathcal{S}_\rho)^\mathcal{G}$ is $\mathcal G$-law invariant.
\end{enumerate}
\end{theorem}

We say that $\mathcal S\subseteq \mathcal M_1(\mathcal G)$ is  almost $\mathcal G$-law invariant
if the $w(\mathcal{M}(\mathcal{G}),L^\infty(\mathcal{G}))$-closure of $\mathcal{S} $ is $\mathcal G$-law invariant as in parts (ii)--(iii) of Theorem \ref{th:friday}.
It is well-known that a Fatou coherent risk measure $\rho$ is law invariant if and only if its supporting set $\mathcal{S}_{\rho}\subseteq \mathcal{M}_1$ is law invariant; see \citet[Theorem 4.59]{FS16}. Theorem \ref{th:friday} is reminiscent of this result, but with some subtle differences; in particular, one needs to take the restriction of elements of $\mathcal S$, as well as a closure of the resulting set. The first operation is clearly needed, and the second one seems essential. For a discussion of this condition and a proof of Theorem \ref{th:friday}, see Section \ref{app:friday}.

\begin{example}
Let $\rho$ be the mapping $X\mapsto \E^\nu[X]$ where $\nu\in \mathcal{E}(\mathbb{P})$.
Clearly, $\rho$ is $\mathcal G$-law invariant, and the only possible set $\mathcal S$ in \eqref{eq:friday} is given by $\{\nu\}$.
We have $\mathcal{S}^\mathcal{G}=\{\p|_{\mathcal G}\}$, which is almost $\mathcal G$-law invariant. 
\end{example}

\begin{example}
    Using the subjective loss interpretation for the random variables, we see the preference relations related to the maxmin expected utility model of \cite{GS89} are $\mathcal{G}$-law invariant if and only if $\mathcal{S}^\mathcal{G}$ is almost $\mathcal{G}$-law invariant, where $\mathcal{S}$ is the set of priors. In the specific example given in Section \ref{sec:22}, it is straightforward to show that $\mathcal{S}^{\mathcal G}=\{\tilde{\lambda}\}$, where $\tilde{\lambda}(\Omega_0\times B)=\lambda(B)$ for all $B\in \mathcal{B}((0,1]).$ Thus $\mathcal{S}^{\mathcal{G}}$ is almost $\mathcal{G}$-law invariant.
\end{example}

Theorem \ref{th:friday} also gives a feasible way to construct $\mathcal G$-law-invariant risk measures.
We can take any law-invariant coherent risk measure $\tilde \rho$ on $L^\infty(\mathcal G)$ (e.g., $\ES_\alpha$),
and for each $\mu\in \mathcal S_{\tilde \rho}\subseteq \M_1(\mathcal G)$  we take any set $\mathcal S(\mu)\subseteq \mathcal{E}(\mu)$.
Setting
\begin{equation}
\label{eq:explain-th-1}
\rho(X)=\sup_{\nu \in \mathcal{S}} \E^\nu [X], ~~X\in L^\infty,\mbox{~~where } \mathcal{S}=\bigcup_{\mu\in \mathcal S_{\tilde\rho} } \mathcal S(\mu) ,
\end{equation}
yields a $\mathcal G$-law-invariant risk measure. Moreover, $\rho$ coincides with $\tilde \rho$ on $L^\infty(\mathcal G)$. For instance, the risk measure $ \overline{\ES}^\mathcal Q_\alpha$ in Section \ref{sec:21} coincides with $\ES_{\alpha}$ on $
L^\infty(\mathcal G)$.

\section{Multiple sources}
\label{sec:mults}

We first clarify the connection between partial law invariance and the source theory of \cite{BBLW23}. In the source theory, the decision maker is assumed to have a specific form of preference relation, and on the source $\mathcal G$ of risk,\footnote{\cite{BBLW23} used algebras instead of $\sigma$-algebras, which is a technical point and does not affect intuition.} the decision maker is law invariant with respect to an objective probability. Our concept of partial law invariance assumes a probability $\p$ on the entire $\sigma$-algebra $\mathcal F$, which is interpreted as objective on $\mathcal G$. The requirement on $\mathcal G$ for the preference relation is the same in both frameworks. 

The two frameworks differ on how random variables outside $\mathcal G$ are treated. In our framework, losses outside $L^\infty(\mathcal G)$ are connected to those in $L^\infty(\mathcal G)$ via financial properties like monotonicity and convexity, but we do not directly specify how their distributions are used in evaluating the risk. In the source theory, multiple sources are specified, and the decision maker is law invariant with respect to a (possibly different) subjective probability on each of them.\footnote{In source theory, the natural property of monotonicity (i.e., $X\succsim Y$ if $X\le Y$ pointwise, representing losses) across different sources (i.e., $X$ and $Y$ may not be measurable to the same source) seems to be difficult to embed.}

Thus far, we have only considered one source (sub-$\sigma$-algebra) on which the risk measure is law invariant. In practical applications, there may be multiple sources of randomness on which the risk measure is law invariant.  Extending our framework to multiple sources, 
we consider a similar idea as in \cite{BBLW23}, to let the decision maker be law invariant on each of the sources, but with respect to the common probability $\p$. This leads to the following formal property. Let $\Sigma(\mathcal{F})$ denote the set of all atomless sub-$\sigma$-algebras of $\mathcal{F}$ and $\mathfrak{G}\subseteq \Sigma(\mathcal{F})$. For a mapping $\rho:L^\infty\to\mathbb{R}$, we say that $\rho$ is \emph{$\mathfrak{G}$-law invariant} if it is $\mathcal{G}$-law invariant for each $\mathcal{G}\in\mathfrak{G}$.  
When $\mathfrak{G}=\{\mathcal{G}\}$,   $\mathfrak{G}$-law invariance is the same as $\mathcal{G}$-law invariance.

A natural corollary of Theorem \ref{th:friday} characterizes all Fatou $\mathfrak{G}$-law invariant coherent risk measures.

\begin{theorem}
\label{th:multi}
Let $\mathfrak G\subseteq \Sigma(\mathcal{F})$.
A mapping 
$\rho:L^\infty\to \R$ is a Fatou $\mathfrak{G}$-law-invariant coherent risk measure if and only if 
\begin{equation*}
\rho(X) =\sup_{\mu\in \mathcal S} \E^\mu[X] ,~~~~X\in L^\infty
\end{equation*}
holds for some convex $\mathcal S\subseteq \mathcal M_1 $ such that
the $\mathcal{S}^\mathcal{G}$ is almost $\mathcal G$-law invariant for each $\mathcal G\in \mathfrak G$.
\end{theorem}
Appendix \ref{app:mult} contains an example of a $\mathfrak{G}$-law-invariant coherent risk measure given by $\rho(X)=\sup_{\mathcal{G}\in \mathfrak{G}}\rho_\mathcal{G}(\mathbb{E}[X|\mathcal{G}])$, where each $\rho_{\mathcal G}$ is law invariant and coherent, along with its representation in the form of Theorem \ref{th:multi}.

There are at least two stronger yet natural ways of specifying partial law invariance across multiple sources, which we discuss below. 
\begin{enumerate}[(a)]
\item 
We may further require that the same risk evaluation applies across different $\mathcal G\in \mathfrak G$, leading to the condition that $\rho(X)=\rho(Y)$ for all $X,Y\in \bigcup_{\mathcal G\in \mathfrak G}L^\infty(\mathcal{G})$ that satisfy $X\laweq Y$. Note that $\bigcup_{\mathcal G\in \mathfrak G}L^\infty(\mathcal{G})$ is not necessarily a space or even a convex set, making the property potentially cumbersome to work with.
 
\item To obtain a proper linear space, we can   consider the sub-$\sigma$-algebra $\sigma_\mathfrak{G}=\sigma\left(\bigcup_{\mathcal G\in \mathfrak{G}} \mathcal G\right)$, and work with $\sigma_\mathfrak{G}$-law invariance.
\end{enumerate}

It is clear that both (a) and (b) are stronger than $\mathfrak G$-law invariance, and moreover, (b) is stronger than (a).
In Example \ref{ex:RRA} below, we illustrate a situation in which (b) is not desirable, whereas (a) may be reasonable. 

\begin{example}\label{ex:RRA} In the context of robust risk aggregation (see e.g., \cite{EPR13, EWW15}, \cite{BJW14} and the references therein),  the marginal distributions of $d$ random losses $X_1,\dots,X_d$ are known,
but their joint distribution is subject to substantial model uncertainty. In this case,
it is natural to assume partial law invariance for $\mathfrak{G}=\{\sigma(X_i): i=1,\dots,d\}$, but not on the larger $\sigma$-algebra $\sigma(X_1,\dots,X_d)$, because the distributions of $X_1$ and $X_2$
 are certain,
 but the distribution of $X_1+X_2$ is not. A concrete model is described below.
    Let $(\Omega,\mathcal{F})=\left(\mathbb{R}^d,\mathcal{B}\left(\mathbb{R}^d\right)\right)$ and $\mathbb{P}$ be the multivariate Gaussian measure with mean $m\in \mathbb{R}^d$ and covariance matrix $\Sigma\in \mathbb{R}^{d\times d}$. Let $\tilde{\Sigma}\in \mathbb{R}^{d\times d}$ be given by $\tilde{\Sigma}_{ij}=\id_{\{i = j\}} \sigma_i\sigma_j$. For $\lambda\in [0,1]$, let $\mu_\lambda\in \mathcal{M}_1$ be the multivariate Gaussian measure with mean $m$ and covariance matrix $\lambda\Sigma+(1-\lambda)\tilde{\Sigma}$. Let $X_1,\dots,X_d$ be the random variables given by the projections and $\mathfrak{G}=\left\{\sigma(X_i):i=1,\dots,d\right\}.$ For $\alpha\in [0,1)$, define the risk measure $$\rho(X)=\sup_{\lambda\in [0,1]}\mathrm{ES}^{\mu_\lambda}_\alpha(X),~~~ X\in L^\infty.$$ For $i\in \{1,...,d\}$, let $\nu_i$ be the Gaussian measure on $\mathbb{R}$ with mean $m_i$ and variance $\sigma_i^2$. For all $\lambda\in [0,1]$ and $B\in \mathcal{B}(\mathbb{R})$, $\mu_\lambda\left(X_i^{-1}(B)\right)=\nu_i(B)$. Hence, $\rho$ is a $\mathfrak{G}$-law invariant coherent risk measure, and further it satisfies   property (a) but not property (b). 
 \end{example}

Our current formulation of $\mathfrak G$-law invariance has an advantage over property (a), that is, the flexibility for modeling the situation where different sources are associated with different levels of model uncertainty (ambiguity), illustrated by Example \ref{ex:HS} below.

\begin{example}\label{ex:HS}
Consider  $\mathfrak G=\{\mathcal G_0, \mathcal G_1,\dots,\mathcal G_n\}$ with their pairwise intersection being generated by constants (an example would be that they are independent).
Our interpretation is that $\mathcal G_0$ has no model uncertainty and each $\mathcal G_i$ is associated with different levels of model uncertainty.
Suppose that the decision maker uses 
 the multiplier preferences of \cite{HS01} on each $L^\infty(\mathcal G_i)$, $i\in \{0,1,\dots,n\}$, represented by
 \begin{align}\label{eq:R1-HS} \rho  (X) =\sup_{\mu\in \mathcal M_1} \left(\E^\mu[\ell(X)]-\frac1{\theta_{i}} R(\mu,\p)\right) =\frac1{\theta_{i}}  \log \E\left[e^{\theta_i \ell(X)}\right],~~ X\in L^\infty(\mathcal G_i),\end{align}
 where $\theta_i \ge  0$, $\ell$ is an increasing loss function and $R(\mu,\p)$ is the relative entropy of $\mu$ with respect to $\p$.
The second equality in \eqref{eq:R1-HS} is  known as the
Donsker--Varadhan variational formula;  
 see e.g.,~\cite{S11} for this formulation of the multiplier preferences via  $\ell(x)=-u(-x)$ as in Section \ref{sec:22}. 
When $\theta_i=0$, $\rho(X)=\E[\ell(X)]$, which corresponds to the limit as $\theta_i\downarrow 0$.  
 The formulation \eqref{eq:R1-HS} is well defined because the intersections of $L^\infty(\mathcal G_i)$ and $L^\infty(\mathcal G_j)$ for $i\ne j$
 only contain constants $c\in \R$, for which $\rho(c)=\ell(c)$.  
 In this model, $\theta_i$ represents the magnitude of model uncertainty in $\mathcal G_i$, with a larger value representing more severe model uncertainty; note that $\rho$ in \eqref{eq:R1-HS} is increasing in $\theta_i$. We may assume $\theta_0=0$ because $\mathcal G_0$ has no model uncertainty (but it is not necessary for the discussions here). 
For $X$ outside $\bigcup_{i=0}^n L^\infty(\mathcal G_i)$, we can take $\rho$ arbitrarily. 
By \eqref{eq:R1-HS}, we can see that  $\rho$ satisfies $\mathfrak G$-law invariance, and it further satisfies property (a) if $\theta_0=\theta_1=\dots=\theta_n$.
An explicit convex risk measure on $L^\infty$ satisfying \eqref{eq:R1-HS} 
is presented in the next example, where the sources are independent. 
\end{example}

In Example \ref{ex:HS}, even if model uncertainty exists in each source $\mathcal G_1,\dots,\mathcal G_n$, the risk measure is law invariant on each subspace as a result of robustification,   discussed by \cite{HS01}.

\begin{example}
\label{ex:R1-HS2}
 The risk measure   $\rho$ in 
 \eqref{eq:R1-HS}   on $L^\infty(\mathcal G_i)$ 
 can be written as  an 
entropic risk measure  (\cite{FS02a}) applied to $\ell(X)$. The entropic risk measure $\mathrm{ER}_\beta: L^\infty\to \R$ with parameter $\beta \ge0 $ is
defined by  $$
\mathrm{ER}_\beta (X) = \frac 1 \beta \log \E\left[ e^{\beta X}\right], \text{~for~$\beta >0$},\text{~~and~~$\mathrm{ER}_0=\E$}.
$$ 
The convex risk measures ER have unique properties among monetary risk measures, characterized by \cite{EMWW21}, and they play a fundamental role in the recent characterization of monotone additive statistics by \cite{MPST24}. Note that \eqref{eq:R1-HS} reads as 
$\rho(X) =\mathrm{ER}_{\theta_i}(\ell(X))$ on $L^\infty (\mathcal G_i)$ for each $i\in\{0,1,\dots,n\}$. 
Suppose that $\mathfrak{G}=\{\mathcal{G}_0,\mathcal G_1,\dots,\mathcal{G}_n\}\subseteq \Sigma(\mathcal{F})$ is a collection of pair-wise independent $\sigma$-algebras. Define  $\widehat\rho:L^\infty\to\mathbb{R}$ by $$\widehat\rho(X)=\max_{i\in \{0,1,\dots,n\}}\mathrm{ER}_{\theta_i}(\mathbb{E}[\ell(X)|\mathcal{G}_i]).$$ 
By Jensen's inequality, $\mathrm{ER}_{\theta_i}(\ell(X))\geq \mathrm{ER}_{\theta_i}(\mathbb{E}[\ell(X)])=\mathbb{E}[\ell(X)]$ for all $X\in L^\infty$ and $i\in \{0, 1,\dots,n\}$. Hence, $\widehat \rho(X)=\mathrm{ER}_{\theta_i}(\ell(X))$ for all $X\in L^\infty(\mathcal{G}_i)$ and thus it satisfies \eqref{eq:R1-HS}.
If $\ell(x)=x$, then $\widehat \rho$ is a convex risk measure because it is the maximum of several convex risk measures of the form $\mathrm{ER}_{\theta_i}(\E[X|\mathcal G_i])$.
\end{example}

\section{Extensions}
\label{sec:4}

In this section, we present three extensions of partially law-invariant coherent risk measures. First, in Section \ref{sec:coh-cha}, we introduce a technical tool called coherent adjustments, which admit interesting mathematical properties. Coherent adjustments allow us to reformulate Theorem \ref{th:friday} and serve as the foundation for the subsequent two extensions. Second, in Section \ref{sec:strong}, we strengthen partial law invariance to a property, which still generalizes law invariance, that allows us to extend the representation of \cite{K01}. Finally, in Section \ref{sec:gen}, we extend some of the results regarding partially law-invariant coherent risk measures to partially law-invariant convex risk measures.

\subsection{Coherent adjustments and an alternative representation}
\label{sec:coh-cha}
Next, we take a different route to introduce the concept of coherent adjustments, a useful technical tool that helps us further understand $\mathcal G$-law-invariant coherent risk measures. In particular, coherent adjustments make precise how the risk measures with supporting sets $\mathcal S(\mu)$ are built in \eqref{eq:explain-th-1}.

\begin{definition}
    \label{def:ca}
    Let $\mu\in \mathcal{M}_{1}(\mathcal{G})$, a \emph{$\mu$-coherent adjustment} is a mapping $\tau:L^\infty\to\mathbb{R}$ satisfying
\begin{itemize}
    \item[] $\mathcal{G}$-invariance: $\tau(X+Y)=\tau(X)$ for all $X\in L^\infty$ and $Y\in L^\infty(\mathcal{G})$;
    \item[] $\mu$-monotonicity: $\mathbb{E}^\mu[X]+\tau(X)\leq \mathbb{E}^\mu[Y]+\tau(Y)$ for all $X,Y\in L^\infty$ with $X\leq Y$; 
    \item[] Convexity: $\tau(\lambda X+(1-\lambda)Y)\leq \lambda\tau(X)+(1-\lambda)\tau(Y)$ for all $X,Y\in L^\infty$ and $\lambda\in [0,1]$;
    \item[] Positive homogeneity: $\tau(\lambda X)=\lambda\tau(X) $ for all $X\in L^\infty$ and $\lambda\geq 0$.
    \end{itemize}
    Moreover, a $\mu$-\emph{convex adjustment} is a mapping that satisfies the first three properties in the above list, or $\tau=-\infty$ on $L^\infty$. Given $\mu\in\mathcal{M}_1(\mathcal{G})$ we denote the set of all Fatou $\mu$-coherent adjustments by $\mathrm{CA}(\mu)$.
    Let $\mathcal{R}\subseteq \mathcal{M}_1(\mathcal{G})$, we say that  $(\tau_\mu)_{\mu\in \mathcal{R}}$ is a \emph{CA-assignment} if $\tau_\mu\in \mathrm{CA}(\mu)$ for all $\mu\in \mathcal{R}$. 
\end{definition}

Let $\mu\in \mathcal{M}_1(\mathcal{G})$. From the above properties, it is straightforward that the supremum of $\mu$-coherent adjustments is again a $\mu$-coherent adjustment.
A simple example of a Fatou $\mu$-coherent adjustment is  $\tau= \E^\nu -\E^\mu$ for any $\nu\in \mathcal{E}(\mu)$, which is easy to check.
In Section \ref{app:42}, Proposition \ref{prop:FatouCont}, shows that, indeed, the supremum of the form  $\tau= \E^\nu -\E^\mu$  contains all possible Fatou $\mu$-coherent adjustments. 

We also note that $\mathcal{G}$-invariance of $\tau$ is equivalent to $\tau(X)=0$ for all $X\in L^\infty(\mathcal G)$ in the presence of subadditivity. 

\begin{remark}
The properties of $\mu$-coherent adjustment $\tau$ appear similar to a generalized deviation measure $D$ of \cite{RUZ06}, which is defined by four properties: location invariance (i.e., $D(X+c)=D(X)$ for $c\in \R$), subadditivity, positive homogeneity and $D(X)\ge 0$ for all $X$ with strict $D(X)>0$ when $X$ is not constant.
Deviation measures naturally appear in a decomposition of insurance premiums; see \cite{NRS21}. Certainly, location invariance is not satisfied by any monetary risk measure because it conflicts cash invariance.
By choosing $\mathcal G=\{\varnothing, \Omega\}$ (in which case $\mathcal{M}_1(\mathcal{G})=\{\mu\}$ and  $D_\mu=1=\E[D_\nu|\mathcal G]$ for all $\nu\in \mathcal M_1$), $\mathcal G$-invariance is the same as location invariance; in this case, $\tau$ and $D$ share three properties. However, there is an important difference between $\tau$ and $D$, as 
$\tau$ may take negative values. For instance, $\tau= \E^\nu -\E^\mu$ defines a $\mu$-coherent adjustment, but not a generalized deviation measure. 
\end{remark}

The name ``$\mu$-coherent adjustment" is explained by Proposition \ref{prop:motivation}.

\begin{proposition}\label{prop:motivation}
    Let $\rho:L^\infty\to\mathbb{R}$ and $\mu\in \mathcal{M}_1(\mathcal{G})$. The following are equivalent.
    \begin{enumerate}[(i) ]
        \item The mapping $\rho$ is a Fatou coherent risk measure with $\rho|_{L^\infty(\mathcal{G})}=\mathbb{E}^\mu|_{L^\infty(\mathcal{G})}$;
        \item there exists $\tau\in \mathrm{CA}(\mu)$ such that $   \rho =\mathbb{E}^\mu +\tau$.
    \end{enumerate}
\end{proposition}

If we assume additivity on $L^\infty(\mathcal G)$ 
and $\mathcal G$-law invariance
for a Fatou coherent risk measure $\rho$,
 then $\rho$ has a form $\E^\p+\tau$  for $\tau \in \mathrm{CA}(\p|_{\mathcal G})$ that is similar to  Proposition \ref{prop:motivation} (ii);  see Section \ref{sec:R1-additivity}.

For some $\mathcal R\subseteq\mathcal M_1(\mathcal G)$
and CA-assignment $(\tau_\mu)_{\mu\in \mathcal{R}}$, 
by Proposition \ref{prop:motivation}, the following risk measure  \begin{equation}
    \rho (X)= \sup_{\mu\in\mathcal{R}}\left(\mathbb{E}^{\mu} [X] +\tau_\mu (X) \right),~~~X\in L^\infty,\label{eq:31}
\end{equation} 
is coherent, as it is the supremum of coherent risk measures. Theorem \ref{theorem:1} below shows that this mapping is Fatou continuous and that all Fatou coherent risk measures have this representation.

To explain why $\mu$-coherent adjustments help us extend a coherent risk measure on $L^\infty(\mathcal G)$ to $L^\infty$, let $\tilde{\rho}:L^\infty(\mathcal{G})\to\mathbb{R}$ be a law-invariant coherent risk measure. For a convex $\mathcal{R}\subseteq \mathcal{M}_1(\mathcal{G})$ that is $w(\mathcal{M}(\mathcal{G}),L^\infty(\mathcal{G}) )$-dense in  $\mathcal{S}_{\tilde{\rho}}$ and a CA-assignment $(\tau_\mu)_{\mu\in \mathcal{R}}$, let $\rho$ be given by \eqref{eq:31}. It is easy to see that $\rho|_{L^\infty(\mathcal{G})}=\tilde{\rho}$, meaning that it is $\mathcal{G}$-law invariant.

The above ideas lead to the following  alternative representation of partially law-invariant and Fatou coherent risk measures.  

\begin{theorem}
    \label{theorem:1}
    The mapping $\rho:L^\infty\to\mathbb{R}$ is a Fatou coherent risk measure if and only if there exists a convex set $\mathcal{R}\subseteq \mathcal{M}_1(\mathcal{G})$ and a CA-assignment $(\tau_\mu)_{\mu\in \mathcal{R}}$ such that
    \begin{equation}
        \rho(X)=\sup_{\mu\in\mathcal{R}}\left(\mathbb{E}^\mu[X]+\tau_\mu(X)\right),~~~X\in L^\infty.\label{eq:32}
    \end{equation}
    Moreover, in \eqref{eq:32}, $\rho$ is $\mathcal{G}$-law invariant if and only if $\mathcal{R}$ is almost $\mathcal{G}$-law invariant.
\end{theorem}

\begin{remark}\label{re:th:1-2}
    There is a straightforward way to switch between representations (\ref{eq:friday}) and (\ref{eq:32}) for a Fatou and $\mathcal{G}$-law-invariant coherent risk measure. Given a convex set $\mathcal{S}\subseteq \mathcal{M}_1(\mathcal{G})$ that represents $\rho$ in Theorem \ref{th:friday}, define $\mathcal{R}=\mathcal{S}^\mathcal{G}$, which is almost $\mathcal{G}$-law invariant. For each $\mu\in \mathcal{R}$, define $\mathcal{S}(\mu)=\mathcal{E}(\mu)\cap\mathcal{S}$ and $\tau_\mu(X)=\sup_{\nu\in \mathcal{S}(\mu)}(\mathbb{E}^\nu[X]-\mathbb{E}^\mu[X])$. These can be used in representation (\ref{eq:32}). Conversely, let an almost $\mathcal{G}$-law invariant convex set $\mathcal{R}\subseteq \mathcal{M}_1(\mathcal{G})$ and a CA-assignment $(\tau_\mu)_{\mu\in \mathcal{R}}$ representing $\rho$ in Theorem \ref{theorem:1} be given. For $\mu\in \mathcal{R}$, let $\mathcal{S}_{\tau_\mu}$ be the $w(\mathcal{M},L^\infty)$-closed and convex set in Proposition \ref{prop:FatouCont} (see Section \ref{app:42}) that represents $\tau_\mu$. It is easy to show that $\mathcal{S}=\bigcup_{\mu\in \mathcal{R}}\mathcal{S}_{\tau_\mu}$ is convex and $S^\mathcal{G}=\mathcal{R}$. We can use $\mathcal{S}$ in representation (\ref{eq:friday}).
    The above connection clarifies the construction in \eqref{eq:explain-th-1}.
\end{remark}

Uniqueness of representation (\ref{eq:32}) under a stronger notion of continuity  is addressed in Section \ref{app:uniq}. Representation (\ref{eq:32}) has little resemblance to the classical representation result of \cite{K01} for law-invariant coherent risk measures. The following section addresses this gap by further imposing a necessary and sufficient condition for a similar representation along that direction. 

\subsection{Strong partial law invariance  and Kusuoka-type representation
}\label{sec:strong}

Next, we introduce a new property, which is stronger than partial law invariance and is necessary and sufficient to reduce \eqref{eq:32} to a Kusuoka-type representation; see Proposition \ref{prop:lawinv} in Appendix \ref{app:reviewRisk}. 

Define $\ker(\mathcal{G})=\{X\in L^\infty:\mathbb{E}[X|\mathcal{G}]=0\}.$ We say that $\rho:L^\infty\to \mathbb{R}$ is \emph{strongly $\mathcal{G}$-law invariant} if for all $Z\in \ker(\mathcal{G})$ and $X,Y\in L^\infty(\mathcal{G})$ with $X\overset{\mathrm{d}}{=}_\mathbb{P}Y$ we have $$\rho(Z+X)=\rho(Z+Y).$$

Taking $Z=0\in \ker(\mathcal{G})$, we see that strong $\mathcal{G}$-law invariance implies $\mathcal{G}$-law invariance, hence the name. Since $\ker(\mathcal{F})=\{0\}$, the risk measure $\rho$ is law invariant if and only if $\rho$ is strongly $\mathcal{F}$-law invariant. Therefore, strong $\mathcal{G}$-law invariance and $\mathcal G$-law invariance are both a generalization of law invariance. The following result gives an equivalent condition to strong $\mathcal{G}$-law invariance.

\begin{lemma}\label{lemma:wednesday}
    The mapping $\rho$ is strongly $\mathcal{G}$-law invariant if and only if for all $X\in L^\infty$ and $Y\in L^\infty(\mathcal{G})$ with $Y\overset{\mathrm{d}}{=}_\mathbb{P}\mathbb{E}[X|\mathcal{G}]$, we have $$\rho(X)=\rho(X-\mathbb{E}[X|\mathcal{G}]+Y).$$ 
\end{lemma}

An example of strongly $\mathcal G$-law invariant risk measure $\rho$ is 
$\rho(X)=\tilde \rho(\E[X|\mathcal G])$, where $\tilde \rho: L^\infty(\mathcal G) \to \R$ is law invariant.
Examples of functionals satisfying $\mathcal G$-law invariance but not strong $\mathcal G$-law invariance will be provided in Section \ref{sec:ku-ex}. 

Before we give the main result of the section, we need some technical details. Given $X\in L^1$, $\mu\in \mathcal{M}_1$ and $\alpha\in [0,1)$, define $F_{\mu,X}^{-1}(\alpha)=\inf\{x\in \mathbb{R}\mid F_{\mu,X}(x)\geq \alpha\}$ and $q_X(\alpha)=F_{\mathbb{P},X}^{-1}(\alpha)$. Given $\mu\in\mathcal{M}_1(\mathcal{G})$, we can define $Q\in \mathcal{M}_B$ (see \citet[Theorem 4.62]{FS16}) by
\begin{equation}
    Q([0,t))=\int_{[0,t)} (1-s)~\d q_{D_\mu}(s),\text{~~for~~}t\in [0,1).\label{eq:Q-def}
\end{equation}
For $Q\in\mathcal{M}_B$, define 
$\mathcal M_Q=\{ \mu\in \M_1(\mathcal G): \eqref{eq:Q-def} \mbox{~holds}\},$ which is non-empty because $(\Omega,\mathcal{G},\mathbb{P})$ is atomless, and   
\begin{equation}\label{eq:sup}
\sup_{\mu\in \mathcal M_Q}\mathbb{E}^{\mu}[X]=\int_{[0,1)}\mathrm{ES}_\alpha(\mathbb{E}[X|\mathcal{G}])~ Q (\d\alpha);
\end{equation}
see \citet[Theorem 4.62]{FS16}. Let $\mathrm{CA}(Q)=\bigcap_{\mu\in \mathcal M_Q}\mathrm{CA}(\mu).$ This set is non-empty since it contains the zero map. We present the following result, generalizing the representation of \cite{K01}.

\begin{theorem}\label{theo:gen1}
    Let $\rho:L^\infty\to\mathbb{R}$, the following are equivalent.
    \begin{enumerate}[(i)]
        \item The mapping $\rho$ is a Fatou and strongly $\mathcal{G}$-law-invariant coherent risk measure;
        \item there exist $\mathcal{D}\subseteq \mathcal{M}_B$ and $\tau_Q\in \mathrm{CA}(Q)$ for each $Q\in \mathcal{D}$ such that
    \begin{equation}\label{eq:gen1}
        \rho(X)=\sup_{Q\in \mathcal{D}}\left(\int_{[0,1)}\mathrm{ES}_\alpha\left(\mathbb{E}[X|\mathcal{G}]\right)~Q(\d\alpha)+\tau_Q(X)\right).
    \end{equation}
    \end{enumerate}
\end{theorem}
In the setting $\mathcal G= \mathcal F$, the representation \eqref{eq:gen1} reduces to the Kusuoka representation, 
$$\rho(X)=\sup_{Q\in\mathcal{D}}\int_{[0,1)}\mathrm{ES}_\alpha(X)~Q(\d\alpha),\hspace{3pt}X\in L^\infty, $$
because in this case 
$\E[X|\mathcal G]=X$,  $\tau_Q(X)=0$, and strong $\mathcal G$-law invariance is law invariance.  
A simple example is to choose $\tau_Q=0$ for all $Q\in \mathcal D$, which yields the risk measure $\rho:X\mapsto \tilde \rho(\E[X|\mathcal G])$ studied in Proposition \ref{prop:example1} below. 

\subsection{Convex risk measures}
\label{sec:gen}

This section discusses convex risk measures and their relation to the variational preferences model of \cite{MMR06}. We provide representations for partially law-invariant convex risk measures akin to Theorem \ref{theorem:1}. A risk measure $\rho:L^\infty\to\mathbb{R}$  is Fatou and convex if and only if there exists $\alpha:\mathcal{M}_1\to (-\infty,\infty]$ such that
$$\rho(X)=\sup_{\mu\in \mathcal{M}_1}(\mathbb{E}^\mu[X]-\alpha(\mu)),~~~X\in L^\infty;$$
see \cite{FS16}. Using the subjective loss interpretation for the random variables, we see that applying convex risk measures to subjective losses agrees with the variational preferences model of \cite{MMR06} in the case that $\inf_{\mu\in \mathcal{M}_1}\alpha(\mu)=0$.

Similar to CA-assignments, 
we say that $(\tau_\mu)_{\mu\in \mathcal{M}_1(\mathcal{G})}$ is a \emph{CoA-assignment} if for all $\mu\in \mathcal{M}_1(\mathcal{G})$, $\tau_\mu$ is a Fatou $\mu$-convex adjustment and the mapping $\mu\mapsto \tau_\mu(0)$ is concave.
The following result is the convex version of Theorem \ref{theorem:1}.
\begin{theorem}\label{prop:convex}
The mapping $\rho:L^\infty\to\mathbb{R}$ is a Fatou convex risk measure if and only if there exists a CoA-assignment $(\tau_\mu)_{\mu\in \mathcal{M}_1(\mathcal{G})}$ such that
    \begin{align*}
        \rho(X)=\sup_{\mu\in \mathcal{M}_1(\mathcal{G})}\left(\mathbb{E}^\mu[ X]+\tau_\mu(X)\right)~~~X\in L^\infty.
    \end{align*}
Moreover, $\rho$ is $\mathcal{G}$-law invariant if and only if for all $c>-\rho(0)$ the set $\left\{\mu\in \mathcal{M}_1(\mathcal{G})\mid -\tau_\mu(0)\leq c\right\}$ is almost $\mathcal{G}$-law invariant.
\end{theorem}

The risk measures ER in Example \ref{ex:R1-HS2} are the most common examples of convex risk measures.
We next discuss a few $\mathcal G$-law-invariant versions of ER, which shed some light on different ways to construct partially law-invariant convex risk measures. In the examples below, partial law invariance is not necessarily associated with the interpretation of  model uncertainty.

We write the conditional version of ER (this is not a risk measure in our definition) as 
$$\mathrm{ER}_\beta (X|\mathcal G) =\frac 1 \beta \log  \E\left[ e^{\beta X}|\mathcal G\right],~~~X\in L^\infty.$$
There are many ways to generalize the family of ER to satisfy $\mathcal G$-partial law-invariance, and we list three simple choices. 
The first is to use the mean of the conditional ER, that is, 
$$
{\mathrm{ER}}^{\mathcal G}_\beta (X) = \E[\mathrm{ER}_\beta (X|\mathcal G)]= \E\left[\frac 1 \beta \log  \E\left[ e^{\beta X}|\mathcal G\right]  \right] 
 ,~~~X\in L^\infty. $$
The second is to compute the ER of the conditional mean, that is, 
$$\widetilde{\mathrm{ER}}^{\mathcal G}_\beta (X) =\mathrm{ER}_\beta (\E[X|\mathcal G]) = \frac 1 \beta \log  \E\left[ e^{\beta \E[X |\mathcal G]}\right]  ,~~~X\in L^\infty.
$$
Finally, we can also compute ER of the conditional ER, but this gives back the usual ER, that is
$$ \mathrm{ER}_\beta (\mathrm{ER}_\beta (X|\mathcal G)) = \frac 1 \beta \log  \E\left [ \exp \left( \beta \frac 1 \beta \log  \E\left[ e^{\beta X}|\mathcal G\right]  \right)   \right]   =\mathrm{ER}_\beta (X),~~~X\in L^\infty.
$$ 
 (This is time-consistency of ER as a dynamic risk measure;  see \cite{KS09}.)
The last version is fully law invariant and not interesting to us. The first two versions are $\mathcal G$-law invariant and have intriguing and distinct properties. For instance, if $X$ is $\mathcal G$-measurable, then 
$$
{\mathrm{ER}}^{\mathcal G}_\beta (X) = \E[X]  \mbox{~~~and~~~} \widetilde{\mathrm{ER}}^{\mathcal G}_\beta  (X) = \mathrm{ER}_\beta(X);
$$
if $Y$ is independent of $\mathcal G$, then 
$$
{\mathrm{ER}}^{\mathcal G}_\beta (Y) = \mathrm{ER}_\beta(Y)  \mbox{~~~and~~~} \widetilde{\mathrm{ER}}^{\mathcal G}_\beta  (Y) = \E[Y].
$$
The above formulas show that ${\mathrm{ER}}^{\mathcal G}_\beta$ and $\widetilde{\mathrm{ER}}^{\mathcal G}_\beta $ are not fully law invariant, and none of them dominates the other.
Moreover,  if     $X$ and $Y$ above are identically distributed and non-degenerate, then 
${\mathrm{ER}}^{\mathcal G}_\beta (X)  < {\mathrm{ER}}^{\mathcal G}_\beta (Y)$
and  
$\widetilde {\mathrm{ER}}^{\mathcal G}_\beta (X)  > \widetilde {\mathrm{ER}}^{\mathcal G}_\beta (Y)$.
In other words, 
${\mathrm{ER}}^{\mathcal G}_\beta $ penalizes risks independent of $\mathcal G$, 
whereas $ \widetilde 
 {\mathrm{ER}}^{\mathcal G}_\beta $ penalizes risks measurable to $\mathcal G$.
This 
shows different emphases of the two risk measures when assessing risks from different sources; note that here we do not interpret events outside $\mathcal G$ as events with ambiguity; indeed, all computations are carried with $\p$ that is known and fixed in this setting. 

\section{Explicit classes and applications}
\label{sec:instant}

In this section, we propose a few explicit constructions of partially law-invariant risk measures and illustrate with numerical examples. 
Recall our notation
 $\mathcal{E}(\p)=\{\nu\in \mathcal{M}_1:\nu|_\mathcal{G}=\p\}$, which will be used frequently below.

\subsection{Partially law-invariant ES}
\label{sec:PartES}

This section gives tractable formulas for an important class of partially law-invariant risk measures, the $\mathcal G$-law invariant ES introduced in Section $\ref{sec:21}$, given by
$$
 \overline{\ES}^\mathcal Q_\alpha(X) =\sup_{\mu\in \mathcal{Q}}\mathrm{ES}^{\mu}_\alpha(X),~~~~X\in L^\infty,
 $$
 where $\mathcal Q\subseteq \mathcal E(\p)$ and satisfies some regularity (to be made rigorous soon).  Afterward, we numerically implement the   results of this section and see how decisions in a simple investment vary as we change the parameters in the risk measure.
 
 We focus on ES here for two reasons. 
 First, ES is the standard regulatory risk measure in the banking sector (\cite{BASEL19}), and hence the most relevant for financial applications. Second, ES has special advantages 
in optimization (\cite{RU02}), and 
 the partially law-invariant version of ES admits tractable formulas (Proposition \ref{propsition:computational} below).  
 
To prepare for the formula of $\sup_{\mu \in \mathcal{Q}}\mathrm{ES}^{\mu}$,  we first define two technical terms. 
\begin{definition}
    Let $\mathcal{X}\subseteq L^\infty$, we call $Y\in L^\infty$ the \emph{essential supremum} of $\mathcal{X}$ if
    \begin{enumerate}[(i) ]
        \item For all $X\in \mathcal{X}$, $X\leq Y$;
        \item if $Z\in L^\infty$ and $X\leq Z$ for all $X\in \mathcal{X}$, then $Y\leq Z$.
    \end{enumerate}
    We will denote $Y$ by $\esssup\mathcal{X}$, as the essential supremum (if it exists) is $\mathbb{P}$-a.s.~unique.
\end{definition}

\begin{definition}
    We say that $\mathcal{X}\subseteq L^1$ is \emph{$\mathcal{G}$-convex} if for all $X,Y\in \mathcal{X}$ and $\lambda\in L^\infty(\mathcal{G})$ with $0\leq \lambda\leq 1$, we have
    $\lambda X+(1-\lambda)Y\in \mathcal{X}.$ We say that $\mathcal{R}\subseteq \mathcal{M}$ is \emph{$\mathcal{G}$-convex} if $\widehat{\mathcal R}$ is $\mathcal{G}$-convex.
\end{definition}

Given certain conditions on $\mathcal{Q} \subseteq \mathcal{E}(\mathbb{P})$, the $\mathcal{G}$-law invariant risk measure $\sup_{\mu \in \mathcal{Q}}\mathrm{ES}^{\mu}$ mentioned in Section \ref{sec:21} can be determined via a minimization formula, similar to the formulation of ES by \cite{RU02}.  
Below, a coherent conditional risk measure is a mapping from $L^\infty$  to $L^\infty(\mathcal{G})$ satisfying the four axioms of coherence but with constants replaced by elements of $L^\infty(\mathcal G)$; see \cite{DS05}.

\begin{proposition}\label{propsition:computational}
    Let $\mathcal{Q}\subseteq \mathcal{E}(\mathbb{P})$ be $w(\mathcal{M},L^\infty)$-compact and $\mathcal{G}$-convex. For $\alpha\in [0,1)$, we have 
    $$\sup_{\mu\in \mathcal{Q}}\mathrm{ES}^{\mu}_\alpha(X)=\min_{x\in\mathbb{R}}\left(x+\frac{1}{1-\alpha}\mathbb{E}\left[\rho_\mathcal{G}^{\mathcal{Q}}\left((X-x)^+\right)\right]\right), ~~~X\in L^\infty,$$
    where $\rho_\mathcal{G}^{\mathcal{Q}}:L^\infty\to L^\infty(\mathcal{G})$ is the coherent conditional risk measure  $$\rho_\mathcal{G}^{\mathcal{Q}}(X)=\esssup_{\mu\in \mathcal{Q}}\mathbb{E}^\mu[X|\mathcal{G}],~~~~X\in L^\infty.$$
    Moreover,
    the set $\mathcal{Q}_{\beta}=\left\{\mu \in \mathcal{E}(\mathbb{P}):~D_\mu \le 1/(1-\beta)\right\}$ for $\beta\in [0,1)$  is $w(\mathcal{M},L^\infty)$-compact and $\mathcal{G}$-convex.
\end{proposition}

Denote by $\mathrm{ES}_{\beta}(X|\mathcal{G})= \rho_\mathcal{G}^{\mathcal{Q}_\beta}(X)$, $X\in L^\infty$, which is the conditional version of ES in \cite{DS05}. The benefit of using $\mathcal{Q}_\beta$ is that under tractable assumptions of the conditional distribution of $X$ given $\mathcal{G}$, $\mathbb{E}[\mathrm{ES}_\beta((X-x)^+|\mathcal{G})]$ can be computed using conditional distributions. {Section \ref{sec:computation} offers a numerical illustration of this idea.}

\subsection{Two constructions of partially law-invariant coherent risk measures}
\label{sec:twocon}

 We next present two approaches to constructing partially law-invariant coherent risk measures. The first class arises from mappings analogous to conditional expectations.

Let $\psi:L^\infty\to L^\infty(\mathcal{G})$ be a mapping such that $\psi(X)=X$ for all $X\in L^\infty(\mathcal{G})$; an example would be $\psi(X)=\E[X|\mathcal G]$ or, more generally, the coherent conditional risk measures.   Consider risk measures $\rho$ of the form 
\begin{equation}
    \label{eq:moti2}
     {\rho}(X)=\tilde \rho(\psi(X)),\hspace{3pt} X\in L^\infty,
\end{equation}
where $\tilde \rho:L^\infty (\mathcal G)\to \R$  is a law-invariant mapping. 
Clearly, $\rho$ defined in \eqref{eq:moti2} is $\mathcal G$-law invariant because $\rho(X)=\tilde \rho(\psi(X))=\tilde \rho(X)=\tilde \rho(Y)=\rho(Y)$ for any $X,Y\in L^\infty(\mathcal{G})$ with $X\overset{\mathrm{d}}{=}_\mathbb{P}Y$. 

\begin{remark}\label{rem:dynamic}
 There is a large literature on conditional and dynamic risk measures, which also involve a fixed 
 sub-$\sigma$-algebra $\mathcal{G}$; see \cite{R04},  \cite{DS05}, and  \cite{FKV12}. 
 Conditional risk measures take values in random variables, different from our setting, and the motivation for partial law invariance is quite different from the above literature. Nevertheless, it is possible, as shown above, to use the theory of conditional coherent risk measures to construct $\mathcal{G}$-law-invariant coherent risk measures.
\end{remark}

Functionals of the form \eqref{eq:moti2} can be found in the two-stage market evaluation of \cite{PS14} and the market-consistent valuation of \cite{DSBLC17}. Setting $\psi(X)=\mathbb{E}[X|\mathcal{G}]$, equation \eqref{eq:moti2} offers a straightforward method for constructing partially law-invariant risk measures. Specifically,
\begin{equation}
    \label{eq:conditional1}
\rho(X)=\tilde{\rho}(\mathbb{E}[X|\mathcal{G}]), \quad X \in L^\infty.
\end{equation}
Coherent risk measures of the form \eqref{eq:conditional1} are characterized by the following result.

\begin{proposition}
    \label{prop:example1}
    For a Fatou coherent risk measure  $\rho:L^\infty\to\mathbb{R}$,  equivalent are:
    \begin{enumerate}[(i) ]
        \item There exists a Fatou coherent risk measure $\tilde{\rho}:L^\infty(\mathcal{G})\to\mathbb{R}$ such that \eqref{eq:conditional1} holds;
        \item there exists $\tilde{\rho}:L^\infty(\mathcal{G})\to\mathbb{R}$ such that \eqref{eq:conditional1} holds;
        \item $\rho(X)=\rho(\mathbb{E}[X|\mathcal{G}])$ for all $X\in L^\infty$;
        \item $ \widehat{ \mathcal{S}}_\rho \subseteq L^1(\mathcal{G})$.
    \end{enumerate}
    Moreover, if any of the above conditions are satisfied, $\rho$ is $\mathcal{G}$-law invariant if and only if $\tilde{\rho}$ is law invariant.
\end{proposition}

\begin{example}
\label{example:multcon}
The idea of using conditional expectations can be generalized further to incorporate different sub-$\sigma$-algebras (see Section \ref{sec:mults}).
Let $\mathcal{H}$ be a sub-$\sigma$-algebra  independent of $\mathcal{G}$. Let $\rho_1:L^\infty(\mathcal{G})\to\mathbb{R}$ and $\rho_2:L^\infty(\mathcal{H})\to\mathbb{R}$ both be coherent risk measures, with the former being law invariant. Define the following risk measure
\begin{equation*}    \rho(X)=\rho_1(\mathbb{E}[X|\mathcal{G}])\vee\rho_2(\mathbb{E}[X|\mathcal{H}]),\hspace{3pt}X\in L^\infty,
\end{equation*} 
where  $a\vee b$ denotes the maximum between two real numbers $a$ and $b$.
To show that $\rho$ above is $\mathcal G$-law invariant, 
we first notice that for $X$ that is $\mathcal G$-measurable, we have 
$
\rho(X) = \rho_1(X) \vee \E[X].
$
Coherence and law invariance of $\rho_1$ implies  $\rho_1(X)\ge \E[X]$ for $X\in L^\infty(\mathcal{G})$; for this result, see Theorem 4.3 of \cite{BM06}. Therefore $\rho(X)=\rho_1(X)$ for all $X\in L^\infty(\mathcal{G})$, which implies that $\rho$ is $\mathcal{G}$-law invariant.    
\end{example}

The second class of partially law-invariant coherent risk measures arises from statistical functions and tail risk measures, similar to the example of robust ES.
 
Let $\mathcal{P}_c(\mathbb{R})$ denote the set of Borel probability measures on $\mathbb{R}$ with compact support. A \emph{statistical function} is a mapping $\gamma:\mathcal{P}_c(\mathbb{R})\to\mathbb{R}$. We say that a statistical function is \emph{coherent} if the functional on $L^\infty$ given by $X\mapsto \gamma(X_{\#}\mathbb{P})$ is coherent, where $X_{\#}\mathbb{P}\in \mathcal{P}_c(\mathbb{R})$ is the distribution of $X$ under $\mathbb{P}$. We say that a statistical function is \emph{Fatou} if the functional on $L^\infty$ given by $X\mapsto \gamma(X_{\#}\mathbb{P})$ is Fatou. 

\begin{proposition}
    \label{prop:statFunc}
    Given a statistical function $\gamma$ and $\mathcal{S}\subseteq \mathcal{E}(\mathbb{P})$, the risk measure $\rho:L^\infty\to\mathbb{R}$ given by
    \begin{equation}
        \label{eq:statFunc}
        \rho(X)=\sup_{\mu\in \mathcal{S}}\gamma(X_{\#}\mu),~~~X\in L^\infty
    \end{equation}
    is Fatou, $\mathcal{G}$-law invariant, and coherent.
\end{proposition}

Using Proposition \ref{prop:statFunc}, we can establish a connection between a ``bottom-up" approach to addressing model uncertainty and the current framework. Suppose that an agent begins with a probability space $(\Omega,\mathcal G,\p_0)$, where they are certain about their probabilistic model $\p_0$ on $\mathcal{G}$. The agent wishes to extend this model to a finer-$\sigma$-algebra $\mathcal F\supseteq \mathcal G$, on which they have model uncertainty. To start, she defines a reference probability $\p$ on $\mathcal F$, such that $\p|_{\mathcal{G}}=\p_0$, to identify null sets. From here, the agent chooses a set of competing probabilistic models $\mathcal{S}\subseteq \mathcal{E}(\mathbb{P}_0)$. For financial positions in $L^\infty(\mathcal{G})$, the agent intends to use the law-invariant risk measure $\rho:L^\infty(\mathcal{G})\to\mathbb{R}$ given by $\tilde{\rho}(X)=\gamma(X_{\#}\P_0),$ $X\in L^\infty(\mathcal{G})$, where $\gamma$ is a statistical function that is Fatou and coherent. The risk measure defined in equation \eqref{eq:statFunc} would be a natural risk measure to use on $L^\infty$. According to Proposition \ref{prop:statFunc}, this risk measure is Fatou, $\mathcal{G}$-law invariant, and coherent. For a similar approach in the dynamical setting, see \cite{ES03} and \cite{HS22}.

\begin{example}\label{example:motRep}
    For $\beta\in [0,1)$, define $\mathcal S = \{\mu \in \mathcal{E}(\mathbb{P}):~D_\mu \le 1/(1-\beta)\}$ and let $\tilde{\mathbb{E}}:\mathcal{P}_c(\mathbb{R})\to\mathbb{R}$ be the expectation statistical function (mapping a distribution to its mean). Define $\rho:L^\infty\to\mathbb{R}$ by
    $$\rho(X)=\sup_{\mu \in \mathcal{S}}\tilde{\mathbb{E}}(X_{\#}\mu)=\sup_{\mu \in \mathcal{S}}\mathbb{E}^\mu[X],~~X\in L^\infty.$$
    In this case, $\rho$ is of the form in   \eqref{eq:friday} and  \eqref{eq:statFunc}. As $\mathcal{S}^\mathcal{G}=\{\p|_{\mathcal G}\}$ is almost $\mathcal G$-law invariant, we can also verify that $\rho$ is $\mathcal{G}$-law invariant from Theorem \ref{th:friday}.  Take any $\sigma$-algebra $\mathcal H$ independent of $\mathcal{G}$, such that $(\Omega,\mathcal H,\p)$ is atomless (assumed to exist; this corresponds to a notion of conditional atomless property; see Remark \ref{rem:atomless} in Appendix \ref{app:1}). We claim 
\begin{equation}\label{eq:mainex}
 \rho(X)= \E[X],~~X\in L^\infty(\mathcal{G})~~\text{and}~~
    \rho(X) =\ES_{\beta}(X),~~Y\in L^\infty(\mathcal{H}).
\end{equation}
Furthermore, these are the extreme values for random variables with a particular distribution. That is, if $Q\in \mathcal{P}_c(\mathbb{R})$ and  $X\in L^\infty$ with $X_{\#}\mathbb{P}=Q$, we claim 
$$\tilde{\mathbb{E}}(Q)\leq \rho(X)\leq \tilde{\ES}_{\beta}(Q),$$
where $\tilde{\ES}_{\beta}:\mathcal{P}_c(\mathbb{R})\to\mathbb{R}$ is the ES statistical function. The $\sigma$-algebras $\mathcal{G}$ and $\mathcal{H}$ are extreme as $\rho(X)=\tilde{\mathbb{E}}(Q)$ if $X\in L^\infty(\mathcal{G})$ and $\rho(X)=\tilde{\ES}_{\beta}(Q)$ if $X\in L^\infty(\mathcal{H})$. For proofs of the above claims, see Section \ref{app:32}.
\end{example} 

By allowing $\mathcal{S}$ to vary but maintaining the condition that $\mathcal{S}^\mathcal{G}=\{\p|_{\mathcal G}\}$, we allow the risk measure to behave ``nicely" on $L^\infty(\mathcal G)$, while leaving flexibility outside this $\sigma$-field. 
Certainly, it is not necessary to require $\mu\in \mathcal{E}(\mathbb{P})$ for all $\mu\in \mathcal S$ (like $\mathcal{S}$ in Example \ref{example:motRep}); according to Theorem \ref{th:friday}, it is necessary and sufficient to take $\mathcal{S}^\mathcal{G}$ to be almost $\mathcal{G}$-law invariant. 
Indeed, generalizing Example \ref{example:motRep},
we can have $\rho$ behave like a law-invariant risk measure $\widehat \rho$ on $L^\infty(\mathcal G)$ and like a tail risk measure generated by $\widehat \rho$ on $L^\infty(\mathcal H)$ in the sense of \cite{LW21}, this is illustrated by Example \ref{example:tail}.

\begin{example}
    \label{example:tail}
    Given a Fatou law-invariant coherent risk measure $\widehat{\rho}:L^\infty\to\mathbb{R}$, there exists a Fatou coherent statistical function $\gamma:\mathcal{P}_c(\mathbb{R})\to\mathbb{R}$ such that $\widehat{\rho}(X)=\gamma(X_{\#}\mathbb{P})$ for all $X\in L^\infty$. For $\beta\in [0,1)$, define $\mathcal S = \{\mu \in \mathcal{E}(\mathbb{P}):~D_\mu \le 1/(1-\beta)\}$. By Proposition \ref{prop:statFunc}, the risk measure $\rho$ given by \eqref{eq:statFunc} is Fatou, $\mathcal{G}$-law invariant, and coherent. We claim that the risk measure $\rho$  satisfies
     $$
            \rho(X)= \widehat{\rho}(X),~~X\in L^\infty(\mathcal G)~~~\text{and}~~~
            \rho(X) =\widehat{\rho}^\beta(X),~~X\in L^\infty(\mathcal H),
        $$
        where $\widehat{\rho}^\beta$ is the $\beta$-tail risk measure generated by $\widehat{\rho}$. See Section \ref{app:32} for a proof of this claim. 
\end{example}

\subsection{Numerical illustrations}
\label{sec:computation}

In this section, we numerically evaluate the formula from Proposition \ref{propsition:computational} to generate plots and explore how model uncertainty shapes portfolio allocations.

Let $(\Omega,\mathcal{F})=(\mathbb{R}^2,\mathcal{B}(\mathbb{R}^2))$ and $\mathbb{P}$ be multivariate Gaussian with mean $m\in \mathbb{R}^2$ and covariance matrix $\Sigma\in \mathbb{R}^{2\times 2}$, where $\Sigma_{ij}=(c+(1-c)\id_{\{i=j\}})\sigma_i\sigma_j$ for $i,j\in \{1,2\}$ (that is, an equicorrelation matrix). Let $X_1,X_2$ be the random variables given by the projections, representing the single-period percentage losses of two financial positions. We interpret $\mathbb{P}$ as the initial estimated bivariate distribution of these losses. Assuming the distribution of $X_1$ is accurately specified (see Section \ref{sec:21}), set $\mathcal{G} = \sigma(X_1)$. The objective is to evaluate the risk of portfolios of the form $\pi_1 X_1 + \pi_2 X_2$ using the risk measure $\rho_\beta = \sup_{\mu \in \mathcal{Q}_\beta} \mathrm{ES}^{\mu}_\alpha$, where $(\pi_1, \pi_2)$ lies in the unit simplex of $\mathbb{R}^2$. The parameter $\beta \in [0,1)$ represents the decision maker's confidence in $\p$ for the events outside $\mathcal{G}$. When $\beta = 0$, the risk measure $\rho_0$ coincides with $\mathrm{ES}_\alpha$, reflecting complete confidence in the reference probability model $\mathbb{P}$. As $\beta$ increases, it signals a growing doubt in the model's accuracy outside of $\mathcal{G}$, leading to more conservative risk evaluations.

In Section \ref{app:computation}, for $\beta\in [0,1)$, we derive an integral formula for  
\begin{equation}
\label{eq:RW-1}
    f_\beta(\pi_1,\pi_2,x)=\mathbb{E}\left[\mathrm{ES}_\beta\left((\pi_1 X_1+\pi_2 X_2-x)^+|X_1\right)\right],
\end{equation}
which can be evaluated using numerical integration. By Proposition \ref{propsition:computational},
$$\rho_{\beta}\left(\pi_1 X_1+\pi_2 X_2\right)=\min_{x\in \mathbb{R}}\left(x+\frac{1}{1-\alpha}f_\beta(\pi_1,\pi_2,x)\right),$$ which can be computed numerically using a convex program.

Figures \ref{fig:beta1} and \ref{fig:sigma1} display plots of the computed values of $\rho_\beta$ for various portfolio weights, modeling assumptions, and values of $\beta$. The horizontal axes in all panels are given by $\pi_1$, the weight allocated to $X_1$. Since $\rho_\beta(X_1)=\mathrm{ES}_\alpha(X_1)$ for all $\beta\in [0,1)$, the lines agree when $\pi_1=1$. 

\begin{figure}[t]%
    \centering
    \caption{Varying $\beta$, while fixing $m_1=0$, $\sigma_1=0.1$, $\sigma_2=0.1$, $c=0.5$, $\alpha=0.95$; the red dotted curve connects the optimizers for different $\beta$}%
    \includegraphics[width=15.5cm]{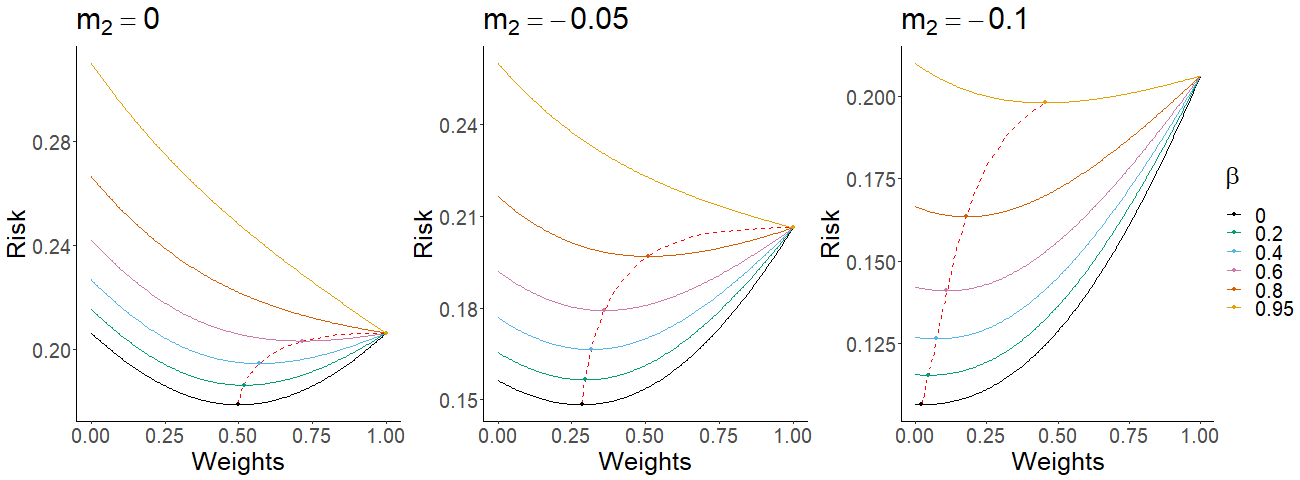}%
    \label{fig:beta1}%
\end{figure}
\begin{figure}[t]%
    \centering
    \caption{Varying $\sigma_2$, while fixing $ m_1=0$, $\sigma_1=0.1$, $c=0.5$, $\alpha=0.95$, $\beta=0.95$; the red dotted curve connects the optimizers for different $\sigma_2$}%
    \includegraphics[width=15.5cm]{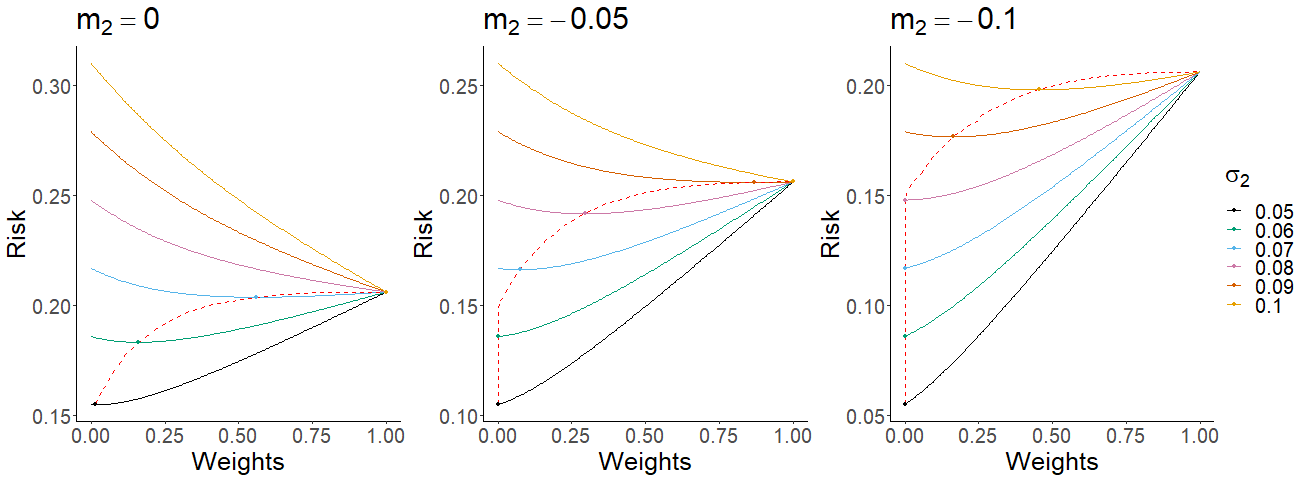}%
    \label{fig:sigma1}%
\end{figure}

Figure \ref{fig:beta1} illustrates how the risk-minimizing portfolio evolves as $\beta$ increases. When $m_2 = 0$ (left panel), an agent with $\beta = 0$ optimally selects an equal allocation between $X_1$ and $X_2$. This is consistent with the preference for pure diversification as both $X_1$ and $X_2$ share the same correctly-specified distribution. As $\beta$ increases, the risk-minimizing portfolio shifts monotonically toward a full allocation to $X_1$. This shift reflects the agent's growing concern that the model may be misspecified outside the estimated distribution for $X_1$. The center and right panels ($m_2 = -0.05$ and $-0.1$) represent cases where $X_2$ is initially expected to yield a higher return. In these cases, the same shift towards a full investment into $X_1$ persists, albeit with reduced intensity as the perceived downside of model uncertainty is partially offset by the favorable mean of $X_2$.

Figure \ref{fig:sigma1} considers a fixed $\beta = 0.95$ and varying $\sigma_2$. As expected, when $\sigma_2$ is very small, even the worst-case distributions within the ``$\beta$-neighborhood" of the estimated distribution for $X_2$ are more favorable than the well-specified distribution of $X_1$. Consequently, the agent continues to allocate fully to $X_2$, despite the high level of model uncertainty. When $\sigma_2$ is large, the risk-minimizing portfolio has an increasing weight allocated to $X_1$.

\section{Conclusion}\label{sec:con}

Our main purpose is to propose and study the concept of partial law invariance 
as a generalization of law invariance for risk measures and preference relations. This new property is appealing from a decision-theoretic perspective and interesting from a purely mathematical perspective. To summarize the main motivation, different economic sources have different levels of model uncertainty, and a law-invariant mapping or preference relation does not take these differences into account,
thus calling for the concept of partial law invariance to capture this discrepancy.
  
The economic relevance of the new framework is illustrated throughout the paper with examples. In addition to conceptual development, the new framework also leads to mathematical novelty in terms of new classes of risk measures and new representation results in forms not present in the literature. In particular, coherent and convex adjustments play an important role in understanding the new property and its representation in both the context of coherent and convex risk measures. 

The possibility of characterizing partial law invariance to other classes of risk measures may be worth discussing. The star-shaped risk measures, studied by \cite{CCMTW22}, are a natural generalization of convex risk measures, and they satisfy the following property in addition to $\rho(0)=0$.
\begin{enumerate}
    \item[] Star-shapedness: $\rho(\lambda X)\leq \lambda \rho(X)$ for all $X\in L^\infty$ and $\lambda\in [0,1]$.
\end{enumerate}
Representations of law-invariant (quasi-)star-shaped risk measures have recently been studied by \cite{HWWX21} and \cite{LRZ23}, which admit a quite different structure than the Kusuoka-type.  
Another research question is how to incorporate different probabilities, instead of projections of the same probability $\p$ as in our setting, in the context of multiple sources.

\newpage
\clearpage

\ECSwitch

\makeatletter
\renewcommand*\theHsection      {EC.\arabic{section}}
\renewcommand*\theHsubsection   {EC.\arabic{section}.\arabic{subsection}}
\renewcommand*\theHsubsubsection{EC.\arabic{section}.\arabic{subsection}.\arabic{subsubsection}}
\renewcommand*{\theHtheorem}      {EC.\arabic{theorem}}
\renewcommand*{\theHlemma}        {EC.\arabic{lemma}}
\renewcommand*{\theHproposition}  {EC.\arabic{proposition}}
\renewcommand*{\theHcorollary}    {EC.\arabic{corollary}}
\renewcommand*{\theHdefinition}   {EC.\arabic{definition}}
\renewcommand*{\theHequation}  {EC.\arabic{section}.\arabic{equation}}
\makeatother

\ECHead{E-companion: Technical Appendices}

The e-companion contains 7 appendices on the background, proofs of technical results, additional results and discussions, and extensions.

\section{Review on risk measures}
\label{app:reviewRisk}

This appendix reviews risk measures in the framework of \cite{ADEHII99}.

We say that a risk measure $\rho$ is \emph{Lebesgue (continuous)} if for all uniformly bounded sequences $(X_n)_{n\in \N}$ in $L^\infty$ that converge $\p$-a.s.~to some $X\in L^\infty$, it holds that $$\lim_{n\to\infty}\rho(X_n)=\rho(X).$$
    An example of a Lebesgue (thus Fatou) coherent risk measure is $\ES_\alpha^\mu$ for $\mu \in \mathcal M_1$ and $\alpha \in [0,1)$. 

Some well-known facts are collected below. 
A coherent risk measure $\rho$ is Fatou if and only if $\rho$ is $w(L^\infty,\mathcal{M})$-lower semicontinuous, a proof of which can be found in \citet[Theorem 4.33]{FSII16}. 
A coherent risk measure has a dual representation in the following form; see, e.g., \citet[Theorem 4.33]{FSII16}. Recall that for a Fatou coherent risk measure $\rho:L^\infty\to\mathbb{R}$, it's \emph{supporting set} is given by
$$\mathcal{S}_\rho=\left\{\mu\in \mathcal{M}_1:\mathbb{E}^{\mu}[X]\leq \rho(X)~\text{for all}~X\in L^\infty\right\}.$$
\begin{proposition}[\cite{DII02}]
    \label{prop:coherent}
    For $\rho:L^\infty\to\mathbb{R}$, the following are equivalent.
    \begin{enumerate}[(i) ]
        \item The mapping $\rho$ is a Fatou coherent risk measure;
        \item there exists a $w(\mathcal{M},L^\infty)$-closed and convex set $\mathcal{S}\subseteq \mathcal{M}_1$ such that 
        \begin{equation}
            \rho(X)=\sup_{\mu\in\mathcal{S}}\mathbb{E}^\mu[X],~~~X\in L^\infty.\label{eq:Rep1}
        \end{equation}
    
    \end{enumerate}
    In case (ii) holds, $\mathcal{S}$ is uniquely given by $\mathcal{S}_\rho$ (the supporting set of $\rho$). Moreover,  $\mathcal{R}\subseteq \mathcal{M}_1$ gives representation \eqref{eq:Rep1} if and only if $\conv(\mathcal{R})$ is $w(\mathcal{M},L^\infty)$-dense in $\mathcal{S}_\rho$.
\end{proposition}

\begin{example}
\label{ex:ES}
    For $\mu\in \mathcal{M}_1$ and $\alpha\in [0,1)$,   $\mathrm{ES}^\mu_\alpha$ admits the dual representation 
    $$\mathrm{ES}^\mu _\alpha(X)=\sup_{\nu\in \mathcal{S}_{\alpha,\mu}}\mathbb{E}^\nu[X],~~~~X\in L^\infty,$$
where  $\mathcal{S}_{\alpha,\mu}= \{\nu\in \mathcal{M}_1\mid ({1-\alpha} ) D_\nu\leq   D_\mu \}$; see e.g., \citet[Theorem 4.52]{FSII16}.
Since $\mathcal{S}_{\alpha,\mu}$ is $w\left(\mathcal{M},L^\infty\right)$-closed and convex, $\mathcal{S}_{\mathrm{ES}^\mu_\alpha}=\mathcal{S}_{\alpha,\mu}$ holds. 
\end{example}

Law-invariant coherent risk measures also admit a representation of \cite{KII01},
and Fatou continuity is automatic for law-invariant convex risk measures, as shown by \cite{JST06} and \cite{S10}. Let $\mathcal{M}_B$ denote the probability measures on $[0,1)$ equipped with the restricted Borel $\sigma$-algebra.

\begin{proposition}[\cite{KII01,JST06,S10}]
    \label{prop:lawinv}
    For $\rho:L^\infty\to\mathbb{R}$, the following are equivalent. 
    \begin{enumerate}[(i) ]
        \item The mapping $\rho$ is a Fatou law-invariant coherent risk measure;        
        \item the mapping $\rho$ is a law-invariant coherent risk measure;
        \item there exists some $\mathcal{D}\subseteq \mathcal{M}_B$ such that
 $$ \rho(X)=\sup_{Q\in\mathcal{D}}\int_{[0,1)}\mathrm{ES}_\alpha(X)~Q(\d\alpha),\hspace{3pt}X\in L^\infty. 
 $$
    \end{enumerate}
\end{proposition}

It is well known that the monetary risk measure $\rho:L^\infty\to\mathbb{R}$ is convex if and only if $\mathcal{A}_\rho$ is convex, and coherent if and only if $\mathcal{A}_\rho$ is a convex cone.

\section{Additional results and proofs accompanying Section \ref{sec:def}}
\label{app:1}

\begin{proof}{Proof of Proposition \ref{prop:cerEq}:}
    We claim the set $\{c\in \mathbb{R}:X\succsim c\}$ must be of the form $[a,\infty)$ for some $a\in \mathbb{R}$. To see this, let $x\in \{c\in \mathbb{R}:X\succsim c\}$ and $y> x$, then $X\succsim x\succ y$ and $y\in \{c\in \mathbb{R}:X\succsim c\}$ by $(\mathrm{C2})$, which proves $\{c\in \mathbb{R}:X\succsim c\}$ is an interval with no right endpoint. Since $\{c\in \mathbb{R}:X\succsim c\}$ is closed by $(\mathrm{C1})$, it must be of the form $[a,\infty)$ for some $a\in \mathbb{R}$. Using the same argument, we can show that the set $\{c\in \mathbb{R}:c\succsim X\}$  must be of the form $(\infty,b]$ for some $b\in \mathbb{R}$. As $\succsim$ is a total ordering, it must hold that $a\leq b$ since if this were not true, there would be $c\in \mathbb{R}$ that was not comparable to $X$. Assume that $a<b$, then $X\succsim a\succ b\succsim X$, which is clearly a contradiction. Thus, $a=b$ and we can define $c_X=a=b$. Then, by the definition of $a$ and $b$, $c_X\sim X$.
\Halmos \end{proof}

\begin{proof}{Proof of Proposition \ref{FirstMot}:}
  Assume that $\succsim$ is $\mathcal{G}$-law invariant. Let $X,Y\in \mathcal{X}(\mathcal{G})$ with $X\geq_1 Y$. As $(\Omega,\mathcal{G},\mathbb{P})$ is atomless, we can find a $\mathcal{G}$-measurable $U$ such that $U$ has a uniform distribution under $\mathbb{P}$. Define $\tilde{X}=F_{\mathbb{P},X}^{-1}(U)$ and $\tilde{Y}=F_{\mathbb{P},Y}^{-1}(U)$, thus $X\overset{\d}{=}_\mathbb{P}\tilde{X}$,  $Y\overset{\d}{=}_\mathbb{P}\tilde{Y}$, and $\tilde{X}\geq \tilde{Y}.$ By $\mathcal{G}$-law invariance and monotonicity, $Y\sim \tilde{Y}\succsim \tilde{X}\sim X$, thus $\succsim$ is $\mathcal{G}$-FSD mononotonic. The converse is straightforward.
\Halmos \end{proof}

\begin{proof}{Proof of Proposition \ref{prop:connection}:}
  Let $X,Y\in L^\infty(\mathcal{G})$ and $X \laweq_\p Y$.
  For $\mu\in \Q$ and $x\in \R$ we have 
      $$\mu(X\le x) = \E\left [D_\mu\id_{\{X\le x\}} \right]=\E\left [D_\mu\right] \E\left [  \id_{\{X\le x\}} \right]   = \p(X\le x)$$
  and similarly $\mu(Y\le x)= \p(Y\le x)$.
  Hence,  $X$ and $Y$ are identically distributed under each element of $\Q$. Since $\rho$ is $\Q$-based, we have $\rho(X)=\rho(Y)$. This implies that $\rho$ is $\mathcal G$-law invariant.
\Halmos \end{proof}

The next remark discusses how Proposition \ref{prop:connection} is related to conditional non-atomicity proposed by \cite{SSWWII19}.
\begin{remark} \label{rem:atomless} 
Let $\Q$ be a collection of probability measures on $(\Omega,\mathcal{F})$ and $\M_1(\Q)$ be the set of probability measures dominating (in the sense of absolute continuity) each element of $\Q$, and we assume it is non-empty (it is always non-empty if $\Q$ is finite or $\Q\subseteq \M_1$). The set $\Q$ is said to be \emph{conditionally atomless} if there exist  $\nu\in \mathcal M_1(\Q)$ and random variable $X$ such that under $\nu$, $X$ is continuously distributed and independent of $\sigma (D_\mu:\mu\in \Q )$.   
As discussed by \cite{D21},  
an equivalent statement for conditional non-atomicity of $\Q$ is that 
there exists an atomless $\sigma$-algebra $\mathcal G$ independent of $\sigma(D_\mu:\mu\in \Q)$ under some $\nu\in \M_1(\Q)$.
Clearly, when $(\Omega,\mathcal{G},\mathbb{P})$ is atomless, the condition that $\mathcal G$ is independent of $\sigma (D_\mu:\mu\in \Q )$ in  Proposition \ref{prop:connection} is sufficient for conditional non-atomicity of $\mathcal{Q}$. 
\end{remark}

\section{Additional results and proofs accompanying Section \ref{sec:3}}
\label{app:2}

\subsection{Additional results and proofs accompanying Section \ref{sec:31}}

For a preference relation $\succsim$, we define the following properties:
\begin{enumerate}
    \item[] Translation invariance: $X+c\succsim Y+c.$ for all $X,Y\in L^\infty$ with $X\succsim Y$ and $c\in \mathbb{R}$;
    \item[] Convexity: $\lambda X+(1-\lambda)Y\succsim Z$ for all $X,Y,Z\in L^\infty$ with $X\succsim Z$ and $Y\succsim Z$ and $\lambda\in [0,1]$;
    \item[] Positive homogeneity: $\lambda X\succsim\lambda Y$ for all $X,Y\in L^\infty$ with $X\succsim Y$ and $\lambda\geq 0$;
    \item[] Continuity from below: for all increasing sequences $(X_n)_{n\in \N}$ in $L^\infty$   that converges to $X$ and $Y\in L^\infty$ with $  X_n\succsim Y$ for $n\in \N$, we have $  X\succsim Y$.
\end{enumerate}  

The next proposition clarifies the connection between CRM preferences and these properties, and we provide a simple self-contained proof.
 \begin{proposition}
    \label{proposition:CRM}
    For a preference relation $\succsim$, the following are equivalent.
    \begin{enumerate}[(i) ]
        \item The preference relation $\succsim$ satisfies monotonicity, translation invariance, convexity, and positive homogeneity;
        \item  The preference relation $\succsim$ is a CRM preference relation. 
    \end{enumerate}
    Furthermore, if any of the above is true, $\succsim$ is continuous from below if and only if $\rho$ is Fatou continuous.
\end{proposition} 

\begin{proof}{Proof:}
    $(\mathrm{i})\Rightarrow (\mathrm{ii})$: Define the mapping $\rho:L^\infty\to\mathbb{R}$ by $\rho(X)=c_X.$ We know that $\rho$ represents the preference relation $\succsim$. As $\succsim$ is monotonic, it is clear that $\rho$ is monotonic. Let $X\in L^\infty$, as $X\sim \rho(X)$, by translation invariance, $\rho(X+c)\sim X+c\sim \rho(X)+c$ for all $c\in \mathbb{R}$. By assumption $(\mathrm{C2})$, $\rho(X+c)=\rho(X)+c$ for all $c\in \mathbb{R}$. Define $\mathcal{A}_\rho=\{X\in L^\infty:\rho(X)\leq 0\}=\{X\in L^\infty:X\succsim 0\}$. It is easy to see that $\mathcal{A}_\rho$ is a convex cone as the preference relation $\succsim$ is convex and positively homogenous. Thus, by \citet[Proposition 4.6]{FSII16}, $\rho$ is a coherent risk measure. $(\mathrm{ii})\Rightarrow (\mathrm{i})$: This is straightforward.

    To prove the last statement, let $\rho:L^\infty\to \mathbb{R}$ be a coherent risk measure that represents the preference relation $\succsim$.
    
    We claim that the preference relation $\succsim$ is continuous from below if and only if for all increasing sequences $(X_n)_{n\in \N}$ in $L^\infty$   that converges to $X$, we have $\rho(X_n)\uparrow \rho(X)$. For the forward, assume that $(X_n)_{n\in \N}$ in $L^\infty$ is an increasing sequence that converges to $X$. Let $l=\lim_{k\to\infty}\rho(X_k)$, which exists as $(\rho(X_n))_{n=1}^\infty$ is increasing.  Since $\rho(X_n)\leq l$ for all $n\in \N$ and $\rho(l)=l$, $X_n\succsim l$ for all $n\in \N$. Therefore, $X\succsim l$ and 
    $$\rho(X)\leq l = \lim_{k\to\infty}\rho(X_k)\leq \rho(X),$$ which implies $\rho(X_n)\uparrow \rho(X)$. For the converse, assume that $(X_n)_{n\in \N}$ in $L^\infty$ is an increasing sequence that converges to $X$ and $Y\in L^\infty$ with $  X_n\succsim Y$ for $n\in \N$. Therefore, it holds that $\rho(X_n)\leq \rho(Y)$ for all $n\in N$ and 
    $$\rho(X)=\lim_{n\to\infty}\rho(X_n)\leq \rho(Y),$$
    meaning that $X\succsim Y$. 
    
    As $\rho$ is monetary, it is a
well-known fact that $\rho$ is Fatou continuous if and only if for all increasing sequences $(X_n)_{n\in \N}$ in $L^\infty$   that converges to $X$, we have $\rho(X_n)\uparrow \rho(X)$, which proves the claim.
\Halmos \end{proof}

\begin{proof}{Proof of Proposition \ref{proposition:CRM2}:}
    $\mathrm{(i)}\Leftrightarrow \mathrm{(ii)}$: Follows from Proposition \ref{FirstMot}. $\mathrm{(i)}\Leftrightarrow \mathrm{(iii)}$: The functional $\rho:L^\infty\to\mathbb{R}$ defined by $\rho(X)=c_X$ is coherent as the preference relation $\succsim$ is a CRM preference relation. We say that $\rho$ is $\mathcal{G}$-SSD monotonic if $\rho(X)\geq \rho(Y)$ for all $X,Y\in L^\infty(\mathcal{G})$ with $X\geq_2 Y$. 
    It is well-known that a convex risk measure is law-invariant 
 if and only if it is $\mathcal F$-SSD monotonic; see e.g., \citet[Theorem 2.1 and Proposition 3.2]{MW20}. 
    Restricting this to $L^\infty (\mathcal G$),   the preference relation $\succsim$ is $\mathcal{G}$-law invariant if and only if $\rho$ is $\mathcal{G}$-SSD monotonic. It is clear that $\rho$ is $\mathcal{G}$-SSD monotonic if and only if the preference relation $\succsim$ is $\mathcal{G}$-SSD monotonic.
\Halmos \end{proof}

\begin{proof}{Proof of Proposition \ref{proposition:acceptance}:}
    The forward is trivial. For the converse note that $\{X\in L^\infty\mid \rho(X)\leq c\}=\mathcal{A}_\rho+c$. Therefore, the level sets of $\rho$ are $\mathcal G$-law invariant; the claim is straightforward to prove after this observation.
\Halmos \end{proof}

\subsection{Discussion and proof of Theorem \ref{th:friday}}
\label{app:friday}
A proof of Theorem \ref{th:friday} follows from an analysis of the supporting set of $\rho$, which will be presented in a few separate results. 

\begin{proposition}
    \label{prop:good}
    Let $\rho:L^\infty\to \mathbb{R}$ be a Fatou coherent risk measure and $\tilde{\rho}=\rho|_{L^\infty(\mathcal{G})}$. Then, $\tilde \rho$ is a Fatou coherent risk measure. Furthermore,   $\mathcal{S}_\rho^\mathcal{G}\subseteq \mathcal{S}_{\tilde{\rho}}$ and $\mathcal{S}_\rho^\mathcal{G}$ is $w(\mathcal{M}(\mathcal{G}),L^\infty(\mathcal{G}))$-dense in $\mathcal{S}_{\tilde{\rho}}$.
\end{proposition}

\begin{proof}{Proof:}
    The fact that $\tilde{\rho}$ is a Fatou coherent risk measure is clear. Since $\mathcal{S}_\rho^\mathcal{G}$ is convex and for $X\in L^\infty(\mathcal{G})$,
    \begin{align*}
    \tilde{\rho}(X)=\rho(X)=\sup_{\mu\in\mathcal{S}_\rho}\mathbb{E}^{\mu}[X]=\sup_{\mu\in\mathcal{S}_\rho}\mathbb{E}[\mathbb{E}[D_\mu|\mathcal{G}]X]=\sup_{\mu\in \mathcal{S}_\rho^{\mathcal{G}}}\mathbb{E}^{\mu}[ X],
\end{align*}
we have that $\mathcal{S}_\rho^\mathcal{G}$ is $w(\mathcal{M}(\mathcal{G}),L^\infty(\mathcal{G}))$-dense in $\mathcal{S}_{\tilde{\rho}}$ by Proposition \ref{prop:coherent}.
\Halmos \end{proof}

It is natural to wonder when $\mathcal{S}_\rho^\mathcal{G}=\mathcal{S}_{\tilde{\rho}}$ in Proposition \ref{prop:good}. Proposition \ref{prop:curious} offers sufficient conditions for this equality. 

\begin{lemma}\label{lemma:re}
    The mapping $R:\mathcal{M}\to\mathcal{M}(\mathcal{G}):\mu\mapsto \mu|_\mathcal{G}$ is $w(\mathcal{M},L^\infty)/w(\mathcal{M}(\mathcal{G}),L^\infty(\mathcal{G}))$-continuous.
\end{lemma}

\begin{proof}{Proof:}
    Let $\mu_\lambda,\mu\in \mathcal{M}$ with $\lim_\lambda\mu_\lambda=\mu$ in $w(\mathcal{M},L^\infty)$. Fix some $X\in L^\infty(\mathcal{G})$. We have
    \begin{align*}
        \mathbb{E}^{R(\mu_\lambda)}[X]=\mathbb{E}^{\mu_\lambda}[X]\to\mathbb{E}^\mu[X]=\mathbb{E}^{R(\mu)}[X].
    \end{align*}
    Therefore $R(\mu_\lambda)\to R(\mu)$ in $w(\mathcal{M}(\mathcal{G}),L^\infty(\mathcal{G}))$.
\Halmos \end{proof}

\begin{proposition}
\label{prop:curious}
    Let $\rho:L^\infty\to \mathbb{R}$ be a Fatou coherent risk measure and $\tilde{\rho}=\rho|_{L^\infty(\mathcal{G})}$. The following are sufficient conditions for $\mathcal{S}_\rho^\mathcal{G}=\mathcal{S}_{\tilde{\rho}}$. 
    \begin{enumerate}[(i) ]
        \item The risk measure $\rho$ satisfies $\rho(\mathbb{E}[X|\mathcal{G}])\leq \rho(X)$ for all $X\in L^\infty$ (equivalently, $\widehat{\mathcal{S}}^\mathcal{G}_\rho\subseteq \widehat{\mathcal{S}}_\rho$).
        \item The risk measure $\rho$ is Lebesgue continuous.
    \end{enumerate}
\end{proposition}

\begin{proof}{Proof:}
    $(\mathrm{i})$ For $\mu\in \mathcal{S}_{\tilde{\rho}}$, we have 
    \begin{align*}
        \mathbb{E}^{\mu}[X]=\mathbb{E}[D_\mu\mathbb{E}[X|\mathcal{G}]]\leq \rho(\mathbb{E}[X|\mathcal{G}])\leq \rho(X),
    \end{align*}
    for all $X\in L^\infty$. Therefore,  the extension of $\mu$ to $\mathcal F$, 
    i.e., the measure $\hat{\mu}:\mathcal{F}\to[0,1]$ given by $     
    A\mapsto \E^\mu[\id_A]$, is in $\mathcal{S}_\rho$. Since $\hat{\mu}|_\mathcal{G}=\mu$, we have $\mathcal{S}_{\tilde{\rho}}\subseteq \mathcal{S}_\rho^\mathcal{G}$. 
    $(\mathrm{ii})$ The restriction $\mu\mapsto \mu|_\mathcal{G}$ is $w(\mathcal{M},L^\infty)/w(\mathcal{M}(\mathcal{G}),L^\infty(\mathcal{G}))$-continuous by Lemma \ref{lemma:re} and $\mathcal{S}_\rho$ is $w(\mathcal{M},L^\infty)$-compact (see \citet[Corollary 4.38]{FSII16}), and hence we have that $\mathcal{S}^\mathcal{G}_{\rho}$ is $w(\mathcal{M}(\mathcal{G}),L^\infty(\mathcal{G}))$-compact. Therefore, $\mathcal{S}_\rho^\mathcal{G}$ is closed since $w(L^1,L^\infty)$ is Hausdorff, which implies $\mathcal{S}_\rho^\mathcal{G}=\mathcal{S}_{\tilde{\rho}}$.
\Halmos \end{proof}

\begin{remark}
\label{rem:open}
Although Proposition \ref{prop:curious} only provides sufficient conditions for $\mathcal{S}_\rho^\mathcal{G} = \mathcal{S}_{\tilde{\rho}}$, 
we did not find any counter-example in which  $\mathcal{S}_\rho^\mathcal{G}\ne \mathcal{S}_{\tilde{\rho}}$, 
which seems to be a challenging question. 
This relates to whether we can omit the operation of taking a closure in Theorem \ref{th:friday}. 
If $\mathcal{S}_\rho^\mathcal{G} = \mathcal{S}_{\tilde{\rho}}$  always holds, then that operation can be omitted.
\end{remark}

We are ready to prove Theorem \ref{th:friday}.

\begin{proof}{Proof of Theorem \ref{th:friday}:}
     $(\mathrm{i})\Rightarrow (\mathrm{iii})$: Since $\tilde{\rho}=\rho|_{L^\infty(\mathcal{G})}$ is law-invariant,  we have that $\mathcal{S}_{\tilde{\rho}}$ is $\mathcal{G}$-law invariant. By Proposition \ref{prop:good}, $\mathcal{S}^\mathcal{G}_\rho$ is almost $\mathcal{G}$-law invariant. $(\mathrm{iii})\Rightarrow (\mathrm{ii})$: By Proposition \ref{prop:coherent}, representation \eqref{eq:friday} holds with $\mathcal{S}_\rho$. $(\mathrm{ii})\Rightarrow (\mathrm{i})$: It is clear that $\rho$ is a Fatou coherent risk measure. It is also clear that $\mathcal{S}^\mathcal{G}$ is convex and supports $\tilde{\rho}=\rho|_{L^\infty(\mathcal{G})}$. 
    Since $\mathcal{S}^\mathcal{\mathcal{G}}$ is almost $\mathcal{G}$-law invariant, $\tilde{\rho}$ is law invariant. This implies that $\rho$ is $\mathcal{G}$-law invariant.
\Halmos \end{proof}

\section{An additional example accompanying Section \ref{sec:mults}}
\label{app:mult}

\begin{example}
    Let $\mathfrak{G}\subseteq \Sigma(\mathcal{F})$ be a collection of pair-wise independent $\sigma$-algebras. For each $\mathcal{G}\in \mathfrak{G}$, let $\rho_\mathcal{G}:L^\infty(\mathcal{G})\to\mathbb{R}$ be a Fatou law-invariant coherent risk measure. Define the risk measure
    $$\rho(X)=\sup_{\mathcal{G}\in \mathfrak{G}}\rho_\mathcal{G}(\mathbb{E}[X|\mathcal{G}]),~~~X\in L^\infty.$$
    Using a similar argument as Example $\ref{example:multcon}$ in Section \ref{sec:32}, one can show that $\rho$ is a Fatou $\mathfrak{G}$-law-invariant coherent risk measure. To see how this example relates to Theorem \ref{th:multi}, for each $\mathcal{G}\in \mathfrak{G}$, let $\iota_\mathcal{G}:\mathcal{M}_1(\mathcal{G})\to\mathcal{M}_1$ be the embedding given by $[\iota_\mathcal{G}\mu](A)=\mathbb{E}^{\mu}[\id_A]$, $A\in \mathcal{F}$. Also, denote the supporting set of $\rho_\mathcal{G}$ by $\mathcal{S}_\mathcal{G}$. Define $\mathcal{S}=\conv\left(\bigcup_{\mathcal{G}\in\mathfrak{G}}\iota_\mathcal{G}(\mathcal{S}_\mathcal{G})\right).$ By Proposition \ref{prop:example1}, $\sup_{\mu\in \iota_{\mathcal{G}}(\mathcal{S}_\mathcal{G})}\mathbb{E}^\mu[X]=\rho_{\mathcal{G}}(\mathbb{E}[X|\mathcal{G}])$ for all $X\in L^\infty$. Thus, $\rho(X)=\sup_{\mu\in \mathcal{S}}\mathbb{E}^\mu[X]$, $X\in L^\infty.$
    Since $\mathfrak{G}$ is independent, for all $\mathcal{G}\in \mathfrak{G}$, we have $\mathcal{S}^\mathcal{G}=\mathcal{S}_\mathcal{G}$, which is $w(\mathcal{M}(\mathcal{G}),L^\infty(\mathcal{G}))$-closed and $\mathcal{G}$-law invariant. Therefore, the set $\mathcal{S}$ is a valid candidate for the set found in Theorem \ref{th:multi}.

\end{example}

\section{Additional results and proofs accompanying Section \ref{sec:4}}
\label{app:3}

\subsection{Additional results and proofs accompanying Section \ref{sec:coh-cha}}
\label{app:42}

\begin{proof}{Proof of Proposition \ref{prop:motivation}:}
    (i)$\Rightarrow$(ii): For $X\in L^\infty$ define $\tau(X)=\rho(X)-\mathbb{E}\left[D_\mu X\right]$. The fact that $\tau$ is a $\mu$-coherent adjustment is easily verified. To check that $\tau$ is Fatou, let $(X_n)_{n=1}^\infty$ be a uniformly bounded sequence in $L^\infty$ which converges $\mathbb{P}$-a.s.~to some $X\in L^\infty$. By the Dominated Convergence Theorem (DCT) we have $\lim_{n\to\infty}\mathbb{E}\left[D_\mu X_n\right]=\mathbb{E}\left[D_\mu X\right]$. Therefore,
    \begin{align*}
        \liminf_{n\to\infty}\tau(X_n)= \liminf_{n\to\infty}(\rho(X_n)-\mathbb{E}\left[D_\mu X_n\right])&\geq \liminf_{n\to\infty}\rho(X_n)-\lim_{n\to\infty}\mathbb{E}\left[D_\mu X_n\right]\\&\geq \rho(X)-\mathbb{E}\left[D_\mu X\right]=\tau(X).
    \end{align*}
    The claim follows since $\rho(X)=\mathbb{E}\left[D_\mu X\right]+\tau(X)$ for all $X\in L^\infty$.
    
    (ii)$\Rightarrow $(i): The facts that $\rho$ is a coherent risk measure and that its restriction is $\E^\mu$ are easily verified. Fatou continuity follows an argument similar to the one given above.
\Halmos \end{proof}

We provide some facts regarding the representation of $\mu$-coherent adjustments.

\begin{proposition}
    \label{prop:FatouCont}
    Let $\tau:L^\infty\to\mathbb{R}$ and $\mu\in \mathcal{M}_1(\mathcal{G})$. The following are equivalent. 
    \begin{enumerate}[(i) ]
        \item The mapping $\tau\in \mathrm{CA}(\mu)$;
        \item there exists a $w(\mathcal{M},L^\infty)$-closed and convex set $\mathcal{S}_\tau\subseteq \mathcal{E}(\mu)$  such that
      \begin{align}\label{eq:Fatou1}
          \tau(X)=\sup_{\nu\in \mathcal{S}_\tau}\left(\mathbb{E}^\nu[ X]-\mathbb{E}^\mu[X]\right), ~~~X\in L^\infty;
      \end{align} 
    \item there exists $\mathcal{R}\subseteq\mathcal{E}(\mu)$
     such that $\tau(X)=\sup_{\nu\in \mathcal{R}}\left(\mathbb{E}^\nu[ X]-\mathbb{E}^\mu[X]\right)$, $X\in L^\infty$.
    \end{enumerate}
In case (ii) holds, the closed and convex set $\mathcal{S}_\tau$ is uniquely given by $$\left\{\nu\in \mathcal{M}_1: \mathbb{E}^\nu[X]-\mathbb{E}^\mu[X]\leq \tau(X) \text{ for all }X\in L^\infty\right\}.$$  
\end{proposition}

\begin{proof}{Proof:}
    $(\mathrm{i})\Rightarrow (\mathrm{ii})$: Assume   $\tau\in \mathrm{CA}(\mu)$. By Proposition \ref{prop:motivation}, the mapping $\rho=\mathbb{E}^\mu+\tau$
    is a Fatou coherent risk measure.  We have
    \begin{align*}
        \mathbb{E}^{\mu}\left[ X\right]+\tau(X)=\sup_{\nu\in \mathcal{S}_\rho}\mathbb{E}^{\nu}\left[ X\right],\hspace{3pt}X\in L^\infty.
    \end{align*} By Proposition \ref{prop:good}, $\mathcal{S}^\mathcal{G}_\rho$ must be contained in the supporting set of $\rho|_{L^\infty(\mathcal{G})}=\mathbb{E}^\mu|_{L^\infty(\mathcal{G})}$, and therefore $\mathcal{S}^\mathcal{G}_\rho=\{\mu\}$ (meaning $\mathcal{S}_\rho\subseteq \mathcal{E}(\mu)$). Note  that
    \begin{align*}
        \mathbb{E}^{\mu}[X]+\tau(X)=\mathbb{E}^{\mu}[X]+\sup_{\nu\in \mathcal{S}_\rho}\left(\mathbb{E}^{\nu}[X]-\mathbb{E}^{\mu}[X]\right),~~~X\in L^\infty,
    \end{align*}
    implying $\tau(X)=\sup_{\nu\in \mathcal{S}_\rho}\left(\mathbb{E}^{\nu}[X]-\mathbb{E}^{\mu}[X]\right)$, $X\in L^\infty$. Therefore \eqref{eq:Fatou1} holds with $\mathcal{S}_\tau=\mathcal{S}_\rho$. Furthermore, 
    \begin{align*}
        \mathcal{S}_{\tau}=\mathcal{S}_\rho&=\left\{\nu\in\mathcal{M}_1\mid \mathbb{E}^{\nu}\left[ X\right]\leq \rho(X)\text{ for all }X\in L^\infty\right\}\\&=\left\{\nu\in\mathcal{M}_1\mid \mathbb{E}^{\nu}[X]-\mathbb{E}^{\mu}[X]\leq \tau(X)\text{ for all }X\in L^\infty\right\}.
    \end{align*}
    The uniqueness follows directly from Proposition \ref{prop:coherent}.
    
    $(\mathrm{ii})\Rightarrow (\mathrm{iii})$: Obvious. $(\mathrm{iii})\Rightarrow (\mathrm{i})$: Let $\mathcal{R}\subseteq \mathcal{M}_1$ have the stated properties. Proving $\tau$
    is a $\mu$-coherent adjustment is straightforward. Let $(X_n)_{n=1}^\infty$ be a bounded sequence that converges $\mathbb{P}$-a.s.~to $X$. By the DCT, $\lim_{n\to\infty}\mathbb{E}\left[\left(D_\nu-D_\mu\right) X_n\right]=\mathbb{E}\left[\left(D_\nu-D_\mu\right) X\right]$. Therefore,
    \begin{align*}
        \tau(X)=\sup_{\nu\in \mathcal{R}}\mathbb{E}\left[\left(D_\nu-D_\mu\right) X\right]=\sup_{\nu\in \mathcal{R}}\liminf\limits_{n\to\infty}\mathbb{E}\left[\left(D_\nu-D_\mu\right) X_n\right]&\leq \liminf\limits_{n\to\infty}\sup_{\nu\in \mathcal{R}}\mathbb{E}\left[\left(D_\nu-D_\mu\right) X_n\right]\\&=\liminf\limits_{n\to\infty}\tau(X_n),
    \end{align*}
    proving the Fatou continuity of $\tau$.
\Halmos \end{proof}

\begin{corollary}
    \label{corr2}
    Let $\mu\in\mathcal{M}_1(\mathcal{G})$ and $\tau:L^\infty\to\mathbb{R}$ be a $\mu$-coherent adjustment. The following are equivalent.
    \begin{enumerate}[(i) ]
        \item The mapping $\tau\in \mathrm{CA}(\mu)$;
        \item the mapping $\tau$ is $w(L^\infty,\mathcal{M})$-lower semicontinuous;
        \item the mapping $\mathbb{E}^\mu+\tau$ is Fatou;
        \item the mapping $\mathbb{E}^\mu+\tau$ is $w(L^\infty,\mathcal{M})$-lower semicontinuous.
    \end{enumerate} 
\end{corollary}

\begin{proof}{Proof:}
    $(\mathrm{i})\Rightarrow (\mathrm{ii})$: This follows from the representation given in Proposition \ref{prop:FatouCont} and by the fact that the supremum of continuous functions is lower semicontinuous. $(\mathrm{ii})\Rightarrow (\mathrm{iv})$: The sum of lower semicontinuous functions is lower semicontinuous. $(\mathrm{iv})\Rightarrow (\mathrm{iii})$: This follows since $\mathbb{E}^\mu+\tau$ is a coherent risk measure that is lower $w(L^\infty,\mathcal{M})$-lower semicontinuous. $(\mathrm{iii})\Rightarrow(\mathrm{i})$: Since $\mathbb{E}^\mu+\tau$ is Fatou, we can use a similar argument as Proposition \ref{prop:FatouCont} to show that $\tau$ has the representation \eqref{eq:Fatou1}. Thus by Proposition \ref{prop:FatouCont}, the mapping $\tau$ is Fatou.
\Halmos \end{proof}

\begin{proof}{Proof of Theorem \ref{theorem:1}:}
    Let $\rho$ be a Fatou coherent risk measure. For $\mu\in\mathcal{S}_\rho^\mathcal{G}$, we define $\mathcal{S}(\mu)=\mathcal{E}(\mu)\cap \mathcal{S}_\rho$, note that $\mathcal{S}_\rho=\bigcup_{\mu\in \mathcal{S}_\rho^\mathcal{G}}\mathcal{S}(\mu)$. We have
    \begin{align*}
        \rho(X)=\sup_{\nu\in\mathcal{S}_\rho}\mathbb{E}\left[D_\mu X\right]&=\sup_{\nu\in\mathcal{S}_\rho}\left(\mathbb{E}\left[\mathbb{E}[D_\nu|\mathcal{G}]X\right]+\mathbb{E}[(D_\nu-\mathbb{E}[D_\nu|\mathcal{G}])X]\right)\\&=\sup_{\mu\in\mathcal{S}_\rho^\mathcal{G}}\sup_{\nu\in \mathcal{S}(\mu)}\left(\mathbb{E}\left[D_\mu X\right]+\mathbb{E}\left[(D_\nu-D_\mu)X\right]\right)\\&=\sup_{\mu\in\mathcal{S}_\rho^\mathcal{G}}\left(\mathbb{E}^\mu[X]+\tau_\mu(X)\right),~~~X\in L^\infty,
    \end{align*}
    where $\tau_\mu(X)=\sup_{\nu\in \mathcal{S}(\mu)}(\mathbb{E}^\nu[X]-\mathbb{E}^\mu[X])$, $X\in L^\infty$. The fact that $\tau_\mu\in \mathrm{CA}(\mu)$ follows from Proposition \ref{prop:FatouCont}. It is straightforward to see that $\mathcal{S}_\rho^\mathcal{G}$ is convex. For the converse, note that ($\ref{eq:32}$) is the supremum of coherent risk measures (by Proposition \ref{prop:motivation}; see \eqref{eq:31}) and thus a coherent risk measure. By Corollary \ref{corr2}, $\mathbb{E}^\mu+\tau_\mu$ is $w(L^\infty,\mathcal{M})$-lower semicontinuous for all $\mu\in \mathcal{R}$. As the supremum of $w(L^\infty,\mathcal{M})$-lower semicontinuous and coherent risk measures,   $\rho$ is also one, and hence it is Fatou. Let $\tilde{\rho}=\rho|_{L^\infty(\mathcal{G})}$. We have that $\mathcal{R}$ in ($\ref{eq:32}$) is dense in $\mathcal{S}_{\tilde{\rho}}$,  and the last claim now follows from the same argument in the proof of Theorem  \ref{th:friday}.
\Halmos \end{proof}

\subsection{Additivity on the subspace}
\label{sec:R1-additivity}

In this section, we discuss coherent risk measures that are not only subadditive but also additive on $L^\infty(\mathcal G)$, which turn out to have a strong connection to coherent adjustments.

For $\rho:L^\infty\to\mathbb{R}$, we say that $\rho$ is \emph{$\mathcal{G}$-additive} if $\rho(X+Y)=\rho(X)+\rho(Y)$ for all $X,Y\in L^\infty(\mathcal{G})$.
Additivity has the interpretation as risk neutrality under a particular probability, an idea that goes back to at least \cite{dF31}; see \citet[p.~23]{W10} for discussions on this.
Hence, 
$\mathcal{G}$-additivity reflects that the risk measure is risk neutral on the subset $L^\infty(\mathcal G)$ of $L^\infty$. In other words,
such a risk measure is neutral for risks without uncertainty; recall that in the context of Section \ref{sec:21}, random losses in $L^\infty (\mathcal G)$ do not have uncertainty in their distributions.

 For $X\in L^\infty$, we say that $\rho$ is \emph{linear along} $X$ if $\rho(\lambda X)=\lambda\rho(X)$ for all $\lambda\in\mathbb{R}$; see \cite{BKMS21} and \cite{LM22} for more discussions  on this technical property. The following proposition gives necessary and sufficient conditions for dropping the supremum in representation (\ref{eq:32}) in the case of $\mathcal{G}$-law invariance. 

\begin{proposition}\label{prop:relation}
    Let $\rho:L^\infty\to\mathbb{R}$ be a Fatou coherent risk measure. The following are equivalent.
    \begin{enumerate}[(i) ]
        \item The mapping $\rho$ is $\mathcal{G}$-additive and $\mathcal{G}$-law invariant;
        \item the mapping $\rho$ is linear along a nonconstant $X\in L^\infty(\mathcal{G})$ and $\mathcal{G}$-law invariant;
        \item there exists $\tau\in \mathrm{CA}(\mathbb{P}|_\mathcal{G})$ such that
        $$
            \rho(X)=\mathbb{E}[X]+\tau(X),\hspace{3pt}X\in L^\infty.
        $$
    \end{enumerate}
\end{proposition}
\begin{proof}{Proof:}
    (i)$\Rightarrow $(ii): Since $\rho$ is $\mathcal{G}$-additive and coherent, it is linear on $L^\infty(\mathcal{G})$. Since $L^\infty(\mathcal{G})$ contains non-constant random variables, (ii) holds. (ii)$\Rightarrow$(iii): The fact that $\rho|_{L^\infty(\mathcal{G})}=\mathbb{E}|_{L^\infty(\mathcal{G})}$ follows from Theorem 4.5 and Theorem 4.7 of \cite{BKMS21}. Using Proposition \ref{prop:motivation}, we prove the claim. (iii)$\Rightarrow$(i): This is clear by definition.
\Halmos \end{proof}

\subsection{Uniqueness of the representation in Theorem \ref{theorem:1}}
\label{app:uniq}

A natural question would be, under what conditions is representation (\ref{eq:32}) unique? When the coherent risk measure is Lebesgue continuous, there is a definite answer to this. We first need some results before we can state and prove the theorem regarding uniqueness.

\begin{proposition}
    \label{Proposition:append1}
    If $X_\lambda\to X$ in $w\left(L^1,L^\infty\right)$ then $\mathbb{E}[X_\lambda|\mathcal{G}]\to \mathbb{E}[X|\mathcal{G}]$ in $w\left(L^1,L^\infty\right)$.
\end{proposition}

\begin{proof}{Proof:}
    Let $X_\lambda\to X$ in $w\left(L^1,L^\infty\right)$ and $Y\in L^\infty$, we have
    \begin{align*}
        \mathbb{E}[\mathbb{E}[X_\lambda|\mathcal{G}]Y]=\mathbb{E}[X_\lambda \mathbb{E}[Y|\mathcal{G}]]\to\mathbb{E}[X\mathbb{E}[Y|\mathcal{G}]]=\mathbb{E}[\mathbb{E}[X|\mathcal{G}]Y].
    \end{align*}
    Implying that $\mathbb{E}[X_\lambda|\mathcal{G}]\to \mathbb{E}[X|\mathcal{G}]$ in $w\left(L^1,L^\infty\right)$.
\Halmos \end{proof}

\begin{corollary}\label{coro:lem-1}
    The set $L^1(\mathcal{G})$ is $w\left(L^1,L^\infty\right)$-closed in $L^1$.
\end{corollary}

\begin{proof}{Proof:}
    Let $X_\lambda\in L^1(\mathcal{G})$ and $X_\lambda\to X$ in $w\left(L^1,L^\infty\right)$. This also implies that $X_\lambda\to \mathbb{E}[X|\mathcal{G}]$ in $w\left(L^1,L^\infty\right)$ by Proposition \ref{Proposition:append1}. Since $w\left(L^1,L^\infty\right)$ is Hausdorff, $X=\mathbb{E}[X|\mathcal{G}]\implies X\in L^1(\mathcal{G})$. Therefore $L^1(\mathcal{G})$ is $w\left(L^1,L^\infty\right)$-closed in $L^1$.
\Halmos \end{proof}

If $\left(\mathfrak{X},\mathfrak{Y},\langle\; ,\;\rangle\right)$ is a dual pairing, we say that $\mathcal{R}\subseteq \mathfrak{X}$ is $w(\mathfrak{X},\mathfrak{Y})$-precompact if the $w(\mathfrak{X},\mathfrak{Y})$-closure of $\mathcal{R}$ is $w(\mathfrak{X},\mathfrak{Y})$-compact.

The next proposition, parallel to  Proposition \ref{prop:FatouCont}, characterizes Lebesgue-continuous $\mu$-coherent adjustments.
\begin{proposition}
    \label{prop:LebCont}  
        Let $\tau:L^\infty\to\mathbb{R}$ and $\mu\in \mathcal{M}_1(\mathcal{G})$. The following are equivalent. 
    \begin{enumerate}[(i) ]
        \item The mapping $\tau$ is a Lebesgue-continuous $\mu$-coherent adjustment;
        \item there exists a $w(\mathcal{M},L^\infty)$-compact and convex set $\mathcal{S}_\tau\subseteq \mathcal{E}(\mu)$ such that
        \begin{align}
            \tau(X)=\sup_{\nu\in \mathcal{S}_\tau}\left(\mathbb{E}^\nu[ X]-\mathbb{E}^\mu[X]\right),~~~ X\in L^\infty;\label{eq:FatouCon}
        \end{align} 
    \item there exists $w(\mathcal{M},L^\infty)$-precompact $\mathcal{R}\subseteq \mathcal{E}(\mu)$ such that
    \begin{align*}
        \tau(X)=\sup_{\nu\in \mathcal{R}}\left(\mathbb{E}^\nu[ X]-\mathbb{E}^\mu[X]\right)~~X\in L^\infty.
    \end{align*}
    \end{enumerate}
In case (ii) holds, the set $\mathcal{S}_\tau$ is uniquely given by $\left\{\nu\in \mathcal{M}_1: \mathbb{E}^\nu[X]-\mathbb{E}^\mu[X]\leq \tau(X) \text{ for all }X\in L^\infty\right\}$.   
\end{proposition}

\begin{proof}{Proof:}
    Everything follows the same line of reasoning as Proposition \ref{prop:FatouCont}, except for proving Lebesgue continuity in
    $(\mathrm{iii})\Rightarrow (\mathrm{i})$. Let $(X_n)_{n=1}^\infty$ be a sequence in $L^\infty$ of uniformly bounded random variables that converge $\mathbb{P}$-a.s.~to $X\in L^\infty$. Since $X\mapsto \sup_{\nu\in \mathcal{R}}\mathbb{E}[D_\nu X]$ is a Lebesgue-continuous coherent risk measure, we have
    \begin{align*}
        \lim_{n\to\infty}\tau(X_n)=\lim_{n\to\infty}\sup_{\nu\in\mathcal{R}}\mathbb{E}\left[\left(D_\nu-D_\mu\right)X_n\right]&=\lim_{n\to\infty}\sup_{\nu\in\mathcal{R}}\mathbb{E}\left[D_\nu X_n\right]-\lim_{n\to\infty}\mathbb{E}[D_\mu X_n]\\&=\sup_{\nu\in\mathcal{R}}\mathbb{E}[D_\nu X]-\mathbb{E}[D_\mu X]=\tau(X).
    \end{align*}
    This proves the Lebesgue continuity of $\tau$.
\Halmos \end{proof}

Let $\mathcal{R}\subseteq \mathcal{M}_1(\mathcal{G})$ be $w(\mathcal{M}(\mathcal{G}),L^\infty(\mathcal{G}))$-closed and convex, and $(\tau_\mu)_{\mu\in \mathcal{R}}$ be a CA-assignment. We say that $(\tau_\mu)_{\mu\in \mathcal{R}}$ is \emph{Lebesgue continuous} if $\tau_\mu$ is Lebesgue continuous for all $\mu\in \mathcal{R}$.
We say that $(\tau_\mu)_{\mu\in \mathcal{R}}$ is \emph{regular} if
    \begin{enumerate}[(a)]
        \item  For any fixed $X\in L^\infty$, concavity of $\mu\mapsto \tau_{\mu}(X)$ holds; that is, for all $\mu_1,\mu_2\in \mathcal{R}$ and $\lambda\in (0,1)$ we have
        \begin{align*}
          \lambda   \tau_{\mu_1}( X)+(1-\lambda) \tau_{\mu_2}(X)\leq \tau_{\lambda\mu_1+(1-\lambda)\mu_2}(X).
        \end{align*} 
        See e.g., \cite{WWW20} for concavity for distributional mixtures in the context of Choquet integrals. 
        \item A weak form of semicontinuity holds:
        If $(\mu_\lambda)_{\lambda\in \Lambda}$ is a net in $\mathcal{R}$ that converges to $\mu\in \mathcal{R}$ 
        in $w(\mathcal{M}(\mathcal{G}),L^\infty(\mathcal{G}))$,
         then for all $X\in L^\infty$,
        \begin{align*}
            \tau_{\mu}(X)\geq\liminf_{\lambda}\tau_{\mu_\lambda}(X).
        \end{align*}
    \end{enumerate}
 
        See \cite{B93} for background information on nets, subnets, and their convergence.

\begin{proposition}
    \label{Th:unique}
    Let $\rho:L^\infty\to\mathbb{R}$ be a Lebesgue-continuous coherent risk measure. There exists a unique $w(\mathcal{M}(\mathcal{G}),L^\infty(\mathcal{G}))$-compact and convex set $\mathcal{R}\subseteq \mathcal{M}_1(\mathcal{G})$ and a unique  regular Lebesgue-continuous CA-assignment $(\tau_\mu)_{\mu\in \mathcal{R}}$
    such that $\rho$ can be represented by equation \eqref{eq:32}.
\end{proposition}

\begin{proof}{Proof:}
    Let $\tilde{\rho}=\rho|_{L^\infty(\mathcal{G})}$, by Proposition \ref{prop:curious}, we have that $\mathcal{S}^\mathcal{G}_\rho=\mathcal{S}_{\tilde{\rho}}$ is compact. Let $(\tau_\mu)_{\mu\in \mathcal{S}_\rho^\mathcal{G}}$ be the same as in the construction of representation (\ref{eq:32}) in the proof of Theorem \ref{theorem:1}. We claim that for all $\mu\in \mathcal{S}_\rho^\mathcal{G}$, $\mathcal{S}(\mu)$ is $w(\mathcal{M},L^\infty)$-compact and convex. Fix some $\mu\in \mathcal{S}_\rho^\mathcal{G}$, convexity of $\mathcal{S}(\mu)$ follows from the convexity of $\mathcal S_\rho$ and the convexity of $\mathcal{E}(\mu)$. For compactness note that $\mathcal{S}(\mu)$ is $w(\mathcal{M},L^\infty)$-precompact as $\mathcal{S}(\mu)\subseteq \mathcal{S}_\rho$. Therefore, compactness of $\mathcal{S}(\mu)$ would follow if we prove that $\mathcal{S}(\mu)$ is  $w(\mathcal{M},L^\infty)$-closed. This can be seen from the fact that $\mathcal{S}_\rho$ and $\mathcal{E}(\mu)$ are both $w(\mathcal{M},L^\infty)$-closed. From the construction in \ref{theorem:1}, $\tau_\mu(X)=\sup_{\nu\in\mathcal{S}(\mu)}\left(\mathbb{E}^\nu[X]-\mathbb{E}^\mu[X]\right)$, $X\in L^\infty$, thus, by Proposition \ref{prop:LebCont},  $\tau_\mu$ is Lebesgue continuous.
    
    We will now show that $(\tau_\mu)_{\mu\in \mathcal{S}_\rho^\mathcal{G}}$ is regular. First, to show concavity (a), let $\mu_1,\mu_2\in \mathcal{S}_\rho^\mathcal{G}$, $\nu_1\in \mathcal{S}(\mu_1), \nu_2\in \mathcal{S}(\mu_2)$, $X\in L^\infty$, and $\lambda\in (0,1)$. Define $\upsilon=\lambda\mu_1+(1-\lambda)\mu_2$. Since $\mathcal{S}_\rho$ is convex, we have $\lambda\nu_1+(1-\lambda)\nu_2\in \mathcal{S}(\upsilon)$, and thus
        \begin{align}
            \mathbb{E}\left[\left(D_{\nu_1}-D_{\mu_1}\right)\lambda X\right]+\mathbb{E}\left[\left(D_{\nu_2}-D_{\mu_2}\right)(1-\lambda) X\right]\leq \tau_{\upsilon}(X).\label{eq:1234}
        \end{align}
    Taking the supremum of the left side of (\ref{eq:1234}) over $(\nu_1,\nu_2)\in \mathcal{S}(\mu_1)\times \mathcal{S}(\mu_2)$ obtains the result. To show semicontinuity (b), let $(\mu_\lambda)_{\lambda\in \Lambda}$ be a net in $\mathcal{S}_\rho^\mathcal{G}$ and $\mu\in \mathcal{S}_\rho^\mathcal{G}$ such that $\mu_\lambda\to\mu$ in $w(\mathcal{M}(\mathcal{G}),L^\infty(\mathcal{G}))$ and let $X\in L^\infty$. For each $\lambda$ choose $\nu_\lambda^*\in \mathcal{S}(\mu_\lambda)$ such that $\tau_{\mu_\lambda}(X)=\mathbb{E}[(D_{\nu_\lambda^*}-D_{\mu_\lambda})X]$, which is possible since $\mathcal{S}(\mu_\lambda)$ is $w(\mathcal{M},L^\infty)$-compact. Since the net $(\nu_\lambda^*)_{\lambda\in \Lambda}$ takes values in the $w(\mathcal{M},L^\infty)$-compact set $\mathcal{S}_\rho$, there exists a convergent subnet $(\nu_{\gamma}^*)_{\gamma\in \Gamma}$.  Denote the limit of this subnet by $\nu\in \mathcal{S}_\rho$. For $Y\in L^\infty(\mathcal{G})$, by using the tower property and the convergence, we get
        \begin{align*}
            \mathbb{E}\left[\mathbb{E}[D_\nu|\mathcal{G}]Y\right]=\mathbb{E}\left[D_\nu Y\right]=\lim_\gamma\mathbb{E}\left[D_{\nu_{\gamma}^*}Y\right]=\lim_\gamma \mathbb{E}\left[D_{\mu_\gamma}Y\right]=\mathbb{E}[D_\mu Y].
        \end{align*}
        Thus $\nu\in \mathcal{S}(\mu)$ and 
        \begin{align*}
            \tau_\mu(X)\geq\mathbb{E}\left[(D_\nu-D_\mu)X\right]=\lim_\gamma\mathbb{E}\left[\left(D_{\nu^*_\gamma}-D_{\mu_\gamma}\right)X\right]&=\lim_{\gamma}\tau_{\mu_\gamma}(X)\\&=\liminf_\gamma\tau_{\mu_\gamma}(X) \ge \liminf_\lambda \tau_{\mu_\lambda}(X).
        \end{align*}
    This shows that $(\tau_\mu)_{\mu\in \mathcal{S}_\rho^\mathcal{G}}$ is regular. 
    
    Now we prove the uniqueness of $(\tau_\mu)_{\mu\in \mathcal{S}_\rho^\mathcal{G}}$ by assuming that there exists another regular Lebesgue-continuous CA-assignment $(\tilde{\tau}_\mu)_{\mu\in \mathcal{S}_\rho^\mathcal{G}}$ such that
        \begin{align*}
            \rho(X)=\sup_{\mu\in\mathcal{S}^\mathcal{G}_\rho}\left(\mathbb{E}^\mu[ X]+\tilde{\tau}_\mu(X)\right),~~~X\in L^\infty.
        \end{align*}
     We will show that $\tau_\mu=\tilde{\tau}_\mu$ for all $\mu\in \mathcal{S}_\rho^\mathcal{G}$. For each $\mu\in\mathcal{S}_\rho^\mathcal{G}$, let $\mathcal{R}(\mu)\subseteq \mathcal{M}_1$ be the $w(\mathcal{M},L^\infty)$-compact convex set that gives representation \eqref{eq:FatouCon} in Proposition \ref{prop:LebCont} for $\tilde{\tau}_\mu$. Let $\tilde{\mathcal{S}}=\bigcup_{\mu\in\mathcal{S}_\rho^\mathcal{G}}\mathcal{R}(\mu)\subseteq \mathcal{M}_1$. Note that
    \begin{align*}
        \sup_{\nu\in \tilde{\mathcal{S}}}\mathbb{E}^{\nu}[ X]=\sup_{\mu\in \mathcal{S}_\rho^\mathcal{G}}\left(\sup_{\nu\in \mathcal{R}(\mu)}\mathbb{E}^{\nu}[ X]\right)&=\sup_{\mu\in \mathcal{S}_\rho^\mathcal{G}}\left(\mathbb{E}^{\mu}[X]+\sup_{\nu\in \mathcal{R}(\mu)}\left(\mathbb{E}^\nu[ X]-\mathbb{E}^\mu[X]\right)\right)\\&=\sup_{\mu\in \mathcal{S}_\rho^\mathcal{G}}\left(\mathbb{E}^{\mu}[ X]+\tilde{\tau}_\mu(X)\right)=\rho(X),~~~X\in L^\infty.
    \end{align*}
    Therefore $\conv(\tilde{\mathcal{S}})$ is $w(\mathcal{M},L^\infty)$-dense in $\mathcal{S}_\rho$ by Proposition \ref{prop:coherent}. We will show that regularity of $(\tilde{\tau}_\mu)_{\mu\in \mathcal{S}_\rho^\mathcal{G}}$ implies that $\tilde{\mathcal{S}}$ is convex and $w(\mathcal{M},L^\infty)$-closed. Let $\nu_1\in \mathcal{R}(\mu_1)$,  $\nu_2\in \mathcal{R}(\mu_2)$, and $\lambda\in (0,1)$. We have, for $X\in L^\infty$,
    \begin{align*}
        &\mathbb{E}\left[\left(\lambda D_{\nu_1}+(1-\lambda)D_{\nu_2}-\lambda D_{\mu_1}-(1-\lambda)D_{\mu_2}\right)X\right]\\&=\mathbb{E}\left[\left(D_{\nu_1}-D_{\mu_1}\right)\lambda X\right]+\mathbb{E}\left[\left(D_{\nu_2}-D_{\mu_2}\right)(1-\lambda) X\right]\\&\leq \tilde{\tau}_{\mu_1}(\lambda X)+\tilde{\tau}_{\mu_2}((1-\lambda)X) \leq \tau_{\lambda D_{\mu_1}+(1-\lambda)D_{\mu_2}}(X).
    \end{align*}
    Therefore, $\lambda \nu_1+(1-\lambda)\nu_2\in \mathcal{R}(\lambda \mu_1+(1-\lambda)\mu_2)\subseteq \tilde{\mathcal{S}}$
    by Proposition \ref{prop:LebCont}. This implies that $\tilde{\mathcal{S}}$ is convex. Let $(\nu_\lambda)_{\lambda\in \Lambda}$ be a net taking values in $\tilde{\mathcal{S}}$ and $\nu\in \mathcal{M}_1$ such that $\lim_\lambda \nu_\lambda=\nu$ in $w(\mathcal{M},L^\infty)$, and define $\mu_\lambda=\nu_\lambda|_\mathcal{G}$ and $\mu=\nu|_\mathcal{G}$. By Proposition \ref{Proposition:append1},  $\lim_\lambda\mu_\lambda=\mu$ in $w(\mathcal{M}(\mathcal{G}),L^\infty(\mathcal{G}))$. Since $\mathcal{S}_\rho^\mathcal{G}$ is $w(\mathcal{M}(\mathcal{G}),L^\infty(\mathcal{G}))$-closed, we have that $\mu\in \mathcal{S}_\rho^\mathcal{G}$. For $X\in L^\infty$,  
    \begin{align*}
        \mathbb{E}\left[\left(D_\nu-D_\mu\right)X\right]=\lim_\lambda\mathbb{E}\left[\left(D_{\nu_\lambda}-D_{\mu_\lambda}\right)X\right]\leq\liminf_\lambda\tilde{\tau}_{\mu_\lambda}(X)\leq \tilde{\tau}_\mu(X),
    \end{align*}
    implying that $\nu\in \mathcal{R}(\mu)\subseteq \tilde{\mathcal{S}}$ and $\tilde{\mathcal{S}}$ is $w(\mathcal{M},L^\infty)$-closed. Therefore $\tilde{\mathcal{S}}=\mathcal{S}_\rho$, which means $\mathcal{S}(\mu)=\mathcal{R}(\mu)$ for all $\mu\in \mathcal{S}_\rho^\mathcal{G}$, implying $\tilde{\tau}_\mu=\tau_\mu$ for all $\mu\in\mathcal{S}_\rho^\mathcal{G}$. 
\Halmos \end{proof}

\subsection{Additional results and proofs accompanying Section \ref{sec:strong}}
\label{app:strong}

\begin{proof}{Proof of Lemma \ref{lemma:wednesday}:}
    The only if is clear. Let $Y_1,Y_2\in L^\infty(\mathcal{G})$ with  $Y_1\overset{\mathrm{d}}{=}_\mathbb{P}Y_2$ and $Z\in \ker(\mathcal{G})$. Define $X=Z+Y_1$, we have $\mathbb{E}[X|\mathcal{G}]=Y_1$ therefore
    $$\rho(Z+Y_1)=\rho(X)=\rho(X-\mathbb{E}[X|\mathcal{G}]+Y_2)=\rho(Z+Y_2).$$
    Therefore $\rho$ is strongly $\mathcal{G}$-law invariant.
\Halmos \end{proof}

The proof of Theorem \ref{theo:gen1} relies on several results we present below. 

\begin{lemma}\citet[Lemma A.32]{FSII16}\label{lem:fel}
    Let $\mathcal{J}\subseteq \mathcal{F}$ be a sub-$\sigma$-algebra such that $(\Omega,\mathcal{J},\mathbb{P})$ is atomless. For all $X\in L^1(\mathcal{J})$ there exists a uniform $(0,1)$ random variable $U\in L^\infty(\mathcal{J})$ such that $X=F_{\mathbb{P},X}^{-1}(U)$~~$\mathbb{P}$-a.s.
\end{lemma}

\begin{lemma}\label{lemNeed}
    Let $X\in L^\infty$ (resp. $L^1$), $Y\in L^1$ (resp. $L^\infty$) and $Z\in L^1(\mathcal{G})$ (resp. $L^\infty(\mathcal{G})$) with $Y\overset{\mathrm{d}}{=}_\mathbb{P}Z$. Then there exists $W\in L^
    \infty(\mathcal{G})$ (resp. $L^1(\mathcal{G})$) such that $X\overset{\mathrm{d}}{=}_\mathbb{P}W$ and $\mathbb{E}[XY]=\mathbb{E}[WZ].$
\end{lemma}

\begin{proof}{Proof:}

Let $X\in L^\infty$, $Y\in L^1$ and $Z\in L^1(\mathcal{G})$ with $Y\overset{\mathrm{d}}{=}_\mathbb{P}Z$. Since $(\Omega,\mathcal{G},\mathbb{P})$ is atomless, by Lemma \ref{lem:fel}, there exists a uniform on $(0,1)$ random variable $U_Z\in L^\infty(\mathcal{G})$ such that $F_{\mathbb{P},Y}^{-1}(U_Z)=Z$ $\mathbb{P}$-a.s. For all $p\in [0,1]$, define $V_p=(p-U_Z)\id_{\{U_Z\leq p\}}+U_Z\id_{\{U_Z>p\}}$. It is easy to see that $V_p$ is a $\mathcal{G}$-measurable uniform $(0,1)$ random variable.  
For $p\in [0,1]$, define $W_p=F_{\mathbb{P},X}^{-1}(V_p)$. We have $W_p\in L^\infty(\mathcal{G})$ and $W_p\overset{\mathrm{d}}{=}_\mathbb{P}X$. Given a sequence $(p_n)_{n=1}^\infty$ with $p_n\to p$, it is clear that $W_{p_n}Z\to W_pZ$ $\mathbb{P}$-a.s.

For all $n\in \mathbb{N}$, $W_{p_n}\overset{\mathrm{d}}{=}_\mathbb{P}X$, implying that there exists an $M>0$ such that $\mathbb{P}(\sup_{n\in \mathbb{N}}|W_{p_n}|\leq M)=1$. Thus, $$\sup_{n\in \mathbb{N}}|W_{p_n}Z|\leq M|Z|\in L^1(\mathcal{G}),$$
and $(W_{p_n}Z)_{n=1}^\infty$ is uniformly integrable. All of this combines to show that $\lim_{n\to\infty}\mathbb{E}[W_{p_n}Z]=\mathbb{E}[W_pZ].$ Therefore, the function
$$\zeta: [0,1]\to \mathbb{R}:p\mapsto \mathbb{E}[W_pZ]$$
is continuous. Note that
$$\zeta(1)=\mathbb{E}\left[F_{\mathbb{P},X}^{-1}(1-U_Z)F_{\mathbb{P},Y}^{-1}(U_Z)\right]\leq \mathbb{E}[XY]\leq \mathbb{E}\left[F_{\mathbb{P},X}^{-1}(U_Z)F_{\mathbb{P},Y}^{-1}(U_Z)\right]=\zeta(0).$$
Thus, by the intermediate value theorem, there exists a $p_0\in [0,1]$ such that $\mathbb{E}[W_{p_0}Z]=\mathbb{E}[XY].$  Let $X\in L^1$, $Y\in L^\infty$ and $Z\in L^\infty(\mathcal{G})$ with $Y\overset{\mathrm{d}}{=}_\mathbb{P}Z$. The only thing that has to be checked is the uniform integrability condition. Let $M=\|Z\|_\infty+1$ and $K>0$, we have
\begin{align*}
    \mathbb{E}\left[|W_{p_n}Z|;|W_{p_n}Z|>K\right]\leq M\mathbb{E}\left[|W_{p_n}|;|W_{p_n}|>K/M\right]=M\mathbb{E}\left[|X|;|X|>K/M\right]
\end{align*}
for all $n\in \mathbb{N}$. This shows that $(W_{p_n}Z)_{n=1}^\infty$ is uniformly integrable.
\Halmos \end{proof}

For a $\mathcal{G}$-law invariant set $\mathcal{R}\subseteq \mathcal{M}_1(\mathcal{G})$, we say that the CA-assignment $(\tau_\mu)_{\mu\in \mathcal{R}}$ is \emph{$\mathcal{G}$-law invariant} if for all $\mu,\nu\in \mathcal{R}$ with $D_\mu\overset{\mathrm{d}}{=}_\mathbb{P}D_\nu$, it holds that $\tau_\mu=\tau_\nu$.

\begin{proposition}\label{proposition:strong}
    Let $\rho$ be a Fatou coherent risk measure. The following are equivalent.
    \begin{enumerate}[(i) ]
        \item The risk measure $\rho$ is strongly $\mathcal{G}$-law invariant;
        \item the set $\mathcal{S}_\rho^\mathcal{G}$ is $\mathcal{G}$-law invariant and for all $\mu\in \mathcal{S}_\rho$ and $Y\in L^1(\mathcal{G})$ with $\mathbb{E}[D_\mu|\mathcal{G}]\overset{\mathrm{d}}{=}_\mathbb{P}Y$, it holds that 
        $$D_\mu-\mathbb{E}[D_\mu|\mathcal{G}]+Y\in \widehat{\mathcal{S}}_\rho.$$ 
        \item There exists a convex $\mathcal{G}$-law invariant set $\mathcal{R}\subseteq \mathcal{M}_1(\mathcal{G})$ and a $\mathcal{G}$-law invariant CA-assignment $(\tau_\mu)_{\mu\in \mathcal{R}}$ such that 
        $$\rho(X)=\sup_{\mu\in \mathcal{R}}(\mathbb{E}^\mu[X]+\tau_\mu(X)),~~~X\in L^\infty.$$
    \end{enumerate}
\end{proposition}

\begin{proof}{Proof:}
    $(\mathrm{i})\Rightarrow (\mathrm{ii}):$ We will prove that $\mathcal{S}_\rho^\mathcal{G}$ is $\mathcal{G}$-law invariant after proving the latter claim. Let $\mu\in \mathcal{S}_\rho$ and let $Y\in L^1(\mathcal{G})$ with $\mathbb{E}[D_\mu|\mathcal{G}]\overset{\mathrm{d}}{=}_\mathbb{P}Y$ such that $$D_\mu-\mathbb{E}[D_\mu|\mathcal{G}]+Y\notin \widehat{\mathcal{S}}_\rho.$$
    By the Hahn-Banach separation Theorem, there exists $X\in L^\infty$ such that $$\rho(X)<\mathbb{E}[(D_\mu-\mathbb{E}[D_\mu|\mathcal{G}]+Y)X].$$ 
    By Lemma \ref{lemNeed}, find $Z\in L^\infty(\mathcal{G})$ with $\mathbb{E}[X|\mathcal{G}]\overset{\mathrm{d}}{=}_\mathbb{P}Z$ such that $\mathbb{E}[Y\mathbb{E}[X|\mathcal{G}]]=\mathbb{E}[\mathbb{E}[D_\mu|\mathcal{G}] Z]$. We have
    \begin{align*}
        \mathbb{E}[(D_\mu-\mathbb{E}[D_\mu|\mathcal{G}]+Y)X]&=\mathbb{E}[(D_\mu-\mathbb{E}[D_\mu|\mathcal{G}])X]+\mathbb{E}[YX]\\&=\mathbb{E}[(D_\mu-\mathbb{E}[D_\mu|\mathcal{G}])X]+\mathbb{E}[Y\mathbb{E}[X|\mathcal{G}]]\\&=\mathbb{E}[D_\mu(X-\mathbb{E}[X|\mathcal{G}])]+\mathbb{E}[\mathbb{E}[D_\mu|\mathcal{G}] Z]\\&=\mathbb{E}[D_\mu(X-\mathbb{E}[X|\mathcal{G}])]+\mathbb{E}[D_\mu Z]\\&=\mathbb{E}[D_\mu(X-\mathbb{E}[X|\mathcal{G}]+Z)]\\&=\mathbb{E}^\mu[X-\mathbb{E}[X|\mathcal{G}]+Z].
    \end{align*}
    Therefore, $\rho(X)<\mathbb{E}^\mu[X-\mathbb{E}[X|\mathcal{G}]+Z]\leq \rho(X-\mathbb{E}[X|\mathcal{G}]+Z)$, implying that $\rho$ is not strongly $\mathcal{G}$-law invariant. We will now prove that $\mathcal{S}_\rho^\mathcal{G}$ is $\mathcal{G}$-law invariant. Let $\mu \in \mathcal{S}_\rho$ and $Y\in L^1(\mathcal{G})$ such that $\mathbb{E}[D_\mu|\mathcal{G}]\overset{\mathrm{d}}{=}_\mathbb{P}Y$. We have $D_\mu-\mathbb{E}[D_\mu]+Y\in \widehat{\mathcal{S}}_\rho$ and $Y=\mathbb{E}[D_\mu-\mathbb{E}[D_\mu|\mathcal{G}]+Y|\mathcal{G}]\in \widehat{\mathcal{S}}_\rho^\mathcal{G}$. 
    
    $(\mathrm{ii})\Rightarrow (\mathrm{iii}):$
    Since $\rho$ is a Fatou coherent risk measure, by Theorem \ref{theorem:1}, we have 
    \begin{align*}
    \rho(X)=\sup_{\mu\in\mathcal{S}^\mathcal{G}_\rho}(\mathbb{E}^\mu[X]+\tau_\mu(X)),~~~ X\in L^\infty,
    \end{align*}
    where $\tau_\mu(X)=\sup_{\nu\in \mathcal{S}(\mu)}(\mathbb{E}^\nu[X]-\mathbb{E}^\mu[X])$, $X\in L^\infty$ and $\mathcal{S}(\mu)=\mathcal{E}(\mu)\cap \mathcal{S}_\rho$. The set $\mathcal{S}_\rho^\mathcal{G}$ is $\mathcal{G}$-law invariant by assumption and is convex as $\mathcal{S}_\rho$ is convex. Let $\mu_1,\mu_2\in\mathcal{S}_\rho^{\mathcal{G}}$ with $D_{\mu_1}\overset{\mathrm{d}}{=}_\mathbb{P}D_{\mu_2}$ and $\nu_1\in \mathcal{S}(\mu_1)$. By assumption, there exists $\upsilon\in \mathcal{S}_\rho$ with $D_\upsilon=D_{\nu_1}-D_{\mu_1}+D_{\mu_2}$. This implies that $\upsilon\in \mathcal{E}(\mu_2)$ and thus $\upsilon\in \mathcal{S}(\mu_2)$. Yet we have $D_\upsilon-D_{\mu_2}=D_{\nu_1}-D_{\mu_1}$, implying that for all $X\in L^\infty$,
\begin{align*}
    \tau_{\mu_2}(X) =\sup_{\upsilon\in \mathcal{S}(\mu_2)}\mathbb{E}[(D_\upsilon-D_{\mu_2})X]\geq \sup_{\upsilon\in \mathcal{S}(\mu_1)}\mathbb{E}[(D_\upsilon-D_{\mu_1})X]=\tau_{\mu_1}(X).
\end{align*}
By symmetry, we get the reverse inequality. 

$(\mathrm{iii})\Rightarrow (\mathrm{i}):$ Let $\mu\in \mathcal{R}$, $X\in L^\infty$ and $Y\in L^\infty(\mathcal{G})$ such that $\mathbb{E}[X|\mathcal{G}]\overset{\mathrm{d}}{=}_\mathbb{P}Y$. Define $W=X-\mathbb{E}[X|\mathcal{G}]+Y$. By Lemma \ref{lemNeed}, find $Z\in L^1(\mathcal{G})$ such $D_\mu\overset{\mathrm{d}}{=}_\mathbb{P}Z$ and $\mathbb{E}[D_\mu Y]=\mathbb{E}[Z \mathbb{E}[X|\mathcal{G}]]$. Let $\nu\in \mathcal{M}_1(\mathcal{G})$ be such that $D_\nu=Z$.
Since $\mathcal{R}$ is $\mathcal{G}$-law invariant, $\nu\in \mathcal{R}$.
We have 
\begin{align*}
    \mathbb{E}^\mu[W]+\tau_\mu(W)=\mathbb{E}[D_\mu Y]+\tau_\mu(X)=\mathbb{E}[D_\nu \mathbb{E}[X|\mathcal{G}]]+\tau_\nu(X)=\mathbb{E}[D_\nu X]+\tau_\nu(X)=\mathbb{E}^\nu[X]+\tau_\nu(X).
\end{align*}
 By Lemma \ref{lemNeed}, find $Z\in L^1(\mathcal{G})$ such $D_\mu\overset{\mathrm{d}}{=}_\mathbb{P}Z$ and $\mathbb{E}[D_\mu \mathbb{E}[X|\mathcal{G}]]=\mathbb{E}[ZY]$. Let $\nu\in \mathcal{M}_1(\mathcal{G})$ be such that $D_\nu=Z$. Since $\mathcal{R}$ is $\mathcal{G}$-law invariant, $\nu\in \mathcal{R}$. 
 We have
 \begin{align*}
     \mathbb{E}^\mu[X]+\tau_\mu(X)=\mathbb{E}[D_\mu \mathbb{E}[X|\mathcal{G}]]+\tau_\mu(W)=\mathbb{E}[D_\nu Y]+\tau_\nu(W)&=\mathbb{E}[D_\nu W]+\tau_\nu(W)\\&=\mathbb{E}^\nu[W]+\tau_\nu(W).
 \end{align*}
Therefore, $\{\mathbb{E}^\mu[X]+\tau_\mu(X):\mu\in \mathcal{R}\}=\{\mathbb{E}^\mu[W]+\tau_\mu(W):\mu\in \mathcal{R}\},$ and $\rho(X)=\rho(W).$
\Halmos \end{proof}

\begin{proof}{Proof of Theorem \ref{theo:gen1}:}
    $(\mathrm{i})\Rightarrow(\mathrm{ii})$: By Proposition \ref{proposition:strong}, we have that there exist a $\mathcal{G}$-law-invariant convex $\mathcal{R}\subseteq \mathcal{M}_1(\mathcal{G})$ and a $\mathcal{G}$-law-invariant CA-assignment $(\tau_\mu)_{\mu\in \mathcal{R}}$ such that 
    \begin{align*}
            \rho(X)=\sup_{\mu\in\mathcal{R}}\left(\mathbb{E}^{\mu}[ X]+\tau_\mu(X)\right),~~~X\in L^\infty.
        \end{align*} 
    By \citet[Theorem 4.62]{FSII16}, for $Q\in \mathcal{M}_B$ and $\mu\in \mathcal{M}_Q$, we have
    \begin{align*}
        \mathcal{M}_Q=\left\{\nu\in \mathcal{M}_1(\mathcal{G})\mid D_\nu\overset{\d}{=}D_\mu\right\}.
    \end{align*}
    Define $\mathcal{D}=\{Q\in \mathcal{M}_\mathcal{B}\mid \mathcal{M}_Q\cap \mathcal{R}\neq \varnothing\}$, and note that $\mathcal{R}=\bigcup_{Q\in \mathcal{D}}\mathcal{M}_Q$ since $\mathcal{R}$ is $\mathcal{G}$-law invariant.  Define $\tau_Q=\tau_\mu$ for some $\mu\in \mathcal{M}_Q$ (thus all $\mu\in \mathcal{M}_Q$ since $(\tau_\mu)_{\mu\in \mathcal{R}}$ is law invariant), therefore 
    \begin{align*}
        \rho(X)=\sup_{\mu\in\mathcal{R}}\left(\mathbb{E}^{\mu}[X]+\tau_\mu(X)\right)&=\sup_{Q\in \mathcal{D}}\sup_{\mu\in \mathcal{M}_Q}\left(\mathbb{E}^{\mu}[X]+\tau_\mu(X)\right)\\&\overset{\eqref{eq:sup}}{=}\sup_{Q\in \mathcal{D}}\left(\int_{[0,1)}\mathrm{ES}_\alpha(\mathbb{E}[X|\mathcal{G}])Q(\d\alpha)+\tau_Q(X)\right).
    \end{align*}
    $(\mathrm{ii})\Rightarrow (\mathrm{i})$: The fact that $\rho$ is cash invariant, convex, and positively homogenous is clear. Let $X,Y\in L^\infty$ with $X\leq Y$ and $Q\in \mathcal{D}$. Then we have
    \begin{align*}
        \int_{[0,1)}\mathrm{ES}_\alpha\left(\mathbb{E}[X|\mathcal{G}]\right)Q(\d\alpha)+\tau_Q(X)&=\sup_{\mu\in \mathcal{M}_Q}\left(\mathbb{E}^\mu\left[X\right]+\tau_Q(X)\right)\\&\leq \sup_{\mu\in \mathcal{M}_Q}\left(\mathbb{E}^{\mu}\left[ Y\right]+\tau_Q(Y)\right)\\&=\int_{[0,1)}\mathrm{ES}_\alpha\left(\mathbb{E}[Y|\mathcal{G}]\right)Q(\d\alpha)+\tau_Q(Y).
    \end{align*}
    Therefore $\rho$ is monotonic as it is the supremum of monotonic functionals, so it is coherent. 
    
    Fix $Q\in\mathcal{D}$. The mapping from      $L^\infty $ to $\mathbb{R}$  given by
    \begin{align*}
    X\mapsto \int_{[0,1)}\mathrm{ES}_\alpha(\mathbb{E}[X|\mathcal{G}])Q(\d\alpha)
    \end{align*}
   is Fatou continuous and coherent by Propositions \ref{prop:lawinv} and  \ref{FatouComp} (see Section \ref{app:32}), and therefore $w(L^\infty,\mathcal{M})$-lower semicontinuous. Thus, the functional  on $L^\infty$ given by 
    \begin{align*}
        X\mapsto \int_{[0,1)}\mathrm{ES}_\alpha(\mathbb{E}[X|\mathcal{G}])Q(\d\alpha)+\tau_Q(X)
    \end{align*}
    is $w(L^\infty,\mathcal{M})$-lower semicontinuous by Corollary \ref{corr2}, therefore $\rho$ is $w(L^\infty,\mathcal{M})$-lower semicontinuous as it is the supremum of $w(L^\infty,\mathcal{M})$-lower semicontinuous functionals. Since $\rho$ is coherent, it is Fatou. Let $X\in L^\infty$ and $Y\in L^\infty(\mathcal{G})$ such that $\mathbb{E}[X|\mathcal{G}]\overset{\mathrm{d}}{=}_\mathbb{P}Y$, let $Z=X-\mathbb{E}[X|\mathcal{G}]+Y$. For all $Q\in D$,
    \begin{align*}
      \int_{[0,1)}\mathrm{ES}_\alpha(\mathbb{E}[Z|\mathcal{G}])Q(\d \alpha)+\tau_Q(Z)&=\int_{[0,1)}\mathrm{ES}_\alpha(Y)Q(\d\alpha)+\tau_Q(X)\\&=\int_{[0,1)}\mathrm{ES}_\alpha(\mathbb{E}[X|\mathcal{G}])Q(\d\alpha)+\tau_Q(X),  
    \end{align*}
    where the last equality follows from the fact that $\mathrm{ES}_\alpha$ is law invariant. Therefore, $\rho(Z)=\rho(X)$.  
\Halmos \end{proof}

\subsection{Examples}\label{sec:ku-ex}

This section will give examples of partially law-invariant and strongly partially law-invariant risk measures. 
We use the following lemma from \citet[Section 4.3.i]{D3} frequently in this section.

\begin{lemma}
    \label{lem:good}
    Let $\rho_1:L^\infty\to\mathbb{R}$ and $\rho_2:L^\infty\to\mathbb{R}$ be Fatou coherent risk measures. Then $\rho_1\vee\rho_2$ is a Fatou coherent risk measure with supporting set equal to $\conv(\mathcal{S}_{\rho_1}\cup\mathcal{S}_{\rho_2})$.
\end{lemma}

Let $(\Omega_k,\mathcal{F}_k,\mathbb{P}_k)=([0,1],\mathcal{B}([0,1]),\lambda)$ for $k=1,2$. For the rest of the section let $(\Omega,\mathcal{F},\mathbb{P})=(\Omega_1\times \Omega_2,\mathcal{F}_1\otimes \mathcal{F}_2,\mathbb{P}_1\times\mathbb{P}_2)$ and $\mathcal{G}=\sigma(\pi_1)$ where $\pi_1(\omega,\tilde{\omega})=\omega$. Define the set $\mathcal{S}_1=\{\mu\in \mathcal{M}_1\mid D_\mu\in L^1(\mathcal{G}), D_\mu\leq 2\}$. This set is $w(\mathcal{M},L^\infty)$-closed and convex and corresponds to the Fatou coherent risk measure $\rho_1=\mathrm{ES}_{1/2}(\mathbb{E}[X|\mathcal{G}])$. 

The following example gives a Fatou $\mathcal{G}$-law-invariant coherent risk measure that is not strongly $\mathcal{G}$-law invariant.

\begin{example}
    Define $\nu\in\mathcal{M}_1$ by $D_\nu:\Omega\to \mathbb{R}:(\omega,\tilde{\omega})\mapsto 4\omega\tilde{\omega}$. We claim that the Fatou coherent risk measure $\rho:L^\infty\to\mathbb{R}:X\mapsto \mathbb{E}[D_\nu X]\vee \rho_1(X)$, which has supporting set $\mathcal{S}_\rho=\conv(\mathcal{S}_1\cup \{\nu\})$, is $\mathcal{G}$-law invariant but not strongly $\mathcal{G}$-law invariant. To see this, note $\mathbb{E}[D_\nu|\mathcal{G}](\omega,\tilde{\omega})=2\omega$. So the measure associated to $\mathbb{E}[D_\nu|\mathcal{G}]$ is in the set $\mathcal{S}_1$. Therefore, $\conv(\mathcal{S}_1\cup \{\nu\})^\mathcal{G}=\conv(\mathcal{S}_1^\mathcal{G}\cup\{\nu\}^\mathcal{G})=\mathcal{S}_1$,
    which is $\mathcal{G}$-law invariant. Thus, $\rho$ is $\mathcal{G}$-law invariant by Theorem \ref{th:friday}. Define the following random variables
    $$
        X(\omega,\tilde{\omega})=\omega,~~~
        Y(\omega,\tilde{\omega})=1-\omega~~~\text{and}~~~
        Z(\omega,\tilde{\omega})=\tilde{\omega}-1/2.
    $$
    We have that $Z\in \ker(\mathcal{G})$, $\mathbb{E}[D_\nu Z]=1/6$, $X,Y\in L^\infty(\mathcal{G})$ and $X\overset{\mathrm{d}}{=}_\mathbb{P}Y$, $\mathrm{ES}_\alpha(X)=\mathrm{ES}_\alpha(Y)=3/4$, $\mathbb{E}[D_\nu X]=2/3$ and $\mathbb{E}[D_\nu Y]=1/3.$ Thus, $\rho(Z+X)=5/6$ while $\rho(Z+Y)=3/4$, implying that $\rho$ is not strongly $\mathcal{G}$-law invariant by definition.
\end{example}

The following example gives a Fatou strongly $\mathcal{G}$-law invariant coherent risk measure that is not in the form of \eqref{eq:conditional1}.

\begin{example}
    Define $\nu\in\mathcal{M}_1$ by $D_\nu:\Omega\to \mathbb{R}:(\omega,\tilde{\omega})\mapsto 2\tilde{\omega}$. We claim that the Fatou coherent risk measure $\rho:L^\infty\to\mathbb{R}:X\mapsto \mathbb{E}[D_\nu X]\vee \rho_1(X)$ is strongly $\mathcal{G}$-law invariant. To see this let $Z\in \ker(\mathcal{G})$ and $X,Y\in L^\infty(\mathcal{G})$ with $X\overset{\mathrm{d}}{=}_\mathbb{P}Y$. We have
    \begin{align*}
        \rho(Z+X)=\rho_1(Z+X)\vee \mathbb{E}[D_\nu(Z+X)]&=\rho_1(X)\vee (\mathbb{E}[D_\nu Z]+\mathbb{E}[X])\\&=\rho_1(Y)\vee (\mathbb{E}[D_\nu Z]+\mathbb{E}[Y])=\rho(Z+Y).
    \end{align*}
    Note that $\widehat{\mathcal{S}}_\rho$ is not contained in $L^1(\G)$,
    thus $\rho$ is not in the form of \eqref{eq:conditional1}.
\end{example}

\subsection{Additional results and proofs accompanying Section \ref{sec:gen}}
\label{app:convex}

To do this, we will need some results from convex analysis. Let $\left(\mathfrak{X},\mathfrak{Y},\langle\cdot,\cdot\rangle\right)$ be a dual pairing and $f:\mathfrak{X}\to\mathbb{R}\cup\{\infty\}$ such that $\dom(f)=\left\{x\in\mathfrak{X}\mid f(x)<\infty\right\}\neq \varnothing$, define
\begin{align*}
    f^*&:\mathfrak{Y}\to\mathbb{R}\cup\{\infty\}:y\mapsto \sup_{x\in \mathfrak{X}}\left(\langle x,y\rangle-f(x)\right),\\f^{**}&:\mathfrak{X}\to\mathbb{R}\cup\{\infty,-\infty\}:x\mapsto \sup_{y\in \mathfrak{Y}}\left(\langle x,y\rangle -f^*(y)\right).
\end{align*}

The upper epigraph of $f$ is given by $G_f=\left\{(x,c)\in \mathfrak{X}\times\mathbb{R}\mid f(x)\leq c \right\}\subseteq \mathfrak{X}\times \mathbb{R}$. Equip $\mathfrak{X}\times \mathbb{R}$ with the product topology $w(\mathfrak{X},\mathfrak{Y})\times E_{\mathbb R}$, where $E_{\mathbb R}$ is the standard Euclidean topology.  
\begin{proposition}\label{prop:convex1}
    Let $f:\mathfrak{X}\to\mathbb{R}\cup\{\infty\}$ with $\dom(f)\neq \varnothing$.  
    \begin{enumerate}[(i) ]
        \item The mapping $f$ is convex if and only if $G_f$ is convex;
        \item the mapping $f$ is $w(\mathfrak{X},\mathfrak{Y})$-lower semicontinuous if and only if $G_f$ is $\left(w(\mathfrak{X},\mathfrak{Y})\times E_{\mathbb R}\right)$-closed;
        \item for all $x\in \mathfrak{X}$, $f^{**}(x)=\sup\left(g(x)\mid g\leq f,\text{ } g\text{ is }w(\mathfrak{X},\mathfrak{Y})\text{-lower semicontinuous and convex}\right)$ and $G_{f^{**}}$ is equal to the $(w(\mathfrak{X},\mathfrak{Y})\times E_{\mathbb R})$-closure of $\conv(G_f)$;
        \item the mapping $f$ is $w(\mathfrak{X},\mathfrak{Y})$-lower semicontinuous and convex if and only if $f^{**}=f$.
    \end{enumerate}
\end{proposition}

\noindent See \cite{Z02} for a proof of Proposition \ref{prop:convex1}. 

First, we obtain a representation of $\mu$-convex adjustments, similar to Proposition \ref{prop:FatouCont}. 
\begin{proposition}
    \label{prop:FatouContCon}
    Let $\tau:L^\infty\to\mathbb{R}$, $\mu\in \mathcal{M}_1(\mathcal{G})$, and $\rho = \mathbb{E}^\mu+\tau$. The following are equivalent. 
    \begin{enumerate}[(i) ]
        \item The mapping $\tau$ is a Fatou $\mu$-convex adjustment;
        \item the mapping $\rho$ is a $w(L^\infty,L^1)$-lower semicontinuous convex risk measure with $\dom(\rho^*)\subseteq \mathcal{E}(\mu)$;
        \item $\tau$ can be represented by
        \begin{align}\label{eq:convexassign}
            \tau(X)=\sup_{\nu\in \mathcal{E}(\mu)}\left(\mathbb{E}^\nu[X]-\mathbb{E}^\mu[X]-\rho^*(\nu)\right),~~~X\in L^\infty;
        \end{align}
        \item there exists $\alpha:\mathcal{E}(\mu)\to\mathbb{R}\cup\{\infty\}$ such that $\tau$ can be represented by
        \begin{align*}
            \tau(X)=\sup_{\nu\in \mathcal{E}(\mu)}\left(\mathbb{E}^\nu[X]-\mathbb{E}^\mu[X]-\alpha(\nu)\right),~~~X\in L^\infty.
        \end{align*}
    \end{enumerate}
\end{proposition}

\begin{proof}{Proof:}
    (i)$\Rightarrow$(ii): It is clear that $\rho$ is a convex risk measure and is Fatou (simple application of DCT). This implies that $\rho$ is $w(L^\infty,L^1)$-lower semicontinuous. Note by \citet[Theorem 4.33]{FSII16}, $\dom(\rho^*)\subseteq \mathcal{M}_1$. Let $\nu\in \mathcal{M}_1\backslash \mathcal{E}(\mu)$. There exists $A\in \mathcal{G}$ such that $\mathbb{E}[D_\nu \id_A]>\mathbb{E}[D_\mu \id_A]$. Thus, we have 
    \begin{align*}
        \rho^*(\nu)=\sup_{X\in L^\infty}\left(\mathbb{E}^{\nu}[X]-\mathbb{E}^\mu[X]-\tau(X)\right)&\geq\sup_{\lambda\in \mathbb{R}}\left(\lambda\mathbb{E}^{\nu}[\id_A]-\lambda\mathbb{E}^\mu[\id_A]-\tau(\lambda \id_A)\right)\\&=\sup_{\lambda\in \mathbb{R}}\left(\lambda(\mathbb{E}^{\nu}[\id_A]-\mathbb{E}^\mu[\id_A])\right)-\tau(0)=\infty.
    \end{align*}
    Hence, $\nu\notin \dom(\rho)$.
   (ii)$\Rightarrow$(iii): Since $\rho$ is convex and $w(L^\infty,L^1)$-lower semicontinuous. By Proposition \ref{prop:convex}, we have for $X\in L^\infty$,
    \begin{align*}
        \tau(X)=\rho(X)-\mathbb{E}^\mu[X]=\sup_{\nu\in \dom(\rho^*)}\left(\mathbb{E}^\nu[X]-\rho^*(\nu)\right)-\mathbb{E}^\mu[X]&=\sup_{\nu\in \mathcal{E}(\mu)}\left(\mathbb{E}^\nu[X]-\rho^*(\nu)\right)-\mathbb{E}^\mu[X]\\&=\sup_{\nu\in \mathcal{E}(\mu)}\left(\mathbb{E}^\nu[X]-\mathbb{E}^\mu[X]-\rho^*(\nu)\right).
    \end{align*}
    (iii)$\Rightarrow$(iv): This is trivial. (iv)$\Rightarrow$(i): The fact that $\tau$ is a $\mu$-convex adjustment is clear. Let $(X_n)_{n=1}^\infty$ be a uniformly bounded sequence in $L^\infty$ which converges $\mathbb{P}$-a.s.~to some $X\in L^\infty$. We have  
    \begin{align*}
        \tau(X)=\sup_{\nu\in \mathcal{E}(\mu)}\left(\mathbb{E}[(D_\nu-D_\mu)X]-\alpha(\nu)\right)&=\sup_{\nu\in \mathcal{E}(\mu)}\left(\liminf\limits_{n\to\infty}\left(\mathbb{E}[(D_\nu-D_\mu)X_n]-\alpha(\nu)\right)\right)\\&\leq \liminf\limits_{n\to\infty}\sup_{\nu\in \mathcal{E}(\mu)}\left(\mathbb{E}[(D_\nu-D_\mu)X]-\alpha(\nu)\right)\\&=\liminf\limits_{n\to\infty}\tau(X_n).
    \end{align*}
    Thus proving the Fatou continuity of $\tau$.
\Halmos \end{proof}

For $f:\mathfrak{X}\to\mathbb{R}\cup\{\infty\}$ let $L_f^{\leq c}=\left\{x\in \mathfrak{X}\mid f(x)\leq c\right\}$, $L_f^{<c}=\left\{x\in \mathfrak{X}\mid f(x)<c\right\}$, and  $L_f^{=c}=\left\{x\in \mathfrak{X}\mid f(x)=c\right\}$.  
We will also need the following two lemmas. 

\begin{lemma}\label{lemma:appendix1}
    Let $\left(\mathfrak{X},\mathfrak{Y},\langle\cdot,\cdot\rangle\right)$ be a dual pairing and 
    $f:\mathfrak{X}\to\mathbb{R}\cup\{\infty\}$ be convex with $\dom(\rho)\neq \varnothing$. Then $\overline{L^{\leq c}_f}=L^{\leq c}_{f^{**}}$ for all $c>-f^*(0)=\inf_{x\in \mathfrak{X}}f(x)=\inf_{x\in \mathfrak{X}}f^{**}(x)$.
\end{lemma}

\begin{proof}{Proof:}
    Let $c>-f^{*}(0)$. Since $f^{**}(x)\leq f(x)$ for all $x\in \mathfrak{X}$, $L^{\leq c}_f\subseteq L^{\leq c}_{f^{**}}$. We have $\overline{L^{\leq c}_f}\subseteq L^{\leq c}_{f^{**}}$ since $f^{**}$ is $w(\mathfrak{X},\mathfrak{Y})$-lower semicontinous as it is the supremum of $w(\mathfrak{X},\mathfrak{Y})$-continuous functions (meaning $L^{\leq c}_{f^{**}}$ is $w(\mathfrak{X},\mathfrak{Y})$-closed). 
    
    Let $\epsilon>0$. We claim that $L^{\leq c-\epsilon}_{f^{**}}\subseteq \overline{L^{\leq c}_f}$. Clearly this is true if $ L^{\leq c-\epsilon}_{f^{**}}=\varnothing$. Therefore assume that $ L^{\leq c-\epsilon}_{f^{**}}\neq\varnothing$ and let $x\in L^{\leq c-\epsilon}_{f^{**}}$, which means $(x,c-\epsilon)\in G_{f^{**}}$. There exists $(x_\lambda,c_\lambda)\in G_f$ such that $\lim_{\lambda}(x_\lambda,c_\lambda)=(x,c-\epsilon)$ since $\overline{G_f}=G_{f^{**}}$. Find $\lambda_0$ such that for all $\lambda\geq \lambda_0$, $c_\lambda\leq c$. Therefore for all $\lambda\geq \lambda_0$ we have $x_\lambda\in L_f^{\leq c}$. Thus the tail net $(x_\lambda)_{\lambda\geq \lambda_0}$ is contained in $L_f^{\leq c}$ and converges to $x$, therefore $x\in \overline{L_f^{\leq c}}$. This shows the claim that $L^{\leq c-\epsilon}_{f^{**}}\subseteq \overline{L_f^{\leq c}}$. Since $L^{\leq c-\epsilon}_{f^{**}}\subseteq \overline{L_f^{\leq c}}$ for arbitrary $\epsilon>0$, we have $L^{<c}_{f^{**}}\subseteq \overline{L_f^{\leq c}}$. 
    
    Clearly $L^{<c}_{f^{**}}\subseteq L^{\leq c}_{f^{**}}$, implying $\overline{L^{<c}_{f^{**}}}\subseteq L^{\leq c}_{f^{**}}$. Take $x\in L^{\leq c}_{f^{**}}$ and $y\in L^{<c}_{f^{**}}$,   which is possible because $c>-f^*(0)$. Define
    \begin{align*}
        z_n=\left(1-\frac{1}{n}\right)x+\frac{1}{n} y,
    \end{align*}
    which is clearly in $L^{<c}_{f^{**}}$ by the convexity of $f^{**}$. Since $\lim_nz_n= x$, we have $x\in \overline{L^{<c}_{f^{**}}}$, showing that $\overline{L^{<c}_{f^{**}}}= L^{\leq c}_{f^{**}}$. Therefore 
    \begin{align*}
        L^{\leq c}_{f^{**}}=\overline{L^{<c}_{f^{**}}}\subseteq\overline{L_f^{\leq c}}\subseteq L^{\leq c}_{f^{**}},
    \end{align*}
    which completes the proof.
\Halmos \end{proof}

\begin{lemma}\label{lemma:appendix2}
    Let $\rho:L^\infty\to\mathbb{R}$ be a Fatou convex risk measure. Then $\rho$ is law invariant if and only if $L^{\leq c}_{\rho^*}$ is law invariant for all $c>-\rho(0)$. 
\end{lemma}

\begin{proof}{Proof:}
    By Theorem 4.59 in \cite{FSII16}, $\rho$ is law invariant if and only if $\rho^*$ is law invariant. Clearly, if $\rho^*$ is law invariant, $L^{\leq c}_{\rho^*}$ is law invariant for all $c>-\rho(0)$. Now assume that $L^{\leq c}_{\rho^*}$ is law invariant for all $c>-\rho(0)$. Fix some $c>-\rho(0)=\inf_{\mu\in \mathcal{M}_1(\mathcal{G})}\rho^*(\mu)$, find $N\in \mathbb{N}$ such that $c-\frac{1}{N}>-\rho(0)$. We have
    \begin{align*}
        L^{=c}_{\rho^*}=\bigcap_{n\geq N}\left(L^{\leq c}_{\rho^*}/L^{\leq c-\frac{1}{n}}_{\rho^*}\right)
    \end{align*}
    is law invariant. Now let $c=-\rho(0)$, we have 
    \begin{align*}
        L^{=c}_{\rho^*}=\bigcap_{n\geq 1}\left(L^{\leq c+\frac{1}{n}}_{\rho^*}\right)
    \end{align*}
    is law invariant (since for $c<-\rho(0)$, $L^{\leq c}_{\rho^*}=\varnothing$). This shows that $\rho^*$ is law invariant.
\Halmos \end{proof}

\begin{proof}{Proof of Theorem \ref{prop:convex}:}
   Let $\rho$ be a Fatou convex risk measure. By \citet[Theorem 4.18, Theorem 4.33]{FSII16}, we have
   \begin{align*}
        \rho(X)=\sup_{\nu\in \mathcal{M}_1}\left(\mathbb{E}^\nu[ X]-\rho^*(\nu)\right)&=\sup_{\mu\in \mathcal{M}_1(\mathcal{G})}\sup_{\nu\in \mathcal{E}(\mu)}\left(\mathbb{E}^{\nu}[X]-\rho^*(\nu)\right)\\&=\sup_{\mu\in \mathcal{M}_1(\mathcal{G})}\left(\mathbb{E}^{\mu}[X]+\sup_{\nu\in \mathcal{E}(\mu)}\left(\mathbb{E}^\nu[X]-\mathbb{E}^{\mu}[X]-\rho^*(\nu)\right)\right)\\&=\sup_{\mu\in \mathcal{M}_1(\mathcal{G})}\left(\mathbb{E}^\mu[ X]+\tau_\mu(X)\right),~~~X\in L^\infty,
   \end{align*}
   where $\tau_\mu(X)=\sup_{\nu\in \mathcal{E}(\mu)}\left(\mathbb{E}^{\nu}[X]-\mathbb{E}^{\mu}[X]-\rho^*(\nu)\right)$, $X\in L^\infty$. Note that for all $\mu\in \mathcal{M}_1(\mathcal{G})$, $\tau_\mu= -\infty$ on $L^\infty$ or $\dom(\tau_\mu)=L^\infty$. By Proposition \ref{prop:FatouContCon}, $\tau_\mu$ is a Fatou $\mu$-convex adjustment. Define $N:\mathcal{M}_1(\mathcal{G})\to\mathbb{R}\cup\{\infty\}:\mu\mapsto -\tau_\mu(0)$. We claim that $N$ is convex. Let $\mu_1,\mu_2\in \mathcal{M}_1(\mathcal{G})$ and $\lambda\in [0,1]$. Define $\upsilon=\lambda\mu_1+(1-\lambda)\mu_2$, we have 
\begin{align*}
    N\left(\upsilon\right)=\inf_{\nu\in \mathcal{E}(\upsilon)}\rho^*\left(\nu\right)&\leq \inf_{(\nu_1,\nu_2)\in \mathcal{E}(\mu_1)\times \mathcal{E}(\mu_2)}\rho^*\left(\lambda\nu_1+(1-\lambda)\nu_2\right)\\&\leq \lambda\inf_{\nu_1\in \mathcal{E}(\mu_1)}\rho^*\left(\nu_1\right)+(1-\lambda)\inf_{\nu_2\in \mathcal{E}(\mu_2)}\rho^*\left(\nu_2\right)\\&=\lambda N(\mu_1)+(1-\lambda)N(\mu_2),
\end{align*}
where we are using the convexity of $\rho^*$. This completes the forward direction. The reverse direction is a simple application of Proposition \ref{prop:FatouContCon}, where one expands the convex adjustments by \eqref{eq:convexassign}. To see the last claim, define $\tilde{\rho}=\rho|_{L^\infty(\mathcal{G})}$, note 
\begin{align*}
    \tilde{\rho}(X)=\sup_{\mu\in \mathcal{M}_1(\mathcal{G})}\left(\mathbb{E}^{\mu}[X]-N(\mu)\right)=N^*(X), ~~~X\in L^\infty(\mathcal{G}).
\end{align*}
Thus $N^{**}=\tilde{\rho}^*$ and by Lemma \ref{lemma:appendix1}, for all $c>-\rho(0)$, $\overline{L^{\leq c}_N}=L^{\leq c}_{\tilde{\rho}^*}$. Using Lemma \ref{lemma:appendix2},
\begin{align*}
    \rho\text{ is }\mathcal{G}\text{-law invariant}\Leftrightarrow \tilde{\rho}\text{ is law invariant}&\Leftrightarrow L^{\leq c}_{\tilde{\rho}^*}\text{ is } \mathcal{G}\text{-law invariant for all }c>-\rho(0)\\&\Leftrightarrow\overline{L^{\leq c}_N}\text{ is }\mathcal{G}\text{-law invariant for all }c>-\rho(0)\\&\Leftrightarrow L^{\leq c}_N\text{ is almost }\mathcal{G}\text{-law invariant for all }c>-\rho(0).
\end{align*}
This completes the proof. 
\Halmos \end{proof}

\section{Additional results and proofs accompanying Section \ref{sec:instant}}

\subsection{Additional results and proofs accompanying Section \ref{sec:PartES}}

We call $\mathcal{X}$ \emph{upward-directed} if for all $X_1,X_2\in \mathcal{X}$, $\max\{X_1,X_2\}\in \mathcal{X}$.

\begin{proposition}[{\citet[Theorem A.11]{FSII16}}]\label{prop:upward1}
    Let $\mathcal{X}\subseteq L^\infty$ be such that  $\sup_{X\in \mathcal{X}}\|X\|_\infty $ is finite. Then $\esssup\mathcal{X}$ exists. Moreover, if $\mathcal{X}$ is upward-directed, then it holds that  $\mathbb{E}\left[\esssup\mathcal{X}\right]=\sup_{X\in \mathcal{X}}\mathbb{E}[X]$.
\end{proposition}

\begin{lemma}\label{lemma:Gconvex}
    Let $\mathcal{Q} \subseteq \mathcal{E}(\mathbb{P})$ be $\mathcal{G}$-convex.
    For all $X\in L^\infty$, $$\sup_{\mu \in \mathcal{Q}}\mathbb{E}^\mu[X]=\mathbb{E}\left[\esssup_{\mu\in \mathcal{Q}}\mathbb{E}^\mu[X|\mathcal{G}]\right].$$
\end{lemma}

\begin{proof}{Proof:}
    Let $X\in L^\infty$, define $I(X)=\{\mathbb{E}^\mu[X|\mathcal{G}]:\mu\in \mathcal{Q}\}$.
    We claim that $I(X)$ is upward-directed. Given $\mu_1,\mu_2\in \mathcal{Q}$, let 
    $$A=\left\{\omega:\mathbb{E}^{\mu_1}[X|\mathcal{G}](\omega)\geq\mathbb{E}^{\mu_2}[X|\mathcal{G}](\omega)\right\}.$$ Since $\{D_\mu:\mu\in \mathcal{Q}\}$ is $\mathcal{G}$-convex, 
    $\id_A D_{\mu_1}+\id_{A^c}D_{\mu_2}\in \{D_\mu:\mu\in \mathcal{Q}\}$. Let $\nu$ be the probability measure associated with $\id_A D_{\mu_1}+\id_{A^c}D_{\mu_2}$. We have 
    \begin{align*}
        \mathbb{E}^\nu[X|\mathcal{G}]=I_A\mathbb{E}^{\mu_1}[X|\mathcal{G}]+I_{A^c}\mathbb{E}^{\mu_2}[X|\mathcal{G}]=\max\left\{\mathbb{E}^{\mu_1}[X|\mathcal{G}],\mathbb{E}^{\mu_2}[X|\mathcal{G}]\right\}.
    \end{align*}
    Since $\mathbb{E}^\nu[X|\mathcal{G}]\in I(X)$, $I(X)$ is upward directed. By Proposition \ref{prop:upward1},
    \begin{align*}
        \mathbb{E}[\esssup I(X)]=\sup_{\mu\in \mathcal{Q}}\mathbb{E}[\mathbb{E}^\mu[X|\mathcal{G}]]=\sup_{\mu\in \mathcal{Q}}\mathbb{E}^\mu[\mathbb{E}^\mu[X|\mathcal{G}]]=\sup_{\mu\in \mathcal{Q}}\mathbb{E}^\mu[X].
    \end{align*}
    Since $X\in L^\infty$ was arbitrary, the proof is complete.
\Halmos \end{proof}

\begin{proof}{Proof of Proposition \ref{propsition:computational}:}
    Fix some $X\in L^\infty$, we have
    \begin{align}\label{eq:minmax}
        \overline{\mathrm{ES}}_\alpha^\mathcal{Q}[X]=\sup_{\mu\in \mathcal{Q}}\mathrm{ES}^\mu_\alpha[X]=\sup_{\mu\in \mathcal{Q}}\min_{x\in \mathbb{R}}\left(x+\frac{1}{1-\alpha}\mathbb{E}^\mu\left[(X-x)^+\right]\right).
    \end{align}
    Define $f:\mathcal{Q}\times \mathbb{R}\to\mathbb{R}:(\mu,x)\mapsto x+\frac{1}{1-\alpha}\mathbb{E}^\mu[(X-x)^+]$. Since $\mathcal{Q}$ is $\mathcal{G}$-convex, it is convex. 
    Clearly, $f(\cdot, x)$ is $w(\mathcal M, L^\infty)|_{\mathcal Q}$-continuous and affine for all $x\in \mathbb{R}$. Likewise, $f(\mu, \cdot)$ is continuous and convex for all $\mu\in \mathcal{Q}$. By the minimax Theorem from \citet[Theorem I.1.1]{MSZ15}, we are able to switch the minimum and supremum in \eqref{eq:minmax}, thus
    \begin{align*}
        \overline{\mathrm{ES}}_\alpha^\mathcal{Q}[X]=\min_{x\in\mathbb{R}}\sup_{\mu\in \mathcal{Q}}\left(x+\frac{1}{1-\alpha}\mathbb{E}^\mu\left[(X-x)^+\right]\right)&=\min_{x\in\mathbb{R}}\left(x+\frac{1}{1-\alpha}\sup_{\mu\in \mathcal{Q}}\mathbb{E}^\mu\left[(X-x)^+\right]\right)\\&=\min_{x\in\mathbb{R}}\left(x+\frac{1}{1-\alpha}\mathbb{E}\left[\esssup_{\mu\in \mathcal{Q}}\mathbb{E}^\mu\left[(X-x)^+|\mathcal{G}\right]\right]\right)\\&=\min_{x\in\mathbb{R}}\left(x+\frac{1}{1-\alpha}\mathbb{E}\left[\rho_\mathcal{G}^\mathcal{Q}\left((X-x)^+\right)\right]\right),
    \end{align*}
    where the third equality follows from Lemma \ref{lemma:Gconvex}.
    
    Let $\mu_1,\mu_2\in \mathcal{Q}_\beta$ and $\lambda\in L^\infty(\mathcal{G})$ with $0\leq \lambda\leq 1$, clearly 
    $$0\leq \lambda D_{\mu_1}+(1-\lambda)D_{\mu_2}\leq \frac{1}{1-\beta}.$$
    Let $A\in \mathcal{G}$,
    \begin{align*}
        \mathbb{E}[(\lambda D_{\mu_1}+(1-\lambda)D_{\mu_2})\id_A]&=\mathbb{E}[\lambda\id_A\mathbb{E}[D_{\mu_1}|\mathcal{G}]+(1-\lambda)\id_A\mathbb{E}[D_{\mu_1}|\mathcal{G}]]\\&=\mathbb{E}[\lambda\id_A+(1-\lambda)\id_A]=\mathbb{P}(A).
    \end{align*}
    Therefore, $\lambda D_{\mu_1}+(1-\lambda)D_{\mu_2}\in \widehat{\mathcal{Q}}_\beta$.
\Halmos \end{proof}

\subsection{Additional results and proofs accompanying Section \ref{sec:twocon}}

\label{app:32}

\begin{proposition}\label{FatouComp}
    Let $\tilde{\rho}$ in \eqref{eq:conditional1} be monotonic (resp.~coherent) and Fatou. Then, the resulting risk measure $\rho$ in \eqref{eq:conditional1} is also monotonic (resp.~coherent) and Fatou.
\end{proposition}

\begin{proof}{Proof:}
    It is straightforward that $\rho$ is monotonic.
    Let $(X_n)_{n\in \N}$ be an increasing sequence in $L^\infty$ that converges $\p$-a.s.~to some $X\in L^\infty$. Since $(X_n)_{n\in \N}$ is uniformly bounded, it is uniformly integrable. Thus $X_n$ converges to $X$ in $L^1$. Therefore, $(\mathbb{E}[X_n|\mathcal{G}])_{n\in \N}$ converges in $L^1$ to $\mathbb{E}[X|\mathcal{G}]$ by Jensen's inequality. Thus, $(\mathbb{E}[X_n|\mathcal{G}])_{n\in \N}$ converges in probability to $\mathbb{E}[X|\mathcal{G}]$. Since $(\mathbb{E}[X_n|\mathcal{G}])_{n\in \N}$ is also increasing, it converges $\p$-a.s.~to $\mathbb{E}[X|\mathcal{G}]$. Therefore, by Fatou continuity of $\tilde{\rho}$,
    $$\lim_{n\to\infty}\rho(X_n)\leq \rho(X)=\tilde{\rho}(\mathbb{E}[X|\mathcal{G}])\leq \liminf_{n\to\infty}\tilde{\rho}(\mathbb{E}[X_n|\mathcal{G}])=\lim_{n\to\infty}\rho(X_n).$$
    Hence, $\rho(X)=\lim_{n\to\infty}\rho(X_n)$. This means that $\rho$ is continuous from below. It is well-known that continuity from below implies Fatou-continuity under the monotonicity assumption. The fact that $\rho$ is coherent is straightforward.
\Halmos \end{proof}

\begin{proof}{Proof of Proposition \ref{prop:example1}:}
    $(\mathrm{i})\Rightarrow (\mathrm{ii})$ and $(\mathrm{ii})\Rightarrow (\mathrm{iii})$: These are clear. $(\mathrm{iii})\Rightarrow (\mathrm{iv})$: We have
    \begin{align*}
        \rho(X)=\rho(\mathbb{E}[X|\mathcal{G}])=\sup_{\mu\in \mathcal{S}_\rho}\mathbb{E}[D_\mu \mathbb{E}[X|\mathcal{G}]]=\sup_{\mu\in \mathcal{S}_\rho}\mathbb{E}[ \mathbb{E}[D_\mu|\mathcal{G}]X].
    \end{align*}
    Therefore the set $\{\mathbb{E}[D_\mu|\mathcal{G}]: \mu\in \mathcal{S}_\rho\}\subseteq \widehat{\mathcal{S}}_\rho$ is $w(L^1,L^\infty)$-dense by Proposition \ref{prop:coherent}. Since $L^1(\mathcal{G})$ is $w(L^1,L^\infty)$-closed by Corollary \ref{coro:lem-1}, we have   $\widehat{\mathcal{S}}_\rho\subseteq L^1(\mathcal{G})$.
    $(\mathrm{iv})\Rightarrow (\mathrm{i})$: Define the map 
    \begin{align*}
        \tilde{\rho}:L^\infty(\mathcal{G})\to\mathbb{R}:X\mapsto \sup_{\mu\in \mathcal{S}_\rho}\mathbb{E}[D_\mu X].
    \end{align*}
    The mapping $\tilde{\rho}$ is clearly coherent and Fatou. By the tower property of conditional expectation, we have 
    \begin{align*}
        \rho(X)=\sup_{\mu\in \mathcal{S}_\rho}\mathbb{E}[D_\mu X]=\sup_{\mu\in \mathcal{S}_\rho}\mathbb{E}[D_\mu\mathbb{E}[X|\mathcal{G}]]=\tilde{\rho}(\mathbb{E}[X|\mathcal{G}]).
    \end{align*}
    The final statement of the proposition follows by definition.
\Halmos \end{proof}

\begin{proof}{Proof of Proposition \ref{prop:statFunc}:}
Define $\tilde{\rho}:L^\infty\to\mathbb{R}$ by $\tilde{\rho}(X)=\gamma(X_{\#}\mathbb{P})$, $X\in L^\infty$. By definition, $\tilde{\rho}$ is coherent and Fatou. For $\mu\in \mathcal{S}$, define $\rho^\mu:L^\infty\to\mathbb{R}$ by $\rho^\mu(X)=\gamma(X_{\#}\mu)$, $X\in L^\infty$. It is straightforward to see that for $X,Y\in L^\infty$, $\rho^\mu(X)=\tilde{\rho}(Y)$ if $F_{\mu,X}=F_{\mathbb{P},Y}$.
We will first prove the coherence of $\rho^{\mu}$. Translation invariance and positive homogeneity of $\rho^\mu$ is apparent.
Let $X,Y\in L^\infty$ with $X\leq Y$, we have that $F_{\mu,Y}(x)\leq F_{\mu,X}(x)$ for all $x\in \mathbb{R}$ or $F_{\mu,X}^{-1}(u)\leq F_{\mu,Y}^{-1}(u)$ for all $u\in [0,1]$. Since $(\Omega,\mathcal{F},\mathbb{P})$ is atomless, there exists a uniform $(0,1)$ random variable $U$. Define $X_\mu =F_{\mu,X}^{-1}(U)$ and $Y_\mu=F_{\mu,Y}^{-1}(U).$ We have $F_{\mu,X}=F_{\mathbb{P},X_\mu}$, $F_{\mu,Y}=F_{\mathbb{P},Y_\mu}$ and $X_\mu\leq Y_\mu.$ Thus,
$\rho^\mu(X)=\tilde{\rho}(X_\mu)\leq \tilde{\rho}(Y_\mu)=\rho^\mu(Y).$
Therefore $\rho^{\mu}$ is monontonic.

Let $X,Y\in L^\infty$. By \citet[Theorem 11.7.5]{DU02}, there exists a Borel-measurable $f:(0,1)\to \mathbb{R}^{2}$ such that $f_{\#}\lambda=(X,Y)_\#\mu$, where $\lambda$ is the Lebesgue measure. Since $(\Omega,\mathcal{F},\mathbb{P})$ is atomless, there exists a uniform $(0,1)$ random variable $U\in L^\infty$. Define $(X_\mu,Y_\mu)=f\circ U$, we have $(X,Y)_\#\mu=(X_\mu,Y_\mu)_\#\mathbb{P}.$ Therefore $F_{\mu,X}=F_{\mathbb{P},X_\mu}$, $F_{\mu,Y}=F_{\mathbb{P},Y_\mu}$ and  $F_{\mu,X+Y}=F_{\mathbb{P},X_\mu+Y_\mu}$. Thus, 
$$\rho^\mu(X+Y)=\tilde{\rho}(X_\mu+Y_\mu)\leq \tilde{\rho}(X_\mu)+\tilde{\rho}(Y_\mu)=\rho^\mu(X)+\rho^\mu(Y).$$
Therefore, $\rho^{\mu}$ is subadditive.

Let $(X_n)_{n\in \N}$ in $L^\infty$ be uniformly bounded and converge $\p$-a.s.~to $X\in L^\infty$. Therefore, $(X_n)_{n\in \N}$ converges $\mu$-a.s.~to $X$. Thus, $$\lim_{n\to\infty}F^{-1}_{\mu,X_n}(q)=F^{-1}_{\mu,X}(q)$$
for all the continuity points of $F^{-1}_{X,\mu}$. Since $(\Omega,\mathcal{F},\mathbb{P})$ is atomless, there exists a uniform $(0,1)$ random variable $U$. Define $X_{\mu}^n=F_{\mu,X_n}^{-1}(U)$, $X_\mu=F_{\mu,X}^{-1}(U)$. We have $F_{\mu,X_n}=F_{\mathbb{P},X_{\mu}^n}$, $F_{\mu,X}=F_{\mathbb{P},X_\mu}$. Since the set of discontinuity points of $F^{-1}_{X,\mu}$ is countable, $X_\mu^n$ converges $\mathbb{P}$-a.s.~to $X_\mu$. Thus,
$$\rho^\mu(X)=\tilde{\rho}(X_\mu)\leq\liminf_{n\to\infty}\tilde{\rho}(X_\mu^n)=\liminf_{n\to\infty}\rho^\mu(X_n),$$
 proving  Fatou continuity.

Finally, since
$$\rho(X)=\sup_{\mu\in \mathcal{S}}\gamma(X_{\#}\mu)=\sup_{\mu\in \mathcal{S}}\rho^\mu(X),$$
it follows that $\rho$ is the supremum of coherent and Fatou risk measures. Thus, $\rho$ is coherent and Fatou. The fact that $\rho$ is $\mathcal{G}$-law invariant is a straightforward consequence of the fact that $\mathcal{S}\subseteq \mathcal{E}(\mathbb{P}_0).$ 
\Halmos \end{proof}

\begin{proof}{Proof of the claims in Example \ref{example:motRep}:} \textbf{Claim 1:} The first equality is straightforward.
To show the second equality, let $X\in L^\infty(\mathcal{H})$. From the dual representation of ES in Example \ref{ex:ES}, we get 
$$\rho(X)= \sup_{\mu\in \mathcal{S}}\E^\mu [X]\leq \sup_{\mu\in \mathcal{S}_{\beta,\mathbb{P}}}\E^\mu [X] = \ES_\beta (X).$$
For the other direction,
let $A\in \mathcal{F}$ be an event independent of $\mathcal G$ satisfying $\mathbb{E}[X|A]=\mathrm{ES}_\beta(X)$ and $\mathbb{P}(A)=1-\beta$.
The existence of such an event (called a $\beta$-tail event) is justified by
\citet[Lemmas A.4 and  A.7]{WZ21}.
 Let $\nu\in \mathcal{M}_1$ be given by $D_\nu=\id_A/(1-\beta)$. 
Since $\id_A$ is independent of $\mathcal G$,
we have $\nu\in \mathcal{E}(\mathbb{P})$, and therefore $\nu\in \mathcal S$, and $\rho(X) \ge \E^{\nu}[X] = \E[X|A]=\ES_\beta(X)$. \textbf{Claim 2:} Take $X_1\in L^\infty(\mathcal{G})$, $X_2\in L^\infty(\mathcal{H})$, $X_3\in L^\infty$ with $X_1\overset{\mathrm{d}}{=}_\mathbb{P}X_2\overset{\mathrm{d}}{=}_\mathbb{P}X_3$. We have  
    \begin{align*}
        \rho(X_3)&= \sup_{\mu\in \mathcal{S}}\mathbb{E}^\mu[X_3]\geq\mathbb{E}[X_3]=\mathbb{E}[X_1]=\rho(X_1); \\
        \rho(X_3)&=\sup_{\mu\in \mathcal{S}}\mathbb{E}^\mu[X_3]\leq \sup_{\mu\in \mathcal{S}_{\beta,\mathbb{P}}}\mathbb{E}^\mu[X_3]=\mathrm{ES}_\beta[X_3]=\mathrm{ES}_\beta[X_2]=\rho(X_2),
    \end{align*}
    as desired.
\Halmos \end{proof}

\begin{proof}{Proof of the claim in Example \ref{example:tail}:}
   This proof follows a very similar line of reasoning as \citet[Proposition 2]{LMWW22}; we reproduce the proof here as there are some nuances in this setting that need to be considered. 
   
   The $\beta$-tail measure $\widehat{\rho}^{\beta}:L^\infty\to\mathbb{R}$ generated by $\widehat{\rho}$ is defined by $$\widehat{\rho}^{\beta}(X)=\widehat{\rho}(F_{\mathbb{P},X}^{-1}\left(\beta+(1-\beta)U)\right),$$ where $U$ is a uniform $(0,1)$ random variable. Since $\widehat{\rho}$ is law-invariant, this mapping is clearly well-defined. For $\mu\in \mathcal{M}_1$, define $\rho^\mu:L^\infty\to\mathbb{R}$ by $\rho^\mu(X)=\gamma(X_{\#}\mu)$, $X\in L^\infty$. It is straightforward to see that for $X,Y\in L^\infty$, $\rho^\mu(X)=\widehat{\rho}(Y)$ if $F_{\mu,X}=F_{\mathbb{P},Y}$. Recall that $\mathcal{S}_{\beta,\mathbb{P}}= \{\nu\in \mathcal{M}_1\mid ({1-\alpha} ) D_\nu\leq   1 \}$.
    By \citet[Proposition 2]{LMWW22}, we have 
    $$\rho(X)\leq\sup_{\mu\in \mathcal{S}_{\beta,\mathbb{P}}}\rho^{\mu}(X)=\widehat{\rho}^\beta(X),$$
    for all $X\in L^\infty$.
    When $X\in L^\infty(\mathcal{H})$, by the fact that $(\Omega,\mathcal{H},\mathbb{P})$ is atomless and Lemma \ref{lem:fel}, we can find a uniform $(0,1)$ random variable $U_X\in L^\infty(\mathcal{H})$ such that $F_{\mathbb{P},X}^{-1}(U_X)=X$ $\mathbb{P}\text{-a.s.}$
    Let $\mu\in \mathcal{M}_1$ have the Radon–Nikodym derivative $$D_\mu=\frac{1}{1-\beta}\id_{\{U_X> \beta\}}.$$
    Since $D_\mu$ is independent of $\mathcal{G}$, $\mathbb{E}[D_\mu|\mathcal{G}]=1$ and $\mu\in \mathcal{S}$. Letting $Z=F_{\mathbb{P},X}^{-1}(\beta+(1-\beta)U_X)$ and $x\in \mathbb{R}$, we have
    \begin{align*}
        \mu(X> x)&=\mu(X> x,U_X> \beta)=\frac{1}{1-\beta}\mathbb{P}(X>x,U_X>\beta)=\min\left(\frac{1}{1-\beta}\mathbb{P}({X>x}),1\right),\\
        F_{\mu,X}(x)&=\left(\frac{1}{1-\beta}(F_{\mathbb{P},X}(x)-\beta)\right)^+=F_{\mathbb{P},Z}(x).
    \end{align*}
    As a consequence,     $$\widehat{\rho}^\beta(X)=\widehat{\rho}(Z)=\rho^{\mu}(X)\leq \rho(X),$$
    and $\rho(X)=\widehat{\rho}^\beta(X).$ Let $X\in L^\infty(\mathcal{G})$ and $\mu\in \mathcal{S}$, since $\mu\in \mathcal{E}(\mathbb{P})$, $F_{\mu,X}=F_{\mathbb{P},X}$. Therefore, $\rho^\mu(X)=\widehat{\rho}(X)$ and $\rho(X)=\widehat{\rho}(X)$.
\Halmos \end{proof}

\subsection{Additional results and proofs accompanying Section \ref{sec:computation}}
\label{app:computation}

\begin{proof}{A computable formula for \eqref{eq:RW-1}:}
Let $\Phi$ (resp.~$\phi$) denote the standard normal cdf (resp.~pdf). For $X\sim \mathrm{N}( r,\sigma^2)$,
    we have
    \begin{align*}
        \mathrm{VaR}_\beta\left((X-x)^+\right)&=\left( r-x+\sigma\Phi^{-1}(\beta)\right)^+,\\
        \mathrm{ES}_\beta\left((X-x)^+\right)&=\begin{cases}
         r-x+\sigma\frac{\phi(\Phi^{-1}(\beta))}{1-\beta}&x\leq r+\sigma \Phi^{-1}(\beta)\\\frac{( r-x)(1-\Phi\left(\frac{x- r}{\sigma}\right))}{1-\beta}+\sigma\frac{\phi(\frac{x- r}{\sigma})}{1-\beta}&x> r+\sigma\Phi^{-1}(\beta).
        \end{cases}
    \end{align*}
For $l\in \mathbb{R}$, the conditional distribution
    $$\pi_1 X_1+\pi_2 X_2|X_1=l\sim N\left(\pi_1l+\pi_2 m_2+c\pi_2\frac{\sigma_2}{\sigma_1}(l- m_1),\pi_2^2(1-c^2)\sigma_2^2\right).$$ 
Define the functions $ M(\pi_1,\pi_2,l)=\pi_1l+\pi_2 m_2+c\pi_2\frac{\sigma_2}{\sigma_1}(l- m_1)$, $\sigma(\pi_2)=\pi_2\sqrt{(1-c^2)}\sigma_2$, and 
\begin{align*}
    g_\beta( r,\sigma,x)=\begin{cases}
         r-x+\sigma\frac{\phi(\Phi^{-1}(\beta))}{1-\beta}&x\leq r+\sigma \Phi^{-1}(\beta)\\\frac{( r-x)(1-\Phi\left(\frac{x- r}{\sigma}\right))}{1-\beta}{}+\sigma\frac{\phi(\frac{x- r}{\sigma})}{1-\beta}&x> r+\sigma\Phi^{-1}(\beta).\end{cases}
\end{align*} 
    We can now verify the following identity  \begin{equation} \mathrm{ES}_\beta((\pi_1 X_1+\pi_2 X_2-x)^+|X_1)=g_\beta( M(\pi_1,\pi_2,X_1),\sigma(\pi_2),x).
    \end{equation}  
Therefore, $$f_\beta(\pi_1,\pi_2,x)=\int_\mathbb{R} g_\beta( M(\pi_1,\pi_2,z),\sigma(\pi_2),x)~ \frac 1 {\sigma_1}\phi\left(\frac{z-m_1}{\sigma_1}\right)~\d z,$$
which is an integral formula for \eqref{eq:RW-1} that is easy to numerically compute.
\Halmos \end{proof}

\section{Analysis in a finite probability space}
We 
will first show that law invariance is equivalent to permutation invariance in a finite uniform probability space. Afterward, we will study partially law-invariant coherent risk measures when our probability space has a natural structure.

\subsection{Law-invariant coherent risk measures
}
In the literature, when proving results regarding law-invariant risk measures, one often assumes that the underlying probability space is atomless. Since finite probability spaces always have atoms, we instead will assume that the underlying probability measure is uniform. 

Let $N\in\mathbb{N}$, denote by $[N]=\{1,\dots,N\}$. Assume that $\Omega=[N]$ and $\mathcal{F}=2^\Omega$ and put the uniform probability measure $\mathbb{P}$ on $(\Omega,\mathcal{F})$, that is $\mathbb{P}(\{k\})=1/N$ for all $k\in [N]$. We can identify elements $X\in L^\infty$ with $\mathbb{R}^N$. For the rest of the section, we will refer only to $\mathbb{R}^N$, but this is to be regarded as $L^\infty$. 

Let $\Delta^N=\left\{\mu\in \mathbb{R}_+^N\hspace{2pt}:\hspace{2pt} \mathbb{E}[\mu]=1\right\}$, it is easy to see  that we can identify $\mathcal{M}_1$ with $\Delta^N$. Since we are working in $\mathbb{R}^N$, the $w\left(\mathbb{R}^N,\mathbb{R}^N\right)$ topology ($w\left(L^\infty,\mathcal{M}\right)$ and $w\left(\mathcal{M},L^\infty\right)$) is the standard euclidean topology. Therefore, there will be no mention in this section as to which topology we are dealing with. Note that any closed set $\mathcal{R}\subseteq \Delta^N$ is compact. Finally, Fatou continuity is implied by any coherent risk measure $\rho:\mathbb{R}^N\to\mathbb{R}$. Thus, there will be no mention of continuity in this section.

To describe law invariance in this space define $S_N=\{\sigma:[N]\to[N]\mid\sigma \text{ is a bijection}\}$. For $\sigma\in S_N$ define the map $T_\sigma:\mathbb{R}^N\to\mathbb{R}^N$ where $(T_\sigma(X))_n=X_{\sigma(n)}$.

\begin{definition}
    We say that a mapping $\rho:\mathbb{R}^N\to\mathbb{R}$ is \emph{permutation invariant} if $\rho(T_\sigma(X))=\rho(X)$ for all $\sigma\in S_N$ and $X\in\mathbb{R}^N$. We say a set $\mathcal{R}\subseteq \mathbb{R}^N$ is \emph{permutation invariant} if $T_\sigma(\mathcal{R})\subseteq \mathcal{R}$ for all $\sigma\in S_N$.
\end{definition}

    It is easy to see that a mapping is law invariant if and only if it is permutation invariant. If $\mathcal{R}\subseteq \mathbb{R}^N$ is permutation invariant then $T_\sigma(\mathcal{R})\supseteq T_\sigma(T_{\sigma^{-1}}(\mathcal{R}))=\mathcal{R}$, so we actually have equality.

\begin{proposition}
    Let $\rho:\mathbb{R}^N\to\mathbb{R}$ be a coherent risk measure. Then $\rho$ is law invariant if and only if $\mathcal{S}_\rho$ is permutation invariant.
\end{proposition}

\begin{proof}{Proof:}
    Let $\mathcal{S}_\rho$ be permutation invariant then we have 
    \begin{align*}
        \rho(T_\sigma(X))=\sup_{\mu\in\mathcal{S}_\rho}\mathbb{E}[\mu T_\sigma(X)]&=\frac{1}{N}\sup_{\mu\in\mathcal{S}_\rho}(\mu\cdot T_\sigma(X))\\&=\frac{1}{N}\sup_{\mu\in\mathcal{S}_\rho}(T_{\sigma^{-1}}(\mu)\cdot X)=\frac{1}{N}\sup_{\mu\in\mathcal{S}_\rho}(\mu\cdot X)=\rho(X),
    \end{align*}
    for all $\sigma\in S_n$. Conversely, assume that $\mathcal{S}_\rho$ is not permutation invariant. Find $R\in\mathcal{S}_\rho$ and $\sigma\in S_n$ such that $T_\sigma(R)\notin \mathcal{S}_\rho$. Since $\mathcal{S}_\rho$ is compact and convex, by the Hahn Banach separation theorem, we can find $X\in \mathbb{R}^N$ such that 
    \begin{align*}
        \rho(X)=\frac{1}{N}\sup_{\mu\in\mathcal{S}_\rho}(\mu\cdot X)<\frac{1}{N}X\cdot T_\sigma(R).
    \end{align*}
    Therefore we have 
    \begin{align*}
        \rho(X)<\frac{1}{N}X\cdot T_\sigma(R)=\frac{1}{N} T_{\sigma^{-1}}(X)\cdot R\leq \rho(T_{\sigma^{-1}}(X)),
    \end{align*}
    implying that $\rho$ is not permutation invariant and thus not law invariant.
\Halmos \end{proof}

\subsection{Partially law-invariant coherent risk measures}

We write $A\times B=\{(a,b): a\in A, ~b\in B\}$ for two sets $A$ and $B$.

Let $\Omega=[M]\times [N]$ for some $M,N\in \mathbb{N}$, $\mathcal{F}=2^\Omega$, $\mathcal{G} $ is generated by the projection map onto the second coordinate, and $\mathbb{P}$ is positive (that is, $\mathbb{P}>0$ on $\mathcal F\setminus \{\varnothing\}$) and uniform in the second coordinate (that is, $\mathbb{P}([M]\times \{ n\})=1/N$ for all $n\in [N]$).

We can identify elements $X\in L^\infty$ with $\mathbb{R}^{M\times N}$. For the rest of the section, we will refer only to $\mathbb{R}^{M\times N}$ but this is to be regarded as $L^\infty$. Let $\Delta^{M\times N}=\{\mu\in \mathbb{R}^{M\times N}_+\hspace{2pt}:\hspace{2pt} \mathbb{E}[\mu]=1\}$, it is easy to see that $\mathcal{M}_1$ can be identified with $\Delta^{M\times N}$. The same comments about topology and Fatou continuity from the previous section also apply in this section.

For $k\in [N]$ define $E^k\in \mathbb{R}^{M\times N}$ by $E^k_{mn}=\delta_{kn}$, where $\delta_{kn}$ is the Kronecker delta. Let $\mathcal{E}=\left\{E^k\hspace{2pt}:\hspace{2pt} k\in [N]\right\}$, clearly $\mathcal{E}$ is a basis for $L^\infty(\mathcal{G})$. Define the map $ \mathbb{L}:\mathbb{R}^{M\times N}\to \mathbb{R}^N:X\mapsto [\mathbb{E}[X|\mathcal{G}]]_\mathcal{E}$. Given some $X\in \mathbb{R}^N$ we will define $\tilde{X}\in \mathbb{R}^{M\times N}$ by $\tilde{X}=\sum_{k=1}^NX_kE^k$.
A set $\mathcal{R}\subseteq \mathbb{R}^{M\times N}$ is $\mathcal{G}$-law invariant if $\mathbb{L}(\mathcal{R})$ is permutation invariant.

The following proposition gives necessary and sufficient conditions for the supporting set of an $\mathcal{G}$-law invariant coherent risk measure. 

\begin{proposition}
    \label{pro:1}
    Let $\rho:\mathbb{R}^{M\times N}\to\mathbb{R}$ be a coherent risk measure. Then $\rho$ is $\mathcal{G}$-law invariant if and only if $\mathcal{S}_\rho$ is $\mathcal{G}$-law invariant.
\end{proposition}

\begin{proof}{Proof:}
    Define the mapping 
    \begin{align*}
        \tilde{\rho}:\mathbb{R}^N\to\mathbb{R}:X\mapsto \rho(\tilde{X}).
    \end{align*}
    It is easy to verify that this mapping is a coherent risk measure. Let $\mathcal{S}_{\tilde{\rho}}\subseteq \Delta^N$ be the supporting set for $\tilde{\rho}$. It follows that $\rho$ is $\mathcal{G}$-law invariant if and only if $\tilde{\rho}$ is law invariant. For $X\in\mathbb{R}^N$ we have that
    \begin{align*}
        \tilde{\rho}(X)=\rho(\tilde{X})=\sup_{\mu\in\mathcal{S}_\rho}\mathbb{E}\left[\mu\tilde{X}\right]&=\sup_{\mu\in\mathcal{S}_\rho}\mathbb{E}\left[\mathbb{E}[\mu|\mathcal{G}]\tilde{X}\right]\\&=\frac{1}{N}\sup_{\mu\in\mathcal{S}_\rho}\mathbb{L}(\mu)\cdot X=\frac{1}{N}\sup_{\mu\in \mathbb{L}(\mathcal{S}_\rho)}\mu\cdot X.
    \end{align*}
    Thus $\mathbb{L}(\mathcal{S}_\rho)\subseteq \mathcal{S}_{\tilde{\rho}}$ by Proposition \ref{prop:coherent}. Since $\mathbb{L}(\mathcal{S}_\rho)$ is compact and convex (image of a compact and convex set under a continuous linear map), the same theorem states that $\mathbb{L}(\mathcal{S}_\rho)=\mathcal{S}_{\tilde{\rho}}$. We then have
    \begin{align*}
        \rho\text{ is }\mathcal{G}\text{-law invariant }\Leftrightarrow \tilde{\rho}\text{ is law invariant}&\Leftrightarrow \mathcal{S}_{\tilde{\rho}}\text{ is permutation invariant}\\&\Leftrightarrow \mathbb{L}(\mathcal{S}_\rho) \text{ is permutation invariant}\\&\Leftrightarrow \mathcal{S}_\rho\text{ is }\mathcal{G}\text{-law  invariant}.
    \end{align*}
This completes the proof. 
\Halmos \end{proof}

\subsection{Representation results}

This section provides representation results similar to Theorem \ref{th:friday} and Theorem \ref{theorem:1}. A natural corollary to Proposition \ref{pro:1} is the following Theorem.

\begin{theorem}
    The mapping $\rho:\mathbb{R}^{M\times N}\to\mathbb{R}$ is a $\mathcal{G}$-law-invariant coherent risk measure if and only if
    \begin{align*}
        \rho(X)=\sup_{\mu\in\mathcal{S}}\mathbb{E}[\mu X],~~X\in L^\infty
    \end{align*}
    holds for some closed and convex $\mathcal{S}\subseteq \mathcal{M}_1$ such that $\mathbb{L}(\mathcal{S})$ is $\mathcal{G}$-law invariant.
\end{theorem}

Define the set $\Delta_\mathcal{G}^{M\times N}=\Delta^{M\times N}\cap L^\infty(\mathcal{G})$, which can be identified with $\mathcal{M}_1(\mathcal{G})$.

\begin{theorem}
    The mapping $\rho:\mathbb{R}^{M\times N}\to\mathbb{R}$ is a coherent risk measure if and only if there exists a unique compact and convex set $\mathcal{R}\subseteq \Delta_\mathcal{G}^{M\times N}$ and a CA-assignment $(\tau_\mu)_{\mu\in \mathcal{R}}$ such that
        \begin{align}
            \rho(X)=\sup_{\mu\in \mathcal{R}}\left(\mathbb{E}\left[\mu X\right]+\tau_\mu(X)\right),\hspace{3pt}X\in L^\infty.\label{eq:discrete}
        \end{align}
    Moreover, in \eqref{eq:discrete}, $\rho$ is $\mathcal{G}$-law invariant if and only if $\mathcal{R}$ is $\mathcal{G}$-law invariant.
\end{theorem}

\begin{proof}{Proof:}
    Let $\rho$ be a coherent risk measure. Define the set $ \mathcal{R}=\left\{\mathbb{E}[\mu|\mathcal{G}]:\mu\in \mathcal{S}_\rho\right\}\subseteq \mathcal{R}\subseteq \Delta_\mathcal{G}^{M\times N}$.  The set $\mathcal{R}$ is compact and convex since it is the image of a compact convex set under a (continuous) linear map. For $\mu\in\mathcal{R}$, define $\mathcal{S}(\mu)=\left\{\nu\in \mathcal{S}_\rho: \mathbb{E}[\nu|\mathcal{G}]=\mu\right\}$. It is clear that $\mathcal{S}_\rho=\bigcup_{\mu\in \mathcal{R}}\mathcal{S}(\mu)$.
    We have 
    \begin{align*}
        \rho(X)=\sup_{\nu\in \mathcal{S}_\rho}\mathbb{E}[\mu X]=\sup_{\mu\in \mathcal{R}}\sup_{\nu\in \mathcal{S}(\mu)}\mathbb{E}[\nu X]&=\sup_{\mu\in \mathcal{R}}\sup_{\nu\in \mathcal{S}(\mu)}(\mathbb{E}[\mu X]+\mathbb{E}[(\nu-\mu)X])\\&=\sup_{\mu\in \mathcal{R}}\left(\mathbb{E}[\mu X]+\tau_\mu(X)\right)
    \end{align*}
    for all $X\in L^\infty$,
    where $\tau_\mu(X)=\sup_{\nu\in \mathcal{S}(\mu)}\mathbb{E}[(\nu-\mu)X]$. The fact that $\tau_\mu\in \mathrm{CA}(\mu)$ follows a similar argument as the proof of Proposition \ref{prop:FatouCont}. Since $\mathcal{R}$ supports $\rho|_{L^\infty(\mathcal{G})}$ and is closed and convex, it is unique from Proposition \ref{prop:coherent}. For the converse, note that ($\ref{eq:32}$) is the supremum of coherent risk measures (by Proposition \ref{prop:motivation}; see \eqref{eq:31}) and thus a coherent risk measure. 
    
    To prove the last claim, let $\tilde{\rho}:\mathbb{R}^N\to\mathbb{R}:X\mapsto \rho(\tilde{X})$ and $\mathcal{S}_{\tilde{\rho}}\subseteq \Delta^{N}$. We have
    \begin{align*}
        \tilde{\rho}(X)=\rho(\tilde{X})=\sup_{\mu\in \mathcal{R}}\mathbb{E}[\mu\tilde{X}]=\frac{1}{N}\sup_{\mu\in \mathbb{L}(\mathcal{R})}\mu\cdot X
    \end{align*}
    for all $X\in \mathbb{R}^N$. Therefore $\mathbb{L}(\mathcal{R})=\mathcal{S}_{\tilde{\rho}}$, and 
    \begin{align*}
        \rho\text{ is }\mathcal{G}\text{-law invariant}\Leftrightarrow \tilde{\rho}\text{ is law invariant}\Leftrightarrow \mathcal{S}_{\tilde{\rho}}\text{ is permutation invariant}\Leftrightarrow \mathcal{R}\text{ is }\mathcal{G}\text{-law invariant}.
    \end{align*}
This completes the proof.
\Halmos \end{proof}

\end{document}